\documentclass[journal]{IEEEtran}
\usepackage[utf8]{inputenc}
\usepackage{amsthm}
\usepackage[pdfencoding=auto,bookmarks=false,colorlinks=true]{hyperref}
\usepackage[style=ieee,mincitenames=1,maxcitenames=1,maxbibnames=9,backend=biber]{biblatex}
\usepackage{enumitem}
\usepackage{booktabs}
\usepackage{mdl}
\usepackage{autonum}
\usepackage{autobreak}
\usepackage{todonotes}

\setlist{noitemsep,leftmargin=*}
\DeclareSourcemap{ 
    \maps{
        \map{ 
            \step[fieldsource=url,
                  match=\regexp{\{\\\_\}|\{\_\}|\\\_},
                  replace=\regexp{\_}]
        }
        \map{ 
            \step[fieldsource=url,
                  match=\regexp{\{\$\\sim\$\}|\{\~\}|\$\\sim\$},
                  replace=\regexp{\~}]
        }
        \map{ 
            \step[fieldsource=doi,
                  match=\regexp{\{\\\_\}|\{\_\}|\\\_},
                  replace=\regexp{\_}]
        }
        \map{ 
            \step[fieldsource=doi,
                  match=\regexp{\{\$\\sim\$\}|\{\~\}|\$\\sim\$},
                  replace=\regexp{\~}]
        }
    }
}

\begin{document}
%
\title{Minimax Optimal Additive Functional Estimation with Discrete Distribution}
%
%
%

\author{Kazuto Fukuchi and
        Jun Sakuma
\thanks{K. Fukuchi is with RIKEN Advanced Intelligence Project, Nihonbashi 1-chome Mitsui Building, 15th floor,1-4-1 Nihonbashi, Chuo-ku, Tokyo 103-0027, Japan, e-mail: kazuto.fukuchi@riken.jp.}
\thanks{J. Sakuma is with the Department of Computer Science, Graduate School of System and Information Engineering, University of Tsukuba, 1-1-1 Tennodai, Tsukuba, Ibaraki 305-8577 Japan, e-mail: jun@cs.tsukuba.ac.jp. He is also with RIKEN Advanced Intelligence Project and with JST CREST, K’s Gobancho 6F, 7, Gobancho, Chiyoda-ku, Tokyo, 102-0076 Japan.}
\thanks{This paper was presented in part at the 2017 IEEE International Symposium on Information Theory (ISIT), Aachen, Germany~\autocite{Fukuchi2017MinimaxDistributionsb} and 2018 IEEE International Symposium on Information Theory (ISIT), Vail, USA~\autocite{Fukuchi2018MinimaxSpeed}.}
}

\maketitle

\begin{abstract}
This paper addresses a problem of estimating an {\em additive functional} given $n$ i.i.d. samples drawn from a discrete distribution $P=(p_1,...,p_k)$ with alphabet size $k$. The additive functional is defined as $\theta(P;\phi)=\sum_{i=1}^k\phi(p_i)$ for a function $\phi$, which covers the most of the entropy-like criteria. The minimax optimal risk of this problem has been already known for some specific $\phi$, such as $\phi(p)=p^\alpha$ and $\phi(p)=-p\ln p$. However, there is no generic methodology to derive the minimax optimal risk for the additive function estimation problem. In this paper, we reveal the property of $\phi$ that characterizes the minimax optimal risk of the additive functional estimation problem; this analysis is applicable to general $\phi$. More precisely, we reveal that the minimax optimal risk of this problem is characterized by the {\em divergence speed} of the function $\phi$. 
\end{abstract}


\section{Introduction}

Let $P$ be a probability measure with alphabet size $k$, where we use a vector representation of $P$; i.e., $P=(p_1,...,p_k)$ and $p_i=P\cbrace{i}$. Let $\phi$ be a mapping from $[0,1]$ to $\RealSet$. Given a set of i.i.d. samples $S_n=\cbrace{X_1,...,X_n} \sim P^n$, we deal with a problem of estimating an {\em additive functional} of $\phi$. The additive functional $\theta$ of $\phi$ is defined as
\begin{align}
 \theta(P;\phi) = \sum_{i=1}^k \phi(p_i).
\end{align}
We simplify this notation to $\theta(P;\phi)=\theta(P)$. Most entropy-like criteria can be formed in terms of $\theta$. For instance, when $\phi(p)=-p\ln p$, $\theta$ is Shannon entropy. For a positive real $\alpha$, letting $\phi(p)=p^\alpha$, $\ln(\theta(P))/(1-\alpha)$ becomes R\'enyi entropy. More generally, letting $\phi = f$ where $f$ is a concave function, $\theta$ becomes $f$-entropies~\autocite{Akaike1998InformationPrinciple}. The estimation problem of such entropy-like criteria is a basic but important component for various research areas, such as physics~\autocite{Lake2011AccurateDevices}, neuroscience~\autocite{Nemenman2004EntropyProblem}, security~\autocite{Gu2005DetectingEstimation}, and machine learning~\autocite{Quinlan1986InductionTrees,Peng2005FeatureMin-redundancy}.

The goal of this study is to construct the minimax optimal estimator of $\theta$ given a function $\phi$. To precisely define the minimax optimality, we introduce the (quadratic) minimax risk. A sufficient statistic of $P$ is a histogram $N=\paren{N_1,...,N_k}$, where letting $\ind{\cbrace*{\cdot}}$ be the indicator function, $N_j = \sum_{i=1}^n \ind{\cbrace{X_i = j}}$ and $N \sim \Mul(n, P)$. The estimator of $\theta$ can thus be defined as a function $\hat\theta:[n]^k\to\RealSet$, where $[m]=\cbrace{1,...,m}$ for an integer $m$. The quadratic minimax risk is defined as
\begin{align}
 R^*(n,k;\phi) = \inf_{\hat\theta}\sup_{P \in \dom{M}_k} \Mean\bracket*{\paren*{\hat\theta(N) - \theta(P)}^2}, \label{eq:minimax-risk}
\end{align}
where $\dom{M}_k$ is the set of all probability measures on $[k]$, and the infimum is taken over all estimators $\hat\theta$. With this definition, an estimator $\hat\theta$ is minimax \mbox{(rate-)optimal} if there is a constant $C>0$ such that
\begin{align}
 \sup_{P \in \dom{M}_k} \Mean\bracket*{\paren*{\hat\theta(N) - \theta(P)}^2} \le C R^*(n,k;\phi).
\end{align}
Since no estimator achieves smaller worst case risk than the minimax risk, we can say that the minimax optimal estimator is the best regarding the worst case risk.

\noindent{\bfseries Notations.}
We now introduce some additional notations. For any positive real sequences $\cbrace{a_n}$ and $\cbrace{b_n}$, $a_n \gtrsim b_n$ denotes that there exists a positive constant $c$ such that $a_n \ge c b_n$. Similarly, $a_n \lesssim b_n$ denotes that there exists a positive constant $c$ such that $a_n \le c b_n$. Furthermore, $a_n \asymp b_n$ implies $a_n \gtrsim b_n$ and $a_n \lesssim b_n$. For an event $\event$, we denote its complement as $\event^c$. For two real numbers $a$ and $b$, $a \lor b = \max\cbrace{a, b}$ and $a \land b = \min\cbrace{a, b}$.

\subsection{Related Work}

Many researchers have been dealing with the estimation problem of the additive functional and provides many estimators and analyses in decades past. The plugin estimator or the maximum likelihood estimator~(MLE) is the simplest way to estimate the additive functional $\theta$, in which the empirical probabilities $\tilde{P} = (\hat{p}_1,...,\hat{p}_k)$ are substituted into $\theta$ as $\theta(\tilde{P})$. The plugin estimator is asymptotically consistent~\autocite{Antos2001ConvergenceDistributions}, asymptotically efficient and minimax optimal~\autocite{Vaart1998AsymptoticStatistics} under weak assumptions for fixed $k$. However, this is not true for the large-$k$ regime. Indeed, \textcite{Jiao2015MinimaxDistributions} and \textcite{Wu2016MinimaxApproximation} derived a lower bound for the quadratic risk for the plugin estimator of $\phi(p)=-p\ln(p)$ and $\phi(p)=p^\alpha$. In the case of Shannon entropy, the lower bound is given by \textcite{Jiao2015MinimaxDistributions} as
\begin{align}
  \sup_{P \in \dom{M}_k}\Mean\bracket*{\paren*{\hat\theta_{\mathrm{plugin}}(N) - \theta(P)}^2} \gtrsim \frac{k^2}{n^2} + \frac{\ln^2 k}{n}. \label{eq:lower-plugin}
\end{align}
The first term $k^2/n^2$ comes from the bias and indicates that if $k$ grows linearly with respect to $n$, the plugin estimator becomes inconsistent. Bias-correction methods, such as \autocite{Miller1955NoteEstimates,Grassberger1988FiniteEstimates,Zahl1977JackknifingDiversity}, can be applied to the plugin estimator of $\phi(p)=-p\ln p$ to reduce the bias whereas these bias-corrected estimators are still inconsistent if $k$ is larger than $n$. The estimators based on Bayesian approaches in \autocite{Schurmann1996EntropySequences,Schober2013SomeDistributions,Holste1998BayesEntropies} are also inconsistent for $k \gtrsim n$~\autocite{Han2015DoesProblem}.

\textcite{Paninski2004EstimatingSamples} firstly revealed existence of a consistent estimator even if the alphabet size $k$ is larger than linear order of the sample size $n$. However, they did not provide a concrete form of the consistent estimator. The first estimator that achieves consistency in the large-$k$ regime is proposed by \textcite{Valiant2011EstimatingCLTs}. However, the estimator of \autocite{Valiant2011EstimatingCLTs} has not been shown to achieve the minimax rate even in a more detailed analysis in \autocite{Valiant2011TheEstimators}.

Recently, many researchers investigated the minimax optimal risk for the additive functionals in the large-$k$ regime for some specific $\phi$. \textcite{Acharya2015TheEntropy} showed that the bias-corrected estimator of R\'enyi entropy achieves the minimax optimal risk in regard to the sample complexity if $\alpha > 1$ and $\alpha \in \NaturalSet$, but they did not show the minimax optimality for other $\alpha$. \textcite{Jiao2015MinimaxDistributions} introduced a minimax optimal estimator for $\phi(p)=p^\alpha$ for any $\alpha \in (0,3/2)$ in the large-$k$ regime. \textcite{Wu2015ChebyshevUnseen} derived a minimax optimal estimator for $\phi(p)=\ind{p > 0}$. For $\phi(p)=-p\ln p$, \textcite{Jiao2015MinimaxDistributions,Wu2016MinimaxApproximation} independently introduced the minimax optimal estimator in the large-$k$ regime.
\cref{tbl:existing-summary} shows the summary of the existing minimax optimal risks for the additive functional estimation with some specific $\phi$. The first column shows the target function $\phi$, and the second column denotes the parameter appeared in $\phi$. The third column shows the minimax optimal risk corresponding to $\phi$, where these rates only proved when the condition shown in the fourth column is satisfied. In the case of Shannon entropy, the optimal risk was obtained as
\begin{align}
  R^*(n,k;\phi) \asymp \frac{k^2}{(n\ln n)^2} + \frac{\ln^2 k}{n}.
\end{align}
The first term is improved from \cref{eq:lower-plugin}. It indicates that the introduced estimator can consistently estimate Shannon entropy even when $k \gtrsim n$, as long as $n \gtrsim k/\ln k$.

\begin{table}[tb]
  \centering
  \caption{Summary of existing results.}\label{tbl:existing-summary}
  \begin{tabular}{ll|lll}
    \toprule
     $\phi$ & $\alpha$ & minimax risk & condition &  \\
    \midrule
    $p^\alpha$ & $(0,1/2]$ & $\frac{k^2}{(n\ln n)^{2\alpha}}$ & $\ln k \gtrsim \ln n$ & \autocite{Jiao2015MinimaxDistributions} \\
    $p^\alpha$ & $(1/2,1)$ & $\frac{k^2}{(n\ln n)^{2\alpha}}+\frac{k^{2-2\alpha}}{n}$ & - & \autocite{Jiao2015MinimaxDistributions} \\
    $-p\ln p$  & - & $\frac{k^2}{(n\ln n)^{2\alpha}}+\frac{\ln^2k}{n}$ & - & \autocite{Jiao2015MinimaxDistributions,Wu2016MinimaxApproximation} \\
    $p^\alpha$ & $(1,3/2)$ & $\frac{1}{(n\ln n)^{2\alpha-2}}$ & $k \asymp n\ln n$ & \autocite{Jiao2015MinimaxDistributions} \\
    $p^\alpha$ & $\ge 3/2$ & $\frac{1}{n}$ & - & \autocite{Jiao2017MaximumDistributions} \\
    $\ind{p > 0}$ & - & $k^2e^{-\paren{\frac{\sqrt{n\ln k}}{k}\lor\frac{n}{k}\lor 1}}$ & - & \autocite{Wu2015ChebyshevUnseen} \\
    \bottomrule
  \end{tabular}
\end{table}

While the recent efforts revealed the minimax optimal estimators for the additive functionals with some specific $\phi$, there is no unified methodology to derive the minimax optimal estimator for the additive functional with general $\phi$. \textcite{Jiao2015MinimaxDistributions} suggested that their proposed estimator can be extended for general additive functional $\theta$. However, the minimax optimality of the estimator was only proved for specific cases of $\phi$, including $\phi(p)=-p\ln p$ and $\phi(p)=p^\alpha$. To prove the minimax optimality for other $\phi$, we need to individually analyze the minimax optimality for specific $\phi$. The aim of the present paper is to clarify which property of $\phi$ substantially influences the minimax optimal risk when estimating the additive functional.

The optimal estimators for divergences with large alphabet size have been investigated in \autocite{Bu2016EstimationDistributions,Han2016MinimaxDistributions,Jiao2016MinimaxDistance,Acharya2018ProfileDivergence}. The estimation problems of divergences are much complicated than the additive function, while the similar techniques were applied to derive the minimax optimality.

\subsection{Contributions}

In this paper, we investigate the minimax optimal risk of the additive functional estimation and reveal a substantial property of $\phi$ that characterizes the minimax optimal risk. More precisely, we show that the divergence speed of $\phi$, which is defined as below, characterizes the minimax optimal risk of the additive functional estimation.
\begin{definition}[Divergence speed]\label{def:div-speed}
  For a positive integer $\ell$ and $\alpha \in \RealSet$, the $\ell$th divergence speed of $\phi \in C^\ell[0,1]$ is $p^{\alpha}$ if there exist constants $W_\ell > 0$, $c_\ell \ge 0$ and $c'_\ell \ge 0$ such that for all $p \in (0,1)$,
  \begin{align}
    W_\ell p^{-\ell+\alpha} - c'_\ell \le \abs*{\phi^{(\ell)}(p)} \le W_\ell p^{-\ell+\alpha} + c_\ell.
  \end{align}
\end{definition}
The divergence speed is faster if $\alpha$ is larger. Informally, the meaning of ``the $\ell$th divergence speed of a function $f(p)$ is $p^{\alpha}$'' is that $\abs*{f^{(\ell)}(p)}$ goes to infinity at the same speed of the $\ell$th derivative of $p^{\alpha}$ when $p$ approaches $0$. In \cref{tbl:existing-summary}, the divergence speed of $\phi(p)=p^\alpha$ for non-integer $\alpha$ is $p^\alpha$ for any $\ell > \alpha$. Also, the divergence speed of $\phi(p)=-p\ln p$ is $p^1$ for any $\ell > 1$. 

\begin{table*}[tb]
  \centering
  \caption{Summary of results.}\label{tbl:results-summary}
  \begin{tabular}{ll|lll}
    \toprule
     $\alpha$ & $\ell$ & minimax risk & condition & estimator \\
    \midrule
     $\le 0$ & $1$ & no consistent estimator & - & \\
     $(0,1/2]$ & $4$ & $\frac{k^2}{(n\ln n)^{2\alpha}}$ & $k \gtrsim \ln^{4} n$ & best poly. \& 2nd-order bias-correction \\
     $(1/2,1)$ & $4$ & $\frac{k^2}{(n\ln n)^{2\alpha}}+\frac{k^{2-2\alpha}}{n}$ & - & best poly. \& 2nd-order bias-correction \\
     $1$ & $4$ & $\frac{k^2}{(n\ln n)^{2}}+\frac{\ln^2k}{n}$ & - & best poly. \& 2nd-order bias-correction \\
     $(1,3/2)$ & $6$ & $\frac{k^2}{(n\ln n)^{2\alpha}}+\frac{1}{n}$ & - & best poly. \& 4th-order bias-correction \\
     $[3/2,2]$ & $2$ & $\frac{1}{n}$ & - & plugin \\
    \bottomrule
  \end{tabular}
\end{table*}

The results are summarized in \cref{tbl:results-summary}. This table shows the minimax optimal risk~(third column), the condition to prove the minimax optimal risk~(fourth column), and the estimator that achieves the optimal risk~(fifth column) for each range of $\alpha$. The column $\ell$~(second column) means that the presented minimax optimality is valid if the $\ell$th divergence speed of $\phi$ is $p^\alpha$. As we can see from \cref{tbl:results-summary}, the minimax risk are affected only by $\alpha$. Thus, we success to characterize the minimax optimal risk only by the property of $\phi$, i.e., $\alpha$ in the divergence speed, without specifying $\phi$. In general, the convergence speed of the minimax optimal risk becomes faster as $\alpha$ increases but is saturated for $\alpha \ge 3/2$. 

As shown in \cref{tbl:results-summary}, the behaviour of the minimax optimal risk is changed by the ranges of $\alpha$. If $\alpha \le 0$, $\phi$ is an unbounded function. Thus, we trivially show that there is no consistent estimator of $\theta$. In other words, the minimax optimal rate is larger than constant order if $\alpha \le 0$. This means that there is no reasonable estimator if $\alpha \le 0$, and thus there is no need to derive the minimax optimal estimator for this case. 

For $\alpha \in (1/2,1)$ and $\alpha \in [3/2,2]$, \textcite{Jiao2015MinimaxDistributions} showed the same minimax optimal risk for $\phi(p)=p^\alpha$. Besides, \textcite{Jiao2015MinimaxDistributions,Wu2016MinimaxApproximation} proved the same minimax optimal risk for $\phi(p)=-p\ln p$. Our result is a generalized version of their results such that it is applicable to the general $\phi$ including $\phi(p)=p^\alpha$ and $\phi(p)=-p\ln p$.

For $\alpha  \in (0,1/2]$, we show that the minimax optimal risk is 
\begin{align}
 \frac{k^2}{(n\ln n)^{2\alpha}}, \label{eq:optimal-rate-0-1/2}
\end{align}
where we assume $k \gtrsim \ln^4 n$ to prove the above rate. As an existing result, \textcite{Jiao2015MinimaxDistributions} proved the same minimax optimal risk for $\phi(p)=p^\alpha$, where their proof requires an assumption $\ln k \gtrsim \ln n$~(first row in \cref{tbl:existing-summary}). Their assumption is stronger than the assumption we assumed, i.e., $k \gtrsim \ln^4 n$. In this sense, we provide clearer understanding of the additive functional estimation problem for this case.

For $\alpha \in (1,3/2)$, we prove the following minimax optimal risk
\begin{align}
  \frac{k^2}{(n\ln n)^{2\alpha}} + \frac{1}{n}. \label{eq:optimal-rate-1-3/2-cont}
\end{align}
As an existing result for this range, \textcite{Jiao2015MinimaxDistributions} investigated the minimax optimal risk for $\phi(p) = p^\alpha$. As shown in \cref{tbl:existing-summary}~(fourth row), they showed that the minimax optimal rate for $\phi(p) = p^\alpha$ is $1/(n\ln n)^{2\alpha-2}$. However, their analysis requires the strong condition $k \asymp n\ln n$, and minimax optimality for $\alpha \in (1,3/2)$ is, therefore, far from clear understanding. In contrast, we success to prove \cref{eq:optimal-rate-1-3/2-cont} without any condition on the relationship between $k$ and $n$. As a result, we clarify the number of samples that is necessary to estimate $\theta$ consistently.

\section{Preliminaries}\label{sec:preliminaries}

\subsection{Poisson Sampling}
We employ the Poisson sampling technique to derive upper and lower bounds for the minimax risk. The Poisson sampling technique models the samples as independent Poisson distributions, while the original samples follow a multinomial distribution. Specifically, the sufficient statistic for $P$ in the Poisson sampling is a histogram $\tilde{N} = \paren{\tilde{N}_i,...,\tilde{N}_k}$, where $\tilde{N}_1,...,\tilde{N}_k$ are independent random variables such that $\tilde{N}_i \sim \Poi(np_i)$. The minimax risk for Poisson sampling is defined as follows:
\begin{align}
 \tilde{R}^*(n,k;\phi) = \inf_{\cbrace{\hat\theta}}\sup_{P \in \dom{M}_k} \Mean\bracket*{\paren*{\hat\theta(\tilde{N}) - \theta(P)}^2}.
\end{align}
The minimax risk of the Poisson sampling well approximates that of the multinomial distribution. Indeed, \textcite{Jiao2015MinimaxDistributions} presented the following lemma.
\begin{lemma}[\textcite{Jiao2015MinimaxDistributions}]\label{lem:well-approx-poisson}
  The minimax risk under the Poisson model and the multinomial model are related via the following inequalities:
  \begin{multline}
    \tilde{R}^*(2n,k;\phi) - \sup_{P \in \dom{M}_k}\abs*{\theta(P)}e^{-n/4} \\ \le R^*(n,k;\phi) \le 2\tilde{R}^*(n/2,k;\phi).
  \end{multline}
\end{lemma}
\cref{lem:well-approx-poisson} states $R^*(n,k;\phi) \asymp \tilde{R}^*(n,k;\phi)$, and thus we can derive the minimax rate of the multinomial distribution from that of the Poisson sampling.

\subsection{Polynomial Approximation}
\textcite{Cai2011TestingFunctional} presented a technique of the best polynomial approximation for deriving the minimax optimal estimators and their lower bounds for the risk. Besides, \textcite{Jiao2017MaximumDistributions} used the Bernstein polynomial approximation to derive the upper bound on the estimation error of the plugin estimator. We use such polynomial approximation techniques to derive the upper and the lower bound on the minimax optimal risk for the additive functional estimation.

The key to characterize these approximations is the {\em (weighted) modulus of smoothness}. For an interval $I \subseteq \RealSet$, let us define the $L$th {\em finite difference} of a real-valued scalar function $\phi$ at point $x \in I$ as
\begin{align}
  (\Delta^L_h \phi)(x) =& \sum_{m=0}^L(-1)^{L-m}\binom{L}{m}\phi(x + (L/2-m)h), \label{eq:finite-diff}
\end{align}
if $x - hL/2 \in I$ and $x + hL/2 \in I$, and otherwise $(\Delta^L_h \phi)(x) = 0$. The modulus of smoothness of $\phi$ on an interval $I$ is defined as
\begin{align}
    \omega^L(\phi,t;I) = \sup_{h \in (0,t]}\sup_{x \in I}\abs*{\paren*{\Delta^L_{h}\phi}(x)}.
\end{align}
More generally, we can define the {\em weighted modulus of smoothness} by introducing a weight function $\varphi$. The weighted modulus of smoothness of $\phi$ on an interval $I$ is defined as
\begin{align}
    \omega^L_\varphi(\phi,t;I) = \sup_{h \in (0,t]}\sup_{x \in I}\abs*{\paren*{\Delta^L_{h\varphi}\phi}(x)}.
\end{align}
Note that $\omega^L_\varphi(\phi,t;I) = \omega^L(\phi,t;I)$ for $\varphi(x)=1$. $\omega^1$ is also known as the {\em modulus of continuity}.

We introduce a useful property of the modulus of smoothness, which will be used in later analyses:
\begin{lemma}[\textcite{DeVore1993ConstructiveApproximation}]\label{lem:mod-derivative}
  For a positive integer $r$ and any $k$ times continuously differentiable function $f:[-1,1]\to\RealSet$ where $k < r$, there exists a constant $C > 0$ only depending on $r$ such that
  \begin{align}
    \omega^r(f,t;[-1,1]) \le Ct^k\omega^{r-k}(f^{(k)},t;(-1,1)).
  \end{align}
\end{lemma}

For later analyses, we use two kinds of polynomial approximations; {\em Bernstein polynomial approximation} and {\em best polynomial approximation}.

{\bfseries Bernstein Polynomial Approximation}
A {\em Bernstein polynomial} is a linear combination of Bernstein basis polynomials, which is defined as
\begin{align}
    b_{\nu,L}(x) = \binom{L}{\nu}x^\nu(1-x)^{L-\nu} \for \nu=0,...,L.
\end{align}
Given a function $\phi:[0,1]\to\RealSet$, the polynomial obtained by the Bernstein polynomial approximation with degree-$L$ is defined as
\begin{align}
    B_L[\phi] = \sum_{\nu=0}^L\phi\paren*{\frac{\nu}{L}}b_{\nu,L}.
\end{align}
If $\phi$ is continuous on $[0,1]$, the Bernstein polynomial converges to $\phi$ as $L$ tends to $\infty$.

\textcite{Ditzian1994DirectPolynomials} provided an upper bound of the pointwise error on the Bernstein polynomial approximation by using the second-order modulus of smoothness:
\begin{theorem}[A special case of \textcite{Ditzian1994DirectPolynomials}]\label{thm:bernstein-error}
  Given a function $\phi:[0,1]\to\RealSet$, for an arbitrary $x \in [0,1]$, we have
  \begin{align}
      \abs*{B_L[\phi](x) - \phi(x)} \lesssim \omega^2\paren*{\phi, \sqrt{\frac{x(1-x)}{L}}}.
  \end{align}
\end{theorem}

{\bfseries Best Polynomial Approximation}
Let $\dom{P}_L$ be the set of polynomials of which degree is up to $L$. Given a polynomial $g$ and a function $\phi$ defined on an interval $I \subseteq [0,1]$, the $L_\infty$ error between $\phi$ and $g$ is defined as
\begin{align}
  \sup_{x \in I}\abs*{\phi(x) - p(x)}.
\end{align}
The best polynomial of $\phi$ with a degree-$L$ polynomial is a polynomial $g \in \dom{P}_L$ that minimizes the $L_\infty$ error.  Such a polynomial uniquely exists if $\phi$ is continuous and can be obtained, for instance, by the Remez algorithm~\autocite{Remez1934SurDonnee} if $I$ is bounded.

The error of the best polynomial approximation is defined as
\begin{align}
 E_L\paren*{\phi, I} = \inf_{g \in \dom{P}_L}\sup_{x \in I}\abs*{\phi(x) - g(x)}.
\end{align}
This error is non-increasing as $L$ increases because $\dom{P}_L$ covers all the smaller degree polynomials, i.e., $\dom{P}_0 \subset \dom{P}_1 \subset \dom{P}_2 \subset ... \subset \dom{P}_L$. Decreasing rate of this error with respect to the degree $L$ has been studied well since the 1960s~\autocite{Timan1963TheoryVariable,Petrushev1988RationalFunctions,Ditzian2012ModuliSmoothness,Achieser2004TheoryApproximation}.

\textcite{Ditzian2012ModuliSmoothness} revealed that if $I = [-1,1]$, the weighted modulus of smoothness with $\varphi(x)=\sqrt{1-x^2}$ characterizes the best polynomial approximation error regarding $L$. They showed that the following direct and converse inequalities:
\begin{lemma}[direct result~\autocite{Ditzian2012ModuliSmoothness}]\label{lem:best-modulus-direct}
  For a continuous function $\phi$ defined on $[-1,1]$, we have for $L \ge r$,
  \begin{align}
      E_L(\phi, [-1,1]) \lesssim \omega^r_\varphi(\phi,L^{-1}).
  \end{align}
\end{lemma}
\begin{lemma}[converse result~\autocite{Ditzian2012ModuliSmoothness}]\label{lem:best-modulus-converse}
  For a continuous function $\phi$ defined on $[-1,1]$, we have
  \begin{align}
      \omega^r_\varphi(\phi,L^{-1};[-1,1]) \lesssim L^{-r}\sum_{m=0}^L(m+1)^{r-1}E_m\paren*{\phi,[-1,1]}.
  \end{align}
\end{lemma}
As a consequence of these lemmas, the best polynomial approximation error is characterized by the weighted modulus of smoothness as follows:
\begin{theorem}[\textcite{Ditzian2012ModuliSmoothness}]\label{thm:modulus-best-approx}
  Let $\phi$ be a continuous real-valued function on $[-1,1]$. Then, for $\beta \in (0,L)$,
  \begin{align}
      E_L\paren*{\phi, [-1,1]} \asymp L^{-\beta} \textand \omega^L_\varphi(\phi,t;[-1,1]) \asymp t^{-\beta}
  \end{align}
  are equivalent, where $\varphi(x) = \sqrt{1-x^2}$.
\end{theorem}
From \cref{thm:modulus-best-approx}, we can obtain the best polynomial approximation error rate regarding $L$ by analyzing the weighted modulus of smoothness.

\section{Main Results}

Our main results reveal the minimax optimal risk of the additive functional estimation in characterizing with the divergence speed defined in \cref{def:div-speed}. The behaviour of the minimax optimal risk varies depending on the range of $\alpha$ as $\alpha < 0$, $\alpha \in (0,1/2]$, $\alpha \in (1/2,1)$, $\alpha = 1$, $\alpha \in (1,3/2)$, $\alpha \ge 3/2$. We will show the minimax optimal risk for each range of $\alpha$ one by one.

First, we demonstrate that we cannot construct a consistent estimator for $\alpha < 0$.
\begin{proposition}\label{prop:no-estimator}
  Suppose $\phi:[0,1]\to\RealSet$ is a function whose first divergence speed is $p^\alpha$ for $\alpha < 0$. Then, there is no consistent estimator, i.e., $R^*(n,k;\phi) \gtrsim 1$.
\end{proposition}
The proof of \cref{prop:no-estimator} is given in \cref{sec:no-estimator}. The consistency is a necessary property for a reasonable estimator, and thus this proposition show that there is no reasonable estimator if $\alpha > 0$. For this reason, there is no need to derive the minimax optimal estimator in this case.

For $\alpha \in (0,1/2]$, the minimax optimal risk is obtained as follows.
\begin{theorem}\label{thm:optimal-rate-0-1/2}
  Suppose $\phi:[0,1]\to\RealSet$ is a function whose fourth divergence speed is $p^{\alpha}$ for $\alpha \in (0,1/2]$. If $n \gtrsim k^{1/\alpha}/\ln k$ and $k \gtrsim \ln^{4}n$, the minimax optimal risk is obtained as
  \begin{align}
    R^*(n,k;\phi) \asymp \frac{k^2}{(n\ln n)^{2\alpha}}.
  \end{align}
  If $n \lesssim k^{1/\alpha}/\ln k$, there is no consistent estimator.
\end{theorem}
For this range, \textcite{Jiao2015MinimaxDistributions} showed the same minimax optimal risk for the specific function $\phi(p) = p^{\alpha}$. However, their analysis needs stronger condition, $\ln k \gtrsim \ln n$, to prove $\frac{k^2}{(n\ln n)^{2\alpha}}$ rate. We success to prove this rate with weaker condition $k \gtrsim \ln^{4}n$. However, for $k \lesssim \ln^4n$, the optimal risk is not obtained in the current analysis and remains as an open problem. 

We follow the similar manner of \autocite{Jiao2015MinimaxDistributions} to derive the estimator that achieves \cref{thm:optimal-rate-0-1/2}. In \autocite{Jiao2015MinimaxDistributions}, the optimal estimator is constructed from two estimators; best polynomial estimator and bias-corrected plugin estimator. The best polynomial estimator is an unbiased estimator of the polynomial that best approximates $\phi$. The bias-corrected plugin estimator is an estimator obtained by applying two techniques to the plugin estimator of $\phi$. The first technique is \citeauthor{Miller1955NoteEstimates}'s bias-correction~\autocite{Miller1955NoteEstimates}, which offsets the second-order approximation of the bias. The second technique deals with the requirement from the plugin estimator; that is, the plugging function should be smooth. \citeauthor{Jiao2015MinimaxDistributions} utilized some interpolation technique to fulfill the smoothness requirement. However, the second technique is not applicable to general $\phi$ directly. We therefore introduce {\em truncation operator} as a surrogate of the interpolation technique. We will describe the optimal estimator for general $\phi$ using the truncation operator in \cref{sec:estimator}.

Next, the following theorem gives the optimal risk of the additive functional estimation for $\alpha \in (1/2,1)$.
\begin{theorem}\label{thm:optimal-rate-1/2-1}
  Suppose $\phi:[0,1]\to\RealSet$ is a function whose fourth divergence speed is $p^{\alpha}$ for $\alpha \in (1/2,1)$. If $n \gtrsim k^{1/\alpha}/\ln k$, the minimax optimal risk is obtained as
  \begin{align}
    R^*(n,k;\phi) \asymp \frac{k^2}{(n\ln n)^{2\alpha}} + \frac{k^{2-2\alpha}}{n}, 
  \end{align}
  otherwise there is no consistent estimator.
\end{theorem}
\textcite{Jiao2015MinimaxDistributions} also showed the same minimax optimal risk for $\phi(p) = p^{\alpha}$ with this range of $\alpha$. \cref{thm:optimal-rate-1/2-1} generalizes their result to general $\phi$. The optimal estimator for $\alpha \in (1/2,1)$ is equivalent to the estimator for $\alpha \in (0,1/2]$.

If $\alpha = 1$, the optimal minimax risk is obtained as the following theorem.
\begin{theorem}\label{thm:optimal-rate-1}
  Suppose $\phi:[0,1]\to\RealSet$ is a function whose fourth divergence speed is $p^{\alpha}$ for $\alpha = 1$. If $n \gtrsim k^{1/\alpha}/\ln k$, the minimax optimal risk is obtained as
  \begin{align}
    R^*(n,k;\phi) \asymp \frac{k^2}{(n\ln n)^{2}} + \frac{\ln^2k}{n},
  \end{align}
  otherwise there is no consistent estimator.
\end{theorem}
One of the functions that satisfies the divergence speed assumption with $\alpha = 1$ is $\phi(p)=-p\ln p$, and the optimal risk of this function was revealed by \textcite{Jiao2015MinimaxDistributions,Wu2016MinimaxApproximation}. They showed the same optimal risk for $\phi(p)=-p\ln p$, and thus \cref{thm:optimal-rate-1} is a generalized version of their result. The optimal estimator is also equivalent to the estimator for $\alpha \in (0,1/2]$. 

If $\alpha \in (1,3/2)$, the optimal minimax risk is obtained as the following theorem.
\begin{theorem}\label{thm:optimal-rate-1-3/2}
  Suppose $\phi:[0,1]\to\RealSet$ is a function whose sixth divergence speed is $p^{\alpha}$ for $\alpha \in (1,3/2)$. If $n \gtrsim k^{1/\alpha}/\ln k$, the minimax optimal risk is obtained as
  \begin{align}
    R^*(n,k;\phi) \asymp \frac{k^2}{(n\ln n)^{2\alpha}} + \frac{1}{n}, 
  \end{align}
  otherwise there is no consistent estimator.
\end{theorem}
For this range, \textcite{Jiao2015MinimaxDistributions} showed the minimax optimal risk for $\phi(p) = p^{\alpha}$ as $\frac{1}{(n\ln n)^{2\alpha-2}}$ under the condition that $k \asymp n\ln n$. In contrast, we success to prove the minimax optimal risk for this range without the condition $k \asymp n\ln n$. The first term corresponds to their result because it is same as their result under the condition they assume. The optimal estimator is similar to the estimator for $\alpha \in (0,1/2]$ except that we apply a slightly different bias-correction technique to the bias-corrected plugin estimator. In the estimator for $\alpha \in (0,1/2]$, we use the bias-correction technique of \textcite{Miller1955NoteEstimates}, which offsets the second order Taylor approximation of bias. Instead of the \citeauthor{Miller1955NoteEstimates}'s technique, we introduce a technique that offsets the fourth order Taylor approximation of the bias. This technique will be explained in detail in \cref{sec:estimator}. 

A lower bound for the second term in \cref{thm:optimal-rate-1/2-1,thm:optimal-rate-1,thm:optimal-rate-1-3/2} is easily obtained by applying Le Cam's two point method~\autocite{Tsybakov2009IntroductionEstimation}. For proving the lower bound of the first term in \cref{thm:optimal-rate-0-1/2,thm:optimal-rate-1/2-1,thm:optimal-rate-1,thm:optimal-rate-1-3/2}, we follow the same manner of the analysis given by \textcite{Wu2015ChebyshevUnseen}, in which the minimax lower bound is connected to the lower bound on the best polynomial approximation.  Our careful analysis of the best polynomial approximation yields the lower bound. 

The optimal minimax risk for $\alpha \in [3/2,2]$ is obtained as follows.
\begin{theorem}\label{thm:optimal-rate-3/2}
  Suppose $\phi:[0,1]\to\RealSet$ is a function whose second divergence speed is $p^{\alpha}$ for $\alpha \in [3/2,2]$. Then, the minimax optimal risk is obtained as
  \begin{align}
    R^*(n,k;\phi) \asymp \frac{1}{n}.
  \end{align}
\end{theorem}
\begin{remark}
  The second divergence speed of $\phi$ is $p^2$ as long as the second derivative of $\phi$ is bounded. Since a function whose $\ell$th divergence speed is $p^\alpha$ for any $\alpha \ge 2$ and $\ell \ge \alpha$ has bounded second derivative, this result covers all cases for $\alpha \ge 3/2$.
\end{remark}
The upper bound is obtained by employing the plugin estimator; its analysis is easy if $\alpha = 2$ because the second derivative of $\phi$ is bounded. For $\alpha \in [3/2,2)$, we extend the analysis of $\phi(p)=p^\alpha$ given by \textcite{Jiao2017MaximumDistributions} to be applicable to general $\phi$. The lower bound can be obtained easily by application of Le Cam's two point method~\autocite{Tsybakov2009IntroductionEstimation}. 

\section{Estimator for $\alpha \in (0,3/2)$}\label{sec:estimator}
In this section, we describe our estimator for $\theta$ in the case $\alpha \in (0,3/2)$. The optimal estimator for $\alpha \in (0,3/2)$ is composed of the bias-corrected plugin estimator and the best polynomial estimator. We first describe the overall estimation procedure on the supposition that the bias-corrected plugin estimator and the best polynomial estimator are black boxes. Then, we describe the bias-corrected plugin estimator and the best polynomial estimator in detail.

For simplicity, we assume the samples are drawn from the Poisson sampling model, where we first draw $n' \sim \Poi(2n)$, and then draw $n'$ i.i.d. samples $S_{n'} = \cbrace{X_1,...,X_{n'}}$. Given the samples $S_{n'}$, we first partition the samples into two sets. We use one set of the samples to determine whether the bias-corrected plugin estimator or the best polynomial estimator should be employed, and the other set is used to estimate $\theta$. Let $\cbrace{B_i}_{i=1}^{n'}$ be i.i.d. random variables drawn from the Bernoulli distribution with parameter $1/2$, i.e., $\p\cbrace{B_i = 0} = \p\cbrace{B_i = 1} = 1/2$ for $i=1,...,n'$. We partition $(X_1,...,X_{n'})$ according to $(B_1,...,B_{n'})$, and construct the histograms $\tilde{N}$ and $\tilde{N}'$, which are defined as
\begin{align}
  \tilde{N}_i = \sum_{j=1}^{n'} \ind{X_j = i}\ind{B_j = 0},\: \tilde{N}'_i = \sum_{j=1}^{n'} \ind{X_j = i}\ind{B_j = 1},
\end{align}
for $i \in [n]$. Then, $\tilde{N}$ and $\tilde{N}'$ are independent histograms, and $\tilde{N}_i,\tilde{N}'_i \sim \Poi(np_i)$.

Given $\tilde{N}'$, we determine whether the bias-corrected plugin estimator or the best polynomial estimator should be employed for each alphabet. Let $\Delta_{n,k}$ be a threshold depending on $n$ and $k$ to determine which estimator is employed, which will be specified as in \cref{thm:upper-bound1} on \cpageref{thm:upper-bound1}. We apply the best polynomial estimator if $\tilde{N}'_i < 2\Delta_{n,k}$, and otherwise, i.e., $\tilde{N}'_i \ge 2\Delta_{n,k}$, we apply the bias-corrected plugin estimator. Let $\phi_{\mathrm{poly}}$ and $\phi_{\mathrm{plugin}}$ be the best polynomial estimator and the bias-corrected plugin estimator for $\phi$, respectively. Then, the estimator of $\theta$ is written as
\begin{align}
  \hat\theta(\tilde{N})\!=\!\sum_{i=1}^k \paren*{ \ind{\tilde{N}'_i \ge 2\Delta_{n,k}}\phi_{\mathrm{plugin}}(\tilde{N}_i)\!+\!\ind{\tilde{N}'_i < 2\Delta_{n,k}}\phi_{\mathrm{poly}}(\tilde{N}_i) }.
\end{align}
Next, we describe the details of the best polynomial estimator $\phi_{\mathrm{poly}}$ and the bias-corrected plugin estimator $\phi_{\mathrm{plugin}}$.

\subsection{Best Polynomial Estimator}
The best polynomial estimator is an unbiased estimator of the polynomial that provides the best approximation of $\phi$. Let $\cbrace{a_m}_{m=0}^L$ be coefficients of the polynomial that achieves the best approximation of $\phi$ by a degree-$L$ polynomial with range $I=[0,\frac{4\Delta_{n,k}}{n}\land1]$, where $L$ will be specified in \cref{thm:upper-bound1} on \cpageref{thm:upper-bound1}. Then, the approximation of $\phi$ by the polynomial at point $p_i$ is written as
\begin{align}
 \phi_L(p_i) =& \sum_{m=0}^L a_m p_i^m. \label{eq:approx-poly}
\end{align}
From \cref{eq:approx-poly}, an unbiased estimator of $\phi_L$ can be derived from an unbiased estimator of $p_i^m$. For the random variable $\tilde{N}_i$ drawn from the Poisson distribution with mean parameter $np_i$, the expectation of the $m$th factorial moment $(\tilde{N}_i)_m = \tfrac{\tilde{N}_i!}{(\tilde{N}_i - m)!}$ becomes $(np_i)^m$. Thus, $\frac{(\tilde{N}_i)_m}{n^m}$ is an unbiased estimator of $p_i^m$. Substituting this into \cref{eq:approx-poly} gives the unbiased estimator of $\phi_L(p_i)$ as
\begin{align}
  \bar\phi_{\mathrm{poly}}(\tilde{N}_i) = \sum_{m=0}^L \frac{a_m}{n^j} (\tilde{N}_i)_m.
\end{align}
Next, we truncate $\bar\phi_{\mathrm{poly}}$ so that it is not outside of the domain of $\phi(p)$. Let $\phi_{{\mathrm{inf}},\frac{4\Delta_{n,k}}{n}} = \inf_{p \in [0,\frac{4\Delta_{n,k}}{n}\land1]} \phi(p)$ and $\phi_{{\mathrm{sup}},\frac{4\Delta_{n,k}}{n}} = \sup_{p \in [0,\frac{4\Delta_{n,k}}{n}\land1]} \phi(p)$. Then, the best polynomial estimator is defined as
\begin{align}
 \phi_{\mathrm{poly}}(\tilde{N}_i) = (\bar\phi_{\mathrm{poly}}(\tilde{N}_i) \land \phi_{{\mathrm{sup}},\frac{4\Delta_{n,k}}{n}}) \lor \phi_{{\mathrm{inf}},\frac{4\Delta_{n,k}}{n}}.
\end{align}

\subsection{Bias-corrected Plugin Estimator}
The problem of the plugin estimator is that it causes a large bias when the occurrence probability $p_i$ of an alphabet is small. This large bias comes from a large derivative around $p = 0$. To avoid this large bias, we truncate the function $\phi$ as follows:
\begin{align}
    T_\Delta[\phi](p) =& \begin{dcases}
      \phi(\Delta) & \textif p < \Delta, \\
      \phi(p) & \textif \Delta \le p \le 1, \\
      \phi(1) & \textif p > 1.
    \end{dcases}
\end{align}
Besides, we define the derivative of the truncated function as
\begin{align}
    T^{(\ell)}_\Delta[\phi](p) =& \begin{dcases}
      0 & \textif p < \Delta, \\
      \phi^{(\ell)}(p) & \textif \Delta \le p \le 1, \\
      0 & \textif p > 1,
    \end{dcases}
\end{align}
where $\phi^{(\ell)}(1) = \lim_{p \uparrow 1}\phi^{(\ell)}(p)$. Then, we construct bias-corrected plugin estimators for this truncated function. Note that this derivative is not the same as the standard derivative of $T_\Delta[\phi](p)$ because $T_\Delta[\phi](p)$ is not differentiable at $p = \Delta$ and $p = 1$. Even when differentiability is lost by this truncation, a technique using the generalized Hermite interpolation shown in \cref{sec:upper-analysis} enables us to obtain the bounds on the bias and variance of the bias-corrected plugin estimator for this truncated function.

We use a slightly different bias correction methods for $\alpha \in (0,1]$ and $\alpha \in (1,3/2)$. We use {\em second-order bias correction} for $\alpha \in (0,1]$ and {\em fourth-order bias correction} for $\alpha \in (1,3/2)$. We describe the second-order bias correction first and then move onto description of the fourth-order bias correction.

{\bfseries Second-order bias correction.}
For $\alpha \in (0,1]$, we employ the plugin estimator with the \citeauthor{Miller1955NoteEstimates}'s bias correction~\autocite{Miller1955NoteEstimates}. The bias correction offsets the second-order Taylor approximation of the bias, which is obtained as follows.
\begin{align}
  \Mean\bracket*{\phi\paren*{\frac{\tilde{N}}{n}} - \phi(p)} \approx& \Mean\bracket*{\frac{\phi^{(2)}(p)}{2}\paren*{\frac{\tilde{N}}{n} - p}^2} = \frac{p\phi^{(2)}(p)}{2n},
\end{align}
where $\tilde{N} \sim \Poi(np)$ for $p \in (0,1)$. The bias corrected function is hence obtained as $\phi_2(p) = \phi(p) - \frac{p\phi^{(2)}(p)}{2n}$.

Using the truncation operator, the truncated second-order bias corrected function is defined as
\begin{align}
    \bar\phi_{2,\Delta}(p) = T_\Delta[\phi](p) - \frac{p}{2n}T^{(2)}_\Delta[\phi](p).
\end{align}
Then, $\phi_{\mathrm{plugin}}$ is the plugin estimator of $\bar\phi_2$; that is
\begin{align}
  \phi_{\mathrm{plugin}}(\tilde{N}) = \bar\phi_{2,\Delta}\paren*{\frac{\tilde{N}}{n}}. \label{eq:plugin-estimator1}
\end{align}

{\bfseries Forth-order bias correction.}
For $\alpha \in (1,3/2)$, we employ the fourth order bias correction. In analogy with the second-order bias correction, the fourth order bias correction offsets the fourth order approximation of bias. By the Taylor approximation, the bias of the plugin estimator for $\phi_2$ is obtained as
\begin{align}
  \Mean\bracket*{\phi_2\paren*{\frac{\tilde{N}}{n}} - \phi(p)} \approx& -\frac{p\phi^{(3)}}{3n^2} - \frac{5p\phi^{(4)}}{24n^3} - \frac{p^2\phi^{(4)}}{8n^2}.
\end{align}
Thus, the fourth-order bias corrected function is obtained as $\phi_4(p) = \phi_2(p) + \frac{p\phi^{(3)}}{3n^2} + \frac{5p\phi^{(4)}}{24n^3} + \frac{p^2\phi^{(4)}}{8n^2}$.

As well as the second order bias correction, we define the truncated fourth-order bias corrected function as
\begin{multline}
  \bar\phi_{4,\Delta}(p) = T_\Delta[\phi](p) - \frac{p}{2n}T^{(2)}_\Delta[\phi](p) \\ + \frac{p}{3n^2}T_\Delta^{(3)}[\phi](p) + \frac{5p}{24n^3}T_\Delta^{(4)}[\phi](p) + \frac{p^2}{8n^2}T_\Delta^{(4)}[\phi](p).
\end{multline}
Then, $\phi_{\mathrm{plugin}}$ is the plugin estimator of $\bar\phi_4$; that is,
\begin{align}
  \phi_{\mathrm{plugin}}(\tilde{N}) = \bar\phi_{4,\Delta}\paren*{\frac{\tilde{N}}{n}}. \label{eq:plugin-estimator2}
\end{align}

\section{Consequential Properties of Divergence Speed Assumption}
In this section, we present some properties that come from the divergence speed assumption; these properties are useful for the later analyses.

\subsection{Lower Order Divergence Speed}
If the $\ell$th divergence speed of a function $\phi$ is $p^\alpha$ for some $\ell$ and $\alpha$, the same divergence speed is satisfied for the lower order, such as $\ell-1$, $\ell-2$,..., under a certain condition. The precise claim is shown in the following lemma.
\begin{lemma}\label{lem:lower-div}
  For a real $\alpha$, let $\phi:[0,1]\to\RealSet$ be an $\ell$-times continuously differentiable function on $(0,1)$ whose $\ell$th divergence speed is $p^\alpha$. If $\ell > 1 + \alpha$, the $(\ell-1)$th divergence speed of $\phi$ is also $p^\alpha$.
\end{lemma}
The proof of this lemma can be found in \cref{sec:lower-div}. We immediately obtain a consequence lemma of \cref{lem:lower-div}:
\begin{lemma}\label{lem:lowest-div}
  For a real $\alpha$, let $\phi:[0,1]\to\RealSet$ be an $\ell$-times continuously differentiable function on $(0,1)$ whose $\ell$th divergence speed is $p^\alpha$. Then, for any positive integer $m$ such that $m < \alpha$ and $m \le \ell$, the $m$th divergence speed of $\phi$ is also $p^\alpha$.
\end{lemma}
\begin{proof}
  Applying \cref{lem:lower-div} $(\ell-m)$ times yields the claim.
\end{proof}

When $\alpha$ is an integer, $\alpha$th derivative of $\phi(p)$ diverges as $p$ approaches to zero with logarithmic divergence speed as shown in the following lemma:
\begin{lemma}\label{lem:lowest-int-div}
  For a positive integer $\alpha$, let $\phi:[0,1]\to\RealSet$ be an $\ell$-times continuously differentiable function on $(0,1)$ whose $\ell$th divergence speed is $p^\alpha$, where $\alpha < \ell$. Then, there exists constants $W_\alpha > 0$, $c_\alpha \ge 0$ and $c'_\alpha \ge 0$ such that for all $p \in (0,1)$,
  \begin{align}
      W_\alpha\ln(1/p) - c'_\alpha \le \abs*{\phi^{(\alpha)}(p)} \le W_\alpha\ln(1/p) + c_\alpha.
  \end{align}
\end{lemma}
The proof of \cref{lem:lowest-int-div} is also shown in \cref{sec:lower-div}. The bounds on $\abs{\phi^{(\alpha)}(p)}$ in \cref{lem:lowest-int-div} is similar to the divergence speed in \cref{def:div-speed} except $p^{-\ell+\alpha}$ is replaced by $\ln(1/p)$. The function $\ln(1/p)$ diverges to infinity as $p$ approaches to zero where its speed is slower than $p^{-\ell+\alpha}$ for any real $\alpha$ and any integer $\ell$ such that $\alpha < \ell$.

\subsection{H\"older Continuity}

The divergence speed assumption induces H\"older continuity to $\phi$. For a real $\beta \in (0,1]$, a function $\phi:I\to\RealSet$ is {\em $\beta$-H\"older continuous} on $I$ if
\begin{align}
  \norm*{\phi}_{C^{H,\beta}} = \sup_{x \ne y \in I}\frac{\abs*{\phi(x)-\phi(y)}}{\abs*{x-y}^\beta} < \infty.
\end{align}
In particular, $1$-H\"older continuity is known as Lipschitz continuity.

We reveal H\"older continuity in $\phi$ and its derivative under the assumption that the $\ell$th divergence speed of $\phi$ is $p^\alpha$. We derive the H\"older continuity by dividing $\alpha$ into three cases; $\alpha \in (0,1)$, $\alpha = 1$, and $\alpha \in (1,2)$.
\begin{lemma}\label{lem:div-speed-holder1}
 Suppose $\phi:[0,1]\to\RealSet$ is a function whose $\ell$th divergence speed is $p^{\alpha}$ for $\alpha \in (0,1)$ and $\alpha \le \ell$. Then, $\phi$ is $\alpha$-H\"older continuous.
\end{lemma}
\begin{lemma}\label{lem:div-speed-holder2}
 Suppose $\phi:[0,1]\to\RealSet$ is a function whose $\ell$th divergence speed is $p^{\alpha}$ for $\alpha = 1$ and $\alpha \le \ell$. Then, $\phi$ is $\beta$-H\"older continuous for any $\beta \in (0,1)$.
\end{lemma}
\begin{lemma}\label{lem:div-speed-holder3}
   Suppose $\phi:[0,1]\to\RealSet$ is a function whose $\ell$th divergence speed is $p^{\alpha}$ for $\alpha \in (1,2)$ and $\alpha \le \ell$ such that $\phi^{(1)}(0) = 0$. If $\alpha \in (1,2)$, $\phi$ is Lipschitz continuous, and $\phi^{(1)}$ is $(\alpha-1)$-H\"older continuous.
\end{lemma}
Note that we can assume $\phi^{(1)}(0)=0$ without loss of generality because, for any $c \in \RealSet$, $\theta(P;\phi) = \theta(P;\phi_c)$ where $\phi_c(p) = \phi(p) + c(p-1/k)$.

\section{Upper Bound Analysis}\label{sec:upper-analysis}
In this section, we analyze the worst-case quadratic errors of the plugin estimator and the proposed estimator described in \cref{sec:estimator}. We will prove the following theorems:
\begin{theorem}\label{thm:upper-bound1}
  Suppose $\phi:[0,1]\to\RealSet$ is a function such that one of the following condition holds:
  \begin{enumerate}
    \item the fourth divergence speed of $\phi$ is $p^{\alpha}$ for $\alpha \in (0,1]$,
    \item the sixth divergence speed of $\phi$ is $p^{\alpha}$ for $\alpha \in (1,3/2)$.
  \end{enumerate}
  Let $L = \floor{C_1\ln n}$ and $\Delta_{n,k} = C_2\ln n$ where $C_1$ and $C_2$ are universal constants such that $C_2 > 8\alpha$, $C_2^3C_1 \le 1/2$, and $2 - 3C_1 \ln 2 - 2\sqrt{C_1C_2}\ln(2e) > \alpha$. If $\alpha \in (0,1/2)$, the worst-case risk of $\hat\theta$ is bounded above as
  \begin{align}
    \sup_{P \in \dom{M}_k} \Mean\bracket*{\paren*{\hat\theta\paren{\tilde{N}} - \theta(P)}^2} \lesssim \frac{k^2}{(n\ln n)^{2\alpha}},
  \end{align}
  where we need $k \gtrsim \ln^{4\alpha-1}n$ if $\alpha \in (1/4,1/2)$. If $\alpha \in [1/2,1)$, the worst-case risk of $\hat\theta$ is bounded above as
  \begin{align}
    \sup_{P \in \dom{M}_k} \Mean\bracket*{\paren*{\hat\theta\paren{\tilde{N}} - \theta(P)}^2} \lesssim \frac{k^2}{(n\ln n)^{2\alpha}} + \frac{k^{2-2\alpha}}{n},
  \end{align}
  If $\alpha = 1$, the worst-case risk of $\hat\theta$ is bounded above as
  \begin{align}
    \sup_{P \in \dom{M}_k} \Mean\bracket*{\paren*{\hat\theta\paren{\tilde{N}} - \theta(P)}^2} \lesssim \frac{k^2}{(n\ln n)^{2}} + \frac{\ln^2 k}{n}.
  \end{align}
  If $\alpha \in (1,3/2)$, the worst-case risk of $\hat\theta$ is bounded above as
  \begin{align}
    \sup_{P \in \dom{M}_k} \Mean\bracket*{\paren*{\hat\theta\paren{\tilde{N}} - \theta(P)}^2} \lesssim \frac{k^2}{(n\ln n)^{2\alpha}} + \frac{1}{n}.
  \end{align}
\end{theorem}
\begin{theorem}\label{thm:upper-bound2}
  Suppose $\phi:[0,1]\to\RealSet$ is a function whose second divergence speed is $p^{\alpha}$ for $\alpha \in [3/2,2]$. Then, the worst-case risk of the plugin estimator is bounded above as
  \begin{align}
    \sup_{P \in \dom{M}_k} \Mean\bracket*{\paren*{\hat\theta_{\mathrm{plugin}}\paren{N} - \theta(P)}^2} \lesssim \frac{1}{n}.
  \end{align}
\end{theorem}
We firstly give the proof of \cref{thm:upper-bound1} in \cref{sec:upper1}. Then, we move to the proof of \cref{thm:upper-bound2} in \cref{sec:upper2}.

For convenience, we use $\Bias$ and $\Var$ to denote bias and variance, which are formally defined as
\begin{align}
  \Bias\bracket*{X} = \abs*{\Mean\bracket*{X}}, \: \Var\bracket*{X} = \Mean\bracket*{\paren*{X - \Mean[X]}^2},
\end{align}
for a random variable $X$.

\subsection{Analyses for $\alpha \in (0,3/2)$}\label{sec:upper1}

The first step to prove \cref{thm:upper-bound1} is the bias-variance decomposition, which gives for any $P \in \dom{M}_k$,
\begin{align}
 \Mean\bracket*{\paren*{\hat\theta(\tilde{N})\!-\!\theta(P)}^2} \!=\!\Bias\bracket*{\hat\theta(\tilde{N})\!-\!\theta(P)}^2\!+\!\Var\bracket*{\hat\theta(\tilde{N})}.
\end{align}
We therefore will derive bounds on the bias and the variance of $\hat\theta$.

Let the bias and variance of $\hat\phi_{\mathrm{poly}}$ and $\hat\phi_{\mathrm{plugin}}$ for each alphabet $i \in [k]$ be 
 \begin{align}
     b_{\mathrm{plugin},i} =& \Bias\bracket*{\phi_{\mathrm{plugin}}(\tilde{N}_i) - \phi(p_i) }, \\
     b_{\mathrm{poly},i} =& \Bias\bracket*{\phi_{\mathrm{poly}}(\tilde{N}_i) - \phi(p_i) }, \\
     v_{\mathrm{plugin},i} =& \Var\bracket*{\phi_{\mathrm{plugin}}(\tilde{N}_i) },\\
     v_{\mathrm{poly},i} =& \Var\bracket*{\phi_{\mathrm{poly}}(\tilde{N}_i) }.
 \end{align}
By the following lemmas, we can obtain upper bounds on the bias and variance of $\hat\theta$ by using $b_{\mathrm{plugin},i}, b_{\mathrm{poly},i}, v_{\mathrm{plugin},i}$, and $v_{\mathrm{poly},i}$.
\begin{lemma}\label{lem:upper-ind-bias}
  Given $P \in \dom{M}_k$, the bias of $\hat\theta$ is bounded above as
  \begin{multline}
   \Bias\bracket*{\tilde\theta\paren{\tilde{N}} - \theta(P)} \le \\
   \sum_{i = 1}^k \paren[\Bigg]{
     \paren*{(e/4)^{\Delta_{n,k}}\ind{np_i \le \Delta_{n,k}} + \ind{np_i > \Delta_{n,k}}}b_{\mathrm{plugin},i} \\
     + \paren*{\ind{np_i \le 4\Delta_{n,k}} + e^{-\Delta_{n,k}/8}\ind{np_i > 4\Delta_{n,k}}}b_{\mathrm{poly},i}
   }.
 \end{multline}
\end{lemma}
\begin{lemma}\label{lem:upper-ind-var}
 Given $P \in \dom{M}_k$, the variance of $\hat\theta$ is bounded above as
 \begin{multline}
  \Var\bracket*{\tilde\theta\paren{\tilde{N}} - \theta(P)} \le \\
   \sum_{i = 1}^k \paren[\Bigg]{
   \paren*{(e/4)^{\Delta_{n,k}}\ind{np_i \le \Delta_{n,k}} + \ind{np_i > \Delta_{n,k}}}v_{\mathrm{plugin},i} \\
   + \paren*{\ind{np_i \le 4\Delta_{n,k}} + e^{-\Delta_{n,k}/8}\ind{np_i > 4\Delta_{n,k}}}v_{\mathrm{poly},i} \\ + 2\paren*{(e/4)^{\Delta_{n,k}}\ind{np_i \le \Delta_{n,k}} + \ind{np_i > \Delta_{n,k}}}b_{\mathrm{plugin},i}^2 \\ + 2\paren*{\ind{np_i \le 4\Delta_{n,k}} + e^{-\Delta_{n,k}/8}\ind{np_i > 4\Delta_{n,k}}}b_{\mathrm{poly},i}^2
  }.
 \end{multline}
\end{lemma}
As proved in \cref{lem:upper-ind-bias,lem:upper-ind-var}, bounds on the bias and variance of the bias-corrected plugin estimator and  best polynomial estimator for each individual alphabet give the bounds on the bias and variance of our estimator. Hence, we next analyze the bias and variance of the bias-corrected plugin estimator and best polynomial estimator for a certain alphabet. In the next two sub-subsections, we use $p$ to denote the occurrence probability of the certain alphabet, and let $\tilde{N} \sim \Poi(np)$.

\subsubsection{Bias and Variance of Best Polynomial Estimator}
Given a positive integer $L$ and a positive real $\Delta$, let $\phi_L(p)=\sum_{m=0}^La_mp^m$ be the optimal uniform approximation of $\phi$ by degree-$L$ polynomials on $[0,\Delta]$, and let $g_L(\tilde{N}) = \sum_{m=0}^La_m(\tilde{N})_m/n^m$ be an unbiased estimator of $\phi_L(p)$. The best polynomial estimator with $\Delta = 4\Delta_{n,k}/n$ is written as
\begin{align}
    \phi_{\mathrm{plugin}}(\tilde{N}) = (g_L(\tilde{N}) \land \phi_{{\mathrm{sup}},\Delta})\lor \phi_{{\mathrm{inf}},\Delta},
\end{align}
where $\phi_{\mathrm{sup},\Delta} = \sup_{p \in [0,\Delta]}\phi(p)$ and $\phi_{\mathrm{inf},\Delta} = \inf_{p \in [0,\Delta]}\phi(p)$. To take advantage of \cref{lem:upper-ind-bias,lem:upper-ind-var}, we derive upper bounds on the bias and variance of this estimator in two cases; $p > \Delta$ and $p \le \Delta$.

For $p > \Delta$, it suffices to prove the bias and variance do not increase as $n$ and $k$ increase because in this range of $p$, the best polynomial estimator is not used with high probability. This fact is proven by the following lemma:
\begin{lemma}\label{lem:poly-o1}
  If there is a finite universal constant $C > 0$ such that $\sup_{p \in [0,1]}\phi(p) \le C$, then
 \begin{align}
  \Bias\bracket*{(g_L(\tilde{N}) \land \phi_{{\mathrm{sup}},\Delta})\lor \phi_{{\mathrm{inf}},\Delta} - \phi(p) } \lesssim 1,
 \end{align}
 and
 \begin{align}
  \Var\bracket*{(g_L(\tilde{N}) \land \phi_{{\mathrm{sup}},\Delta})\lor \phi_{{\mathrm{inf}},\Delta} } \lesssim 1,
 \end{align}
\end{lemma}

We next analyze the bias and variance for $p \le \Delta$:
\begin{lemma}\label{lem:poly-bias}
 If $p \le \Delta$, we have
\begin{multline}
  \Bias\bracket*{(g_L(\tilde{N}) \land \phi_{{\mathrm{sup}},\Delta})\lor \phi_{{\mathrm{inf}},\Delta} - \phi(p) } \\ \lesssim \sqrt{\Var\bracket*{g_L(\tilde{N})}} + E_L\paren*{\phi,[0,\Delta]}.
 \end{multline}
\end{lemma}
\begin{lemma}\label{lem:poly-var}
 If $p \le \Delta$, $2\Delta^3L \le n$, and there is an universal constant $C > 0$ such that $\sup_{p \in [0,1]}\abs*{\phi(p)} \le C$, we have
 \begin{align}
  \Var\bracket*{g_L(\tilde{N})} \lesssim \frac{\Delta^3 L 64^L (2e)^{2\sqrt{\Delta n L }}}{n}.
 \end{align}
\end{lemma}
 It is obviously that truncation does not increase the variance, i.e.,
 \begin{align}
   \Var\bracket*{(g_L(\tilde{N}) \land \phi_{{\mathrm{sup}},\Delta})\lor \phi_{{\mathrm{inf}},\Delta} } \le \Var\bracket*{g_L(\tilde{N}) }.
 \end{align}
 Thus, the result in \cref{lem:poly-var} is the variance upper bound on the best polynomial estimator.

 As shown in \cref{lem:poly-bias}, the best polynomial approximation error $E_L(\phi,[0,\Delta])$ gives the bound on the bias of the best polynomial estimator. The best polynomial approximation error under the divergence speed assumption can be found in \cref{sec:best-poly}. By substituting the result shown in \cref{sec:best-poly} into the result in \cref{lem:poly-bias}, we obtain the bias upper bound on the best polynomial estimator.

\subsubsection{Bias and Variance of Bias-corrected Plugin Estimator}

Given a positive real $\Delta = \Delta_{n,k}/n$, the bias-corrected plugin estimator is
\begin{align}
    \phi_{\mathrm{plugin}}(\tilde{N}) = \begin{dcases}
      \bar\phi_{2,\Delta}\paren*{\frac{\tilde{N}}{n}} & \textif \alpha \in (0,1], \\
      \bar\phi_{4,\Delta}\paren*{\frac{\tilde{N}}{n}} & \textif \alpha \in (1,3/2).
    \end{dcases}
\end{align}
As well as the analysis of the best polynomial estimator, we derive bounds on the bias and variance of this estimator in two cases; $p \le \Delta$ and $p > \Delta$.

{\bfseries Bias and variance analysis for $p \le \Delta$.}
In this case, it suffices to prove the bias and variance are bounded as $\lesssim 1$ with the same reason of the analysis of the best polynomial estimator for $p \ge \Delta$.
\begin{lemma}\label{lem:plugin-o1-2}
  Suppose $\phi:[0,1]\to\RealSet$ be a function whose second divergence speed is $p^\alpha$ for $\alpha \in (0,2)$. If $\Delta \in (0,1)$ and $\Delta \gtrsim n^{-1}$,
 \begin{align}
  \Bias\bracket*{\bar\phi_{2,\Delta}\paren*{\frac{\tilde{N}}{n}} - \phi(p) } \lesssim 1 \textand \Var\bracket*{ \bar\phi_{2,\Delta}\paren*{\frac{\tilde{N}}{n}} } \lesssim 1.
 \end{align}
\end{lemma}
\begin{lemma}\label{lem:plugin-o1-4}
  Suppose $\phi:[0,1]\to\RealSet$ be a function whose second divergence speed is $p^\alpha$ for $\alpha \in (0,2)$. If $\Delta \in (0,1)$ and $\Delta \gtrsim n^{-1}$,
 \begin{align}
  \Bias\bracket*{ \bar\phi_{4,\Delta}\paren*{\frac{\tilde{N}}{n}} - \phi(p) } \lesssim 1 \textand \Var\bracket*{ \bar\phi_{4,\Delta}\paren*{\frac{\tilde{N}}{n}} } \lesssim 1.
 \end{align}
\end{lemma}

{\bfseries Bias and variance analysis for $p > \Delta$.}
In this case, we will prove the following lemmas:
\begin{lemma}\label{lem:plugin-bias-0-1}
 Suppose $\phi:[0,1]\to\RealSet$ is a function whose fourth divergence speed is $p^{\alpha}$ for $\alpha \in (0,1]$. Suppose $\frac{1}{n} \lesssim \Delta < p \le 1$. Then, we have
 \begin{align}
  \Bias\bracket*{\bar\phi_{2,\Delta}\paren*{\frac{\tilde{N}}{n}} - \phi(p)} \lesssim \frac{1}{n^2\Delta^{2-\alpha}}.
 \end{align}
\end{lemma}
\begin{lemma}\label{lem:plugin-var-0-1}
 Suppose $\phi:[0,1]\to\RealSet$ is a function whose fourth divergence speed is $p^{\alpha}$ for $\alpha \in (0,1]$. Suppose $\frac{1}{n} \lesssim \Delta < p \le 1$. For $\alpha \in (0,1)$, we have
 \begin{align}
   \Var\bracket*{\bar\phi_{2,\Delta}\paren*{\frac{\tilde{N}}{n}}} \lesssim \frac{p^{2\alpha-1}}{n} + \frac{1}{n^{2}\Delta^{4-2\alpha}} + \frac{p}{n}.
 \end{align}
 For $\alpha = 1$, we have
 \begin{align}
   \Var\bracket*{\bar\phi_{2,\Delta}\paren*{\frac{\tilde{N}}{n}}} \lesssim \frac{p\ln^2p}{n} + \frac{1}{n^{2}\Delta^{4-2\alpha}} + \frac{p}{n}.
 \end{align}
\end{lemma}
\begin{lemma}\label{lem:plugin-bias-1-3/2}
 Suppose $\phi:[0,1]\to\RealSet$ is a function whose sixth divergence speed is $p^{\alpha}$ for $\alpha \in (1,3/2)$. Suppose $\frac{1}{n} \lesssim \Delta < p \le 1$. Then, we have
 \begin{align}
  \Bias\bracket*{\bar\phi_{4,\Delta}\paren*{\frac{\tilde{N}}{n}} - \phi(p)} \lesssim \frac{1}{n^3\Delta^{3-\alpha}}.
 \end{align}
\end{lemma}
\begin{lemma}\label{lem:plugin-var-1-3/2}
 Suppose $\phi:[0,1]\to\RealSet$ is a function whose fifth divergence speed is $p^{\alpha}$ for $\alpha \in (1,3/2)$. Suppose $\frac{1}{n} \lesssim \Delta < p \le 1$. Then, we have
 \begin{align}
   \Var\bracket*{\bar\phi_{4,\Delta}\paren*{\frac{\tilde{N}}{n}}} \lesssim \frac{1}{n^{2\alpha}} + \frac{p}{n}.
 \end{align}
\end{lemma}
To prove these lemmas, we introduce two techniques; {\em Hermite interpolation} and {\em bounds on the Taylor's reminder term}.

{\bfseries Hermite interpolation technique.}
An important tool for analyzing the bias-corrected plugin estimator for $p > \Delta$ is the generalized Hermite interpolation~\autocite{Spitzbart1960AFormula}. The generalized Hermite interpolation between $\phi(a)$ and $\phi(b)$ is obtained as
\begin{multline}
  H_L(p;\phi,a,b) = \phi(a) + \\ \sum_{m=1}^L\frac{\phi^{(m)}(a)}{m!}(p-a)^m\sum_{\ell=0}^{L-m}\frac{L+1}{L+\ell+1}\Beta_{\ell,L+\ell+1}\paren*{\frac{p-a}{b-a}},
\end{multline}
where $\Beta_{\nu,n}(x)=\binom n\nu x^\nu(1-x)^{n-\nu}$ denotes the Bernstein basis polynomial. Then, $H_L^{(i)}(a;\phi,a,b) = \phi^{(i)}(a)$ for $i = 0,...,L$ and $H_L^{(i)}(b;\phi,a,b) = 0$ for $i = 1,...,L$. Given an integer $L > 0$, positive reals $\Delta,\delta > 0$, and a function $\phi$, define a functional:
\begin{multline}
  H_{L,\Delta,\delta}[\phi](p) = \\ \begin{dcases}
    \phi(\Delta) & \textif p \le \Delta(1-\delta), \\
    H_L\paren*{p;\phi,\Delta,\Delta(1-\delta)} & \textif \Delta(1-\delta) < p \le \Delta, \\
    \phi(p) & \textif \Delta < p < 1, \\
    H_L\paren*{p;\phi,1,1+\delta} & \textif 1 \le p < 1+\delta, \\
    \phi(1) & \textif p \ge 1+\delta.
  \end{dcases}
\end{multline}
Note that $\phi^{(L)}(1) = \lim_{p \uparrow 1}\phi^{(L)}(p)$. If $\phi$ is an $L$ times continuously differentiable function, $H_{L,\Delta,\delta}[\phi]$ is $L$ times continuously differentiable everywhere on $p > 0$. Moreover, as $\delta$ tends to zero, the converted function $H_{L,\Delta,\delta}[\phi]$ converges in pointwise to the truncated function $T_\Delta[\phi]$:
\begin{lemma}\label{lem:converge-hermite}
  Suppose $\Delta \in (0,1)$ and $\phi$ is an $L$ times continuously differentiable function. For any positive real $p > 0$,
  \begin{align}
      \lim_{\delta \downarrow 0}H_{L,\Delta,\delta}[\phi](p) = T_\Delta[\phi](p).
  \end{align}
  Moreover, under the same condition, for any $\ell \le L$ and any positive real $p > 0$,
  \begin{align}
      \lim_{\delta \downarrow 0}H^{(\ell)}_{L,\Delta,\delta}[\phi](p) = T^{(\ell)}_\Delta[\phi](p).
  \end{align}
\end{lemma}

Let $\bar\phi_{2,L,\Delta,\delta}$ and $\bar\phi_{4,L,\Delta,\delta}$ be functions in which the truncated function $T_\Delta$ and $T^{(\ell)}_\Delta$ are replaced by $H_{L,\Delta,\delta}$ and $H^{(\ell)}_{L,\Delta,\delta}$, respectively. Then, by \cref{lem:converge-hermite} and Fatou--Lebesgue theorem, we have
\begin{align}
 & \Bias\bracket*{\bar\phi_{2,\Delta}\paren*{\frac{\tilde{N}}{n}} - \phi(p)} \\
 =& \Bias\bracket*{\lim_{\delta \downarrow 0}\bar\phi_{2,L,\Delta,\delta}\paren*{\frac{\tilde{N}}{n}} - \phi(p)} \\ =& \lim_{\delta \downarrow 0}\Bias\bracket*{\bar\phi_{2,L,\Delta,\delta}\paren*{\frac{\tilde{N}}{n}} - \phi(p)}.
\end{align}
Similarly, we have
\begin{align}
   &\Bias\bracket*{\bar\phi_{4,\Delta}\paren*{\frac{\tilde{N}}{n}} - \phi(p)} \\
  =& \lim_{\delta \downarrow 0}\Bias\bracket*{\bar\phi_{4,L,\Delta,\delta}\paren*{\frac{\tilde{N}}{n}} - \phi(p)},
\end{align}
\begin{align}
  \Var\bracket*{\bar\phi_{2,\Delta}\paren*{\frac{\tilde{N}}{n}}} =& \lim_{\delta \downarrow 0}\Var\bracket*{\bar\phi_{2,L,\Delta,\delta}\paren*{\frac{\tilde{N}}{n}}}, \textand \\
  \Var\bracket*{\bar\phi_{4,\Delta}\paren*{\frac{\tilde{N}}{n}}}
  =& \lim_{\delta \downarrow 0}\Var\bracket*{\bar\phi_{4,L,\Delta,\delta}\paren*{\frac{\tilde{N}}{n}}},
\end{align}
Thus, by deriving bounds on the bias and variance of the plugin estimator for $\bar\phi_{2,L,\Delta,\delta}$ and $\bar\phi_{4,L,\Delta,\delta}$, we obtain bounds on the bias and variance of the plugin estimator for $\bar\phi_{2,\Delta}$ and $\bar\phi_{4,\Delta}$.

A benefit of using the truncation is that it prevents a large derivative around $p = 0$. The functional $H_{L,\Delta,\delta}$ inherits this property, as shown in the following lemma.
\begin{lemma}\label{lem:hermite-bound}
  Let $\phi:[0,1]\to\RealSet$ be a function of which $L$the divergence speed is $p^\alpha$, where $L > \alpha$ is an universal constant. If $\ell \le L$ and $\beta > 0$ such that $\ell > \alpha + \beta$, for any $\delta \in (0,1)$, and a decreasing sequence of $\Delta \in (0,1)$, we have
  \begin{align}
    \sup_{p > 0}p^\beta\abs*{H^{(\ell)}_{L,\Delta,\delta}[\phi](p)} \lesssim \Delta^{\alpha+\beta-\ell}.
  \end{align}
  Moreover, if $\ell \le L$ and $\beta$ such that $1 \le \ell \le \alpha + \beta$, for any $\delta > 0$, and a decreasing sequence of $\Delta$, we have
  \begin{align}
    \sup_{p > 0}p^\beta\abs*{H^{(\ell)}_{L,\Delta,\delta}[\phi](p)} \lesssim 1.
  \end{align}
\end{lemma}

{\bfseries Bounds on the Taylor's reminder term.}
We analyze the bias and variance of the plugin estimator for $\bar\phi_{2,L,\Delta,\delta}$ and $\bar\phi_{4,L,\Delta,\delta}$ by using the Taylor theorem. For example, the bias of the plugin estimator of $\bar\phi_{2,L,\Delta,\delta}$ is obtained as
\begin{multline}
   \Bias\bracket*{\bar\phi_{2,4,\Delta,\delta}\paren*{\frac{\tilde{N}}{n}} - \phi(p)} = \\
     \abs[\Bigg]{\frac{p}{6n^2}\phi^{(3)}(p) + \Mean\bracket*{R_3\paren*{\frac{\tilde{N}}{n};H_{4,\Delta,\delta}[\phi],p}} \\ - \frac{1}{2n}\Mean\bracket*{R_1\paren*{\frac{\tilde{N}}{n};x \to xH^{(2)}_{4,\Delta,\delta}[\phi](x),p}}}, \label{eq:plugin-taylor}
\end{multline}
where $R_\ell(x;f,a)$ denotes the $\ell$th-order reminder term of a function $f$ at the point $a$. The precise derivation of this equation can be found in the proof of \cref{lem:plugin-bias-0-1} shown in below. Here, we introduce a technique to derive bounds on the second and third terms in \cref{eq:plugin-taylor}.

There are some explicit formulas of the Taylor's reminder term, such as he Lagrange form and the Cauchy form. However, the direct application of these form cannot yield the desired bound. We therefore derive the following lemma that provides a bound on the Taylor's reminder term by using the mean value theorem.
\begin{lemma}\label{lem:reminder-bound1}
  Let $\tilde{N} \sim \Poi(np)$ for $p \in (0,1)$ such that $p \gtrsim n^{-1}$. For an integer $\ell \ge 1$, suppose $g$ is a $2\ell$ times continuously differentiable function on $(0,\infty)$. For an integer $r$ such that $1 \le r < 2\ell$,
  \begin{align}
    \abs*{\Mean\bracket*{R_{2\ell-1}\paren*{\frac{\tilde{N}}{n};g,p}}} \lesssim \sup_{\xi > 0}\abs*{\xi^{r}g^{(2\ell)}(\xi)}\frac{p^{\ell-r}}{n^{\ell}}.
  \end{align}
\end{lemma}

\cref{lem:reminder-bound1} will be used for deriving the bounds on the bias. The next lemma gives a useful bound for analyzing the variance.
\begin{lemma}\label{lem:reminder-bound2}
  Let $\tilde{N} \sim \Poi(np)$ for $p \in (0,1)$ such that $p \gtrsim n^{-1}$. For an integer $\ell \ge 1$, suppose $g$ is a $\ell$ times continuously differentiable function on $(0,\infty)$. For an integer $r$ such that $1 \le r < 2\ell$,
  \begin{align}
    \Mean\bracket*{\paren*{R_{\ell-1}\paren*{\frac{\tilde{N}}{n};g,p}}^2} \lesssim \sup_{\xi > 0}\abs*{\xi^{r}\paren*{g^{(\ell)}(\xi)}^2}\frac{p^{\ell-r}}{n^{\ell}}.
  \end{align}
\end{lemma}

{\bfseries Proofs of \cref{lem:plugin-bias-0-1,lem:plugin-var-0-1,lem:plugin-bias-1-3/2,lem:plugin-var-1-3/2}.}
By taking advantage of these techniques, we prove \cref{lem:plugin-bias-0-1,lem:plugin-var-0-1,lem:plugin-bias-1-3/2,lem:plugin-var-1-3/2}.
\begin{proof}[Proof of \cref{lem:plugin-bias-0-1}]
  Application of the Taylor theorem yields
  \begin{align}
      & H_{4,\Delta,\delta}[\phi]\paren*{\frac{\tilde{N}}{n}} - \phi(p) \\
      =& \begin{multlined}[t]
       \phi^{(1)}(p)\paren*{\frac{\tilde{N}}{n} - p} + \frac{\phi^{(2)}(p)}{2}\paren*{\frac{\tilde{N}}{n} - p}^2 \\ + \frac{\phi^{(3)}(p)}{6}\paren*{\frac{\tilde{N}}{n} - p}^3 + R_3\paren*{\frac{\tilde{N}}{n};H_{4,\Delta,\delta}[\phi],p}.
      \end{multlined}
  \end{align}
  For $X \sim \Poi(\lambda)$, $\Mean[(X-\lambda)] = 0$, $\Mean[(X-\lambda)^2] = \lambda$, and $\Mean[(X-\lambda)^3] = \lambda$. Thus, we have
  \begin{align}
      & \Mean\bracket*{H_{4,\Delta,\delta}[\phi]\paren*{\frac{\tilde{N}}{n}} - \phi(p)} \\
      =& \frac{p\phi^{(2)}(p)}{2n} + \frac{p\phi^{(3)}(p)}{6n^2} + \Mean\bracket*{R_3\paren*{\frac{\tilde{N}}{n};H_{4,\Delta,\delta}[\phi],p}}.
  \end{align}
  Again, using the Taylor theorem, we have
  \begin{align}
      & \frac{p\phi^{(2)}(p)}{2n} - \frac{\tilde{N}}{2n^2}H^{(2)}_{4,\Delta,\delta}[\phi]\paren*{\frac{\tilde{N}}{n}} \\
      =& \begin{multlined}[t]
        \frac{\phi^{(2)}(p)+p\phi^{(3)}(p)}{2n}\paren*{p - \frac{\tilde{N}}{n}} \\ - \frac{1}{2n}R_1\paren*{\frac{\tilde{N}}{n};x \to xH^{(2)}_{4,\Delta,\delta}[\phi](x),p}.
      \end{multlined}
  \end{align}
  Hence,
  \begin{align}
    & \Mean\bracket*{\frac{p\phi^{(2)}(p)}{2n} - \frac{\tilde{N}}{2n^2}H^{(2)}_{4,\Delta,\delta}[\phi]\paren*{\frac{\tilde{N}}{n}}} \\
    =& -\frac{1}{2n}\Mean\bracket*{R_1\paren*{\frac{\tilde{N}}{n};x \to xH^{(2)}_{4,\Delta,\delta}[\phi](x),p}}.
  \end{align}
  Consequentially, we have
  \begin{align}
      & \Bias\bracket*{\bar\phi_{2,4,\Delta,\delta}\paren*{\frac{\tilde{N}}{n}} - \phi(p)} \\
      =& \begin{multlined}[t]
        \abs[\Bigg]{\frac{p\phi^{(3)}(p)}{6n^2} + \Mean\bracket*{R_3\paren*{\frac{\tilde{N}}{n};H_{4,\Delta,\delta}[\phi],p}} \\ - \frac{1}{2n}\Mean\bracket*{R_1\paren*{\frac{\tilde{N}}{n};x \to xH^{(2)}_{4,\Delta,\delta}[\phi](x),p}} }
      \end{multlined} \\
      \le& \begin{multlined}[t]
        \frac{p\abs*{\phi^{(3)}(p)}}{6n^2} \abs*{\Mean\bracket*{R_3\paren*{\frac{\tilde{N}}{n};H_{4,\Delta,\delta}[\phi],p}}} \\ + \frac{1}{2n}\abs*{\Mean\bracket*{R_1\paren*{\frac{\tilde{N}}{n};x \to xH^{(2)}_{4,\Delta,\delta}[\phi](x),p}} }.
      \end{multlined}
  \end{align}

  From \cref{lem:lower-div}, we have
  \begin{align}
    \frac{p\abs*{\phi^{(3)}(p)}}{6n^2} \le \frac{W_3p^{\alpha-2} + c_3p}{6n^2} \lesssim \frac{\Delta^{\alpha-2}}{n^2}.
  \end{align}
  Noting that the second derivative of the function $x \to xH^{(2)}_{4,\Delta,\delta}[\phi](x)$ is $2H^{(3)}_{4,\Delta,\delta}[\phi](x) + xH^{(4)}_{4,\Delta,\delta}[\phi](x)$, by \cref{lem:hermite-bound,lem:reminder-bound1}, we have
  \begin{align}
    \abs*{\Mean\bracket*{R_3\paren*{\frac{\tilde{N}}{n};H_{4,\Delta,\delta}[\phi],p}}} \lesssim& \frac{\Delta^{\alpha-2}}{n^2}, \\
    \frac{1}{2n}\abs*{\Mean\bracket*{R_1\paren*{\frac{\tilde{N}}{n};x \to xH^{(2)}_{4,\Delta,\delta}[\phi](x),p}} } \lesssim& \frac{\Delta^{\alpha-2}}{n^2}.
  \end{align}
  The arbitrariness of $\delta > 0$ gives the desired claim.
\end{proof}
\begin{proof}[Proof of \cref{lem:plugin-var-0-1}]
  Since
  \begin{align}
    \Var\bracket*{\bar\phi_{2,4,\Delta,\delta}\paren*{\frac{\tilde{N}}{n}}} \le \Mean\bracket*{\paren*{\bar\phi_{2,4,\Delta,\delta}\paren*{\frac{\tilde{N}}{n}} - \phi_2(p)}^2},
  \end{align}
  application of the Taylor theorem and triangle inequality gives
  \begin{multline}
    \frac{1}{6}\Var\bracket*{\bar\phi_{2,4,\Delta,\delta}\paren*{\frac{\tilde{N}}{n}}} \le \\
      \paren*{\phi^{(1)}(p)}^2\Mean\bracket*{\paren*{\frac{\tilde{N}}{n} - p}^2} + \frac{\paren*{\phi^{(2)}(p)}^2}{4}\Mean\bracket*{\paren*{\frac{\tilde{N}}{n} - p}^4} \\ + \frac{\paren*{\phi^{(3)}(p)}^2}{36}\Mean\bracket*{\paren*{\frac{\tilde{N}}{n} - p}^6} \\ + \Mean\bracket*{\paren*{R_3\paren*{\frac{\tilde{N}}{n};H_{4,\Delta,\delta}[\phi],p}}^2} \\ + \frac{\paren*{\phi^{(2)}(p)+p\phi^{(3)}(p)}^2}{4n^2}\Mean\bracket*{\paren*{\frac{\tilde{N}}{n} - p}^2} \\ + \frac{1}{4n^2}\Mean\bracket*{\paren*{R_1\paren*{\frac{\tilde{N}}{n};x \to xH^{(2)}_{4,\Delta,\delta}[\phi](x),p}}^2 }.
  \end{multline}
  Hence,
  \begin{multline}
    \frac{1}{6}\Var\bracket*{\bar\phi_{2,4,\Delta,\delta}\paren*{\frac{\tilde{N}}{n}}} \le \\
    \frac{p\paren*{\phi^{(1)}(p)}^2}{n} + \frac{(3np+1)p\paren*{\phi^{(2)}(p)}^2}{4n^3} \\ + \frac{(15n^2p^2+25np+1)p\paren*{\phi^{(3)}(p)}^2}{36n^5} \\ + \Mean\bracket*{\paren*{R_3\paren*{\frac{\tilde{N}}{n};H_{4,\Delta,\delta}[\phi],p}}^2} \\ + \frac{p\paren*{\phi^{(2)}(p)+p\phi^{(3)}(p)}^2}{4n^3} \\ + \frac{1}{4n^2}\Mean\bracket*{\paren*{R_1\paren*{\frac{\tilde{N}}{n};x \to xH^{(2)}_{4,\Delta,\delta}[\phi](x),p}}^2 },
  \end{multline}
  where we use $\Mean[(X-\lambda)^6]=15\lambda^3+25\lambda^2+\lambda$ for $X \sim \Poi(\lambda)$.

  From \cref{lem:lower-div}, we have
  \begin{align}
     & \frac{(3np+1)p\paren*{\phi^{(2)}(p)}^2}{4n^3} \\ \le& \frac{(3np+1)p\paren*{W_2p^{\alpha-2}+c_2}^2}{4n^3} \lesssim \frac{p^{2\alpha-2}}{n^2} + \frac{p}{n^2}, \\
     & \frac{(15n^2p^2+25np+1)p\paren*{\phi^{(3)}(p)}^2}{36n^5} \\ \le& \frac{(15n^2p^2+25np+1)p\paren*{W_3p^{\alpha-3} + c_3}^2}{36n^5} \lesssim \frac{p^{2\alpha-3}}{n^3} + \frac{p}{n^3}, \\
     & \frac{p\paren*{\phi^{(2)}(p)+p\phi^{(3)}(p)}^2}{4n^3} \\ \le& \frac{p\paren*{(W_2+W_3)p^{\alpha-2}+c_2+pc_3}^2}{4n^3} \lesssim \frac{p^{2\alpha-3}}{n^3} + \frac{p}{n^3}.
  \end{align}
  By \cref{lem:hermite-bound,lem:reminder-bound2}, we have
  \begin{align}
     \Mean\bracket*{\paren*{R_3\paren*{\frac{\tilde{N}}{n};H_{4,\Delta,\delta}[\phi],p}}^2} \lesssim& \frac{\Delta^{2\alpha-4}}{n^4}, \\
     \frac{1}{4n^2}\Mean\bracket*{\paren*{R_1\paren*{\frac{\tilde{N}}{n};x \to xH^{(2)}_{4,\Delta,\delta}[\phi](x),p}}^2 } \lesssim& \frac{\Delta^{2\alpha-4}}{n^4}.
  \end{align}
  
  If $\alpha \in (0,1)$, from \cref{lem:lower-div}, we have
  \begin{align}
      \frac{p\paren*{\phi^{(1)}(p)}^2}{n} \le& \frac{p\paren*{W_1p^{\alpha-1}+c_1}^2}{n} \lesssim \frac{p^{2\alpha-1}}{n} + \frac{p}{n}.
  \end{align}
  If $\alpha = 1$, from \cref{lem:div-speed-holder3}, we have
  \begin{align}
      \frac{p\paren*{\phi^{(1)}(p)}^2}{n} \le \frac{p\paren*{W_1\ln(1/p) + c_1}^2}{n} \lesssim \frac{p\ln^2p}{n} + \frac{1}{n}
  \end{align}
  The arbitrariness of $\delta > 0$ gives the desired claim.
\end{proof}
\begin{proof}[Proof of \cref{lem:plugin-bias-1-3/2}]
  In the same manner of the proof of \cref{lem:plugin-bias-0-1}, we have
  \begin{multline}
      \Mean\bracket*{H_{6,\Delta,\delta}[\phi]\paren*{\frac{\tilde{N}}{n}} - \phi(p)} = \\ \frac{p\phi^{(2)}(p)}{2n} + \frac{p\phi^{(3)}(p)}{6n^2} + \frac{(3np+1)p\phi^{(4)}(p)}{24n^3} \\ + \frac{(10np + 1)p\phi^{(5)}(p)}{120n^4} + \Mean\bracket*{R_5\paren*{\frac{\tilde{N}}{n};H_{6,\Delta,\delta}[\phi],p}},
  \end{multline}
  where we use $\Mean[(X-\lambda)^4] = 3\lambda^2+\lambda$ and $\Mean[(X-\lambda)^5]=10\lambda^2+\lambda$ for $X \sim \Poi(\lambda)$. Besides, we have
  \begin{multline}
    \Mean\bracket*{\frac{p\phi^{(2)}(p)}{2n} - \frac{\tilde{N}}{2n^2}H^{(2)}_{6,\Delta,\delta}[\phi]\paren*{\frac{\tilde{N}}{n}}} = \\ - \frac{2p\phi^{(3)}(p)+p^2\phi^{(4)}(p)}{4n^2} - \frac{3p\phi^{(4)}(p)+p^2\phi^{(5)}(p)}{12n^3} \\ - \frac{1}{2n}\Mean\bracket*{R_3\paren*{\frac{\tilde{N}}{n};x \to xH^{(2)}_{6,\Delta,\delta}[\phi](x),p}},
  \end{multline}
  \begin{multline}
      \Mean\bracket*{\frac{p\phi^{(3)}(p)}{3n^2} - \frac{\tilde{N}}{3n^3}H^{(3)}_{6,\Delta,\delta}[\phi]\paren*{\frac{\tilde{N}}{n}}} = \\ - \frac{1}{3n^2}\Mean\bracket*{R_1\paren*{\frac{\tilde{N}}{n};x \to xH^{(3)}_{6,\Delta,\delta}[\phi](x),p}},
  \end{multline}
  \begin{multline}
    \Mean\bracket*{\frac{5p\phi^{(4)}(p)}{24n^3} - \frac{5\tilde{N}}{24n^4}H^{(4)}_{6,\Delta,\delta}[\phi]\paren*{\frac{\tilde{N}}{n}}} = \\ - \frac{5}{24n^3}\Mean\bracket*{R_1\paren*{\frac{\tilde{N}}{n};x \to xH^{(4)}_{6,\Delta,\delta}[\phi](x),p}},
  \end{multline}
  and
  \begin{multline}
    \Mean\bracket*{\frac{p^2\phi^{(4)}(p)}{8n^2} - \frac{\tilde{N}^2}{8n^4}H^{(4)}_{6,\Delta,\delta}[\phi]\paren*{\frac{\tilde{N}}{n}}}
     = \\ - \frac{1}{8n^2}\Mean\bracket*{R_1\paren*{\frac{\tilde{N}}{n};x \to x^2H^{(4)}_{6,\Delta,\delta}[\phi](x),p}}.
  \end{multline}
  Hence,
  \begin{multline}
     \Bias\bracket*{\bar\phi_{4,6,\Delta,\delta}\paren*{\frac{\tilde{N}}{n}} - \phi(p)}
     \le \\ \frac{p\abs*{\phi^{(5)}(p)}}{120n^4} + \abs*{\Mean\bracket*{R_5\paren*{\frac{\tilde{N}}{n};H_{6,\Delta,\delta}[\phi],p}}} \\ + \frac{1}{2n}\abs*{\Mean\bracket*{R_3\paren*{\frac{\tilde{N}}{n};x \to xH^{(2)}_{6,\Delta,\delta}[\phi](x),p}}} \\ + \frac{1}{3n^2}\abs*{\Mean\bracket*{R_1\paren*{\frac{\tilde{N}}{n};x \to xH^{(3)}_{6,\Delta,\delta}[\phi](x),p}}} \\ + \frac{5}{24n^3}\abs*{\Mean\bracket*{R_1\paren*{\frac{\tilde{N}}{n};x \to xH^{(4)}_{6,\Delta,\delta}[\phi](x),p}}} \\ + \frac{1}{8n^2}\abs*{\Mean\bracket*{R_1\paren*{\frac{\tilde{N}}{n};x \to x^2H^{(4)}_{6,\Delta,\delta}[\phi](x),p}}}.
  \end{multline}

  From \cref{lem:lower-div}, we have
  \begin{align}
    \frac{p\abs*{\phi^{(5)}(p)}}{120n^4} \le \frac{W_5p^{\alpha-4} + c_5p}{120n^4} \lesssim \frac{\Delta^{\alpha-4}}{n^4} \lesssim \frac{\Delta^{\alpha-3}}{n^3}.
  \end{align}
  The fourth derivative of the function $x \to xH^{(2)}_{6,\Delta,\delta}[\phi](x)$ is $4H^{(5)}_{6,\Delta,\delta}[\phi](x) + xH^{(6)}_{6,\Delta,\delta}[\phi](x)$. The second derivatives of the functions $x \to xH^{(3)}_{6,\Delta,\delta}[\phi](x)$, $x \to xH^{(4)}_{6,\Delta,\delta}[\phi](x)$, and $x \to x^2H^{(4)}_{6,\Delta,\delta}[\phi](x)$ are $2H^{(4)}_{6,\Delta,\delta}[\phi](x)+xH^{(5)}_{6,\Delta,\delta}[\phi](x)$, $2H^{(5)}_{6,\Delta,\delta}[\phi](x)+xH^{(6)}_{6,\Delta,\delta}[\phi](x)$, and $2H^{(4)}_{6,\Delta,\delta}[\phi](x)+4xH^{(5)}_{6,\Delta,\delta}[\phi](x)+x^2H^{(6)}_{6,\Delta,\delta}[\phi](x)$, respectively. Thus, application of \cref{lem:hermite-bound,lem:reminder-bound1} yields
  \begin{align}
     \abs*{\Mean\bracket*{R_5\paren*{\frac{\tilde{N}}{n};H_{6,\Delta,\delta}[\phi],p}}} \lesssim& \frac{\Delta^{\alpha-3}}{n^3}, \\
     \frac{1}{2n}\abs*{\Mean\bracket*{R_3\paren*{\frac{\tilde{N}}{n};x \to xH^{(2)}_{6,\Delta,\delta}[\phi](x),p}}} \lesssim& \frac{\Delta^{\alpha-3}}{n^3}, \\
     \frac{1}{3n^2}\abs*{\Mean\bracket*{R_1\paren*{\frac{\tilde{N}}{n};x \to xH^{(3)}_{6,\Delta,\delta}[\phi](x),p}}} \lesssim& \frac{\Delta^{\alpha-3}}{n^3}, \\
     \frac{5}{24n^3}\abs*{\Mean\bracket*{R_1\paren*{\frac{\tilde{N}}{n};x \to xH^{(4)}_{6,\Delta,\delta}[\phi](x),p}}} \lesssim& \frac{\Delta^{\alpha-4}}{n^4}, \\
     \frac{1}{8n^2}\abs*{\Mean\bracket*{R_1\paren*{\frac{\tilde{N}}{n};x \to x^2H^{(4)}_{6,\Delta,\delta}[\phi](x),p}}} \lesssim& \frac{\Delta^{\alpha-3}}{n^3}.
  \end{align}
  The arbitrariness of $\delta > 0$ gives the desired claim.
\end{proof}
\begin{proof}[Proof of \cref{lem:plugin-var-1-3/2}]
 In the same manner of the proof of \cref{lem:plugin-var-0-1}, application of the Taylor theorem and triangle inequality yields
 \begin{multline}
    \frac{1}{5}\Var\bracket*{\bar\phi_{4,6,\Delta,\delta}\paren*{\frac{\tilde{N}}{n}}} \le \\ \Mean\bracket*{\paren*{H_{6,\Delta,\delta}[\phi]\paren*{\frac{\tilde{N}}{n}} - \phi(p)}^2} \\ + \frac{1}{4n^2}\Mean\bracket*{\paren*{R_0\paren*{\frac{\tilde{N}}{n};x \to xH^{(2)}_{6,\Delta,\delta}[\phi](x),p}}^2} \\+ \frac{1}{9n^4}\Mean\bracket*{\paren*{R_0\paren*{\frac{\tilde{N}}{n};x \to xH^{(3)}_{6,\Delta,\delta}[\phi](x),p}}^2} \\ + \frac{25}{576n^6}\Mean\bracket*{\paren*{R_0\paren*{\frac{\tilde{N}}{n};x \to xH^{(4)}_{6,\Delta,\delta}[\phi](x),p}}^2} \\+ \frac{1}{64n^4}\Mean\bracket*{\paren*{R_0\paren*{\frac{\tilde{N}}{n};x \to x^2H^{(4)}_{6,\Delta,\delta}[\phi](x),p}}^2}.
 \end{multline}

 From \cref{lem:hermite-bound}, we have $H_{6,\Delta,\delta}[\phi]\paren*{x} \lesssim 1$ for any $x > 0$. By the Taylor theorem, we have
 \begin{align}
    &\Mean\bracket*{\paren*{H_{6,\Delta,\delta}[\phi]\paren*{\frac{\tilde{N}}{n}} - \phi(p)}^2} \\
    \le& \Mean\bracket*{\paren*{\frac{\tilde{N}}{n} - p}^2}\sup_{x > 0}{\paren*{H_{6,\Delta,\delta}[\phi](x)}^2} \lesssim \frac{p}{n}.
 \end{align}
 By \cref{lem:hermite-bound,lem:reminder-bound2}, we have
 \begin{align}
    \frac{1}{4n^2}\Mean\bracket*{\paren*{R_0\paren*{\frac{\tilde{N}}{n};x \to xH^{(2)}_{6,\Delta,\delta}[\phi](x),p}}^2} \lesssim& \frac{\Delta^{2\alpha-3}}{n^3}, \\
    \frac{1}{9n^4}\Mean\bracket*{\paren*{R_0\paren*{\frac{\tilde{N}}{n};x \to xH^{(3)}_{6,\Delta,\delta}[\phi](x),p}}^2} \lesssim& \frac{\Delta^{2\alpha-5}}{n^5}, \\
    \frac{25}{576n^6}\Mean\bracket*{\paren*{R_0\paren*{\frac{\tilde{N}}{n};x \to xH^{(4)}_{6,\Delta,\delta}[\phi](x),p}}^2} \lesssim& \frac{\Delta^{2\alpha-7}}{n^7}, \\
    \frac{1}{64n^4}\Mean\bracket*{\paren*{R_0\paren*{\frac{\tilde{N}}{n};x \to x^2H^{(4)}_{6,\Delta,\delta}[\phi](x),p}}^2} \lesssim& \frac{\Delta^{2\alpha-5}}{n^5}.
 \end{align}
 Noting that $\frac{\Delta^{2\alpha-7}}{n^7} \lesssim \frac{\Delta^{2\alpha-5}}{n^5} \lesssim \frac{\Delta^{2\alpha-3}}{n^3} \lesssim \frac{1}{n^{2\alpha}}$, the arbitrariness of $\delta > 0$ gives the desired claim.
\end{proof}

\subsubsection{Overall Bias and Variance}
Combining the analyses above, we prove \cref{thm:upper-bound1}.
\begin{proof}[Proof of \cref{thm:upper-bound1}]
  Set $L = \floor{C_1 \ln n}$ and $\Delta_{n,k} = C_2 \ln n$ where $C_1$ and $C_2$ are positive universal constants. We derive bounds on the terms in \cref{lem:upper-ind-bias,lem:upper-ind-var}.

  From \cref{lem:plugin-o1-2,lem:plugin-o1-4}, we have
  \begin{multline}
    \paren*{\sum_{i=1}^k (e/4)^{\Delta_{n,k}}\ind{np_i \le \Delta_{n,k}}\Bias\bracket*{\phi_{\mathrm{plugin}}(\tilde{N}_i) - \phi(p_i)}}^2 \\ \lesssim k^2n^{-2C_2\ln(4/e)},  
  \end{multline}
  and
  \begin{multline}
      \sum_{i=1}^k (e/4)^{\Delta_{n,k}}\ind{np_i \le \Delta_{n,k}}\paren[\Bigg]{\Var\bracket*{\phi_{\mathrm{plugin}}(\tilde{N}_i)} \\ +2\Bias\bracket*{\phi_{\mathrm{plugin}}(\tilde{N}_i) - \phi(p_i)}^2} \lesssim kn^{-C_2\ln(4/e)}.
  \end{multline}
  From \cref{lem:poly-o1}, we have
  \begin{multline}
    \paren*{\sum_{i=1}^k e^{-\Delta_{n,k}/8}\ind{np_i > 4\Delta_{n,k}}\Bias\bracket*{\phi_{\mathrm{plugin}}(\tilde{N}_i) - \phi(p_i)}}^2 \\ \lesssim k^2n^{-C_2/4},
  \end{multline}
  and
  \begin{multline}
    \sum_{i=1}^k e^{-\Delta_{n,k}/8}\ind{np_i > 4\Delta_{n,k}}\paren[\Bigg]{\Var\bracket*{\phi_{\mathrm{plugin}}(\tilde{N}_i)} \\ +2\Bias\bracket*{\phi_{\mathrm{plugin}}(\tilde{N}_i) - \phi(p_i)}^2} \lesssim kn^{-C_2/8} .
  \end{multline}
  Since $\ln(4/e) \ge 1/8$, as long as $C_2 > 8\alpha$, we have
  \begin{align}
    k^2n^{-2C_2\ln(4/e)} \lesssim kn^{-C_2\ln(4/e)} \lesssim \frac{k^2}{(n\ln n)^{2\alpha}}, 
  \end{align}
  and
  \begin{align}
     k^2n^{-C_2/4} \lesssim kn^{-C_2/8} \lesssim \frac{k^2}{(n\ln n)^{2\alpha}}.
  \end{align}

  If $C_2^3C_1 \le 1/2$, by \cref{lem:poly-bias,lem:poly-var}, we have
  \begin{multline}
     \paren*{\sum_{i=1}^k \ind{np_i \le 4\Delta_{n,k}}\Bias\bracket*{\phi_{\mathrm{poly}}(\tilde{N}_i) - \phi(p_i)}}^2 \lesssim  \\k^2\paren*{E_{\floor{C_1\ln n}}(\phi,[0,4C_2\ln n/n])}^2 \\+ \frac{k^2C_1C_2^3\ln^4n}{n^{4-6C_1\ln 2 + 4\sqrt{C_1C_2}\ln(2e)}},
  \end{multline}
  and
  \begin{multline}
     \sum_{i=1}^k \ind{np_i \le 4\Delta_{n,k}}\paren[\Bigg]{\Var\bracket*{\phi_{\mathrm{poly}}(\tilde{N}_i)} \\ + \Bias\bracket*{\phi_{\mathrm{poly}}(\tilde{N}_i) - \phi(p_i)}^2} \lesssim  \\k\paren*{E_{\floor{C_1\ln n}}(\phi,[0,4C_2\ln n/n])}^2 \\ + \frac{kC_1C_2^3\ln^4n}{n^{4-6C_1\ln 2 + 4\sqrt{C_1C_2}\ln(2e)}},
  \end{multline}
  From \cref{thm:upper-poly-approx1}, we have
  \begin{align}
    & k\paren*{E_{\floor{C_1\ln n}}(\phi,[0,4C_2\ln n/n])}^2 \\ \lesssim& k^2\paren*{E_{\floor{C_1\ln n}}(\phi,[0,4C_2\ln n/n])}^2 \\ \lesssim& \frac{k^2}{(n\ln n)^{2\alpha}}.
  \end{align}
  As long as $2 - 3C_1 \ln 2 - 2\sqrt{C_1C_2}\ln(2e) > \alpha$, we have
  \begin{align}
      \frac{k^2C_1C_2^3\ln^4n}{n^{4-6C_1\ln 2 + 4\sqrt{C_1C_2}\ln(2e)}} \lesssim \frac{k^2}{(n\ln n)^{2\alpha}}.
  \end{align}

  If $\alpha \in (0,1]$, by \cref{lem:plugin-bias-0-1}, we have
  \begin{multline}
     \paren*{\sum_{i=1}^k \ind{np_i > \Delta_{n,k}}\Bias\bracket*{\phi_{\mathrm{plugin}}(\tilde{N}_i) - \phi(p_i)}}^2 \\ \lesssim \frac{k^2}{n^{2\alpha}\ln^{4-2\alpha}n}.
  \end{multline}
  Besides, by \cref{lem:plugin-var-0-1}, we have
  \begin{multline}
     \sum_{i=1}^k \ind{np_i > \Delta_{n,k}}\paren[\Bigg]{\Var\bracket*{\phi_{\mathrm{plugin}}(\tilde{N}_i)} \\ + \Bias\bracket*{\phi_{\mathrm{plugin}}(\tilde{N}_i) - \phi(p_i)}^2} \lesssim \\ \frac{\sum_{i=1}^k\ind{np_i > \Delta_{n,k}}p_i^{2\alpha-1}}{n} + \frac{k}{n^{2\alpha}\ln^{4-2\alpha}n} + \frac{1}{n},
  \end{multline}
  Since $4 - 2\alpha \ge 2\alpha$ for $\alpha \in (0,1]$, we have
  \begin{align}
    \frac{k}{n^{2\alpha}\ln^{4-2\alpha}n} \lesssim \frac{k^2}{n^{2\alpha}\ln^{4-2\alpha}n} \lesssim \frac{k^2}{(n\ln n)^{2\alpha}}.
  \end{align}
  If $\alpha \in [1/2,1)$, by \cref{lem:bound-sum-alpha}, we have
  \begin{align}
      \frac{\sum_{i=1}^k\ind{np_i > \Delta_{n,k}}p_i^{2\alpha-1}}{n}  \lesssim \frac{k^{2-2\alpha}}{n}.
  \end{align}
  If $\alpha = 1$, by \cref{lem:bound-sum-log}, we have
  \begin{align}
      \frac{\sum_{i=1}^k\ind{np_i > \Delta_{n,k}}p_i\ln^2p_i}{n}  \lesssim \frac{\ln^2k}{n}.
  \end{align}
  If $\alpha \in (0,1/2)$ and $k \gtrsim \ln^{4\alpha-1}n$, we have
  \begin{align}
      \frac{\sum_{i=1}^k\ind{np_i > \Delta_{n,k}}p_i^{2\alpha-1}}{n}  \lesssim \frac{k}{n^{2\alpha}\ln^{1-2\alpha}n} \lesssim \frac{k^2}{(n\ln n)^{2\alpha}}.
  \end{align}

  If $\alpha \in (1,3/2)$, by \cref{lem:plugin-bias-1-3/2}, we have
  \begin{multline}
     \paren*{\sum_{i=1}^k \ind{np_i > \Delta_{n,k}}\Bias\bracket*{\phi_{\mathrm{plugin}}(\tilde{N}_i) - \phi(p_i)}}^2 \\ \lesssim \frac{k^2}{n^{2\alpha}\ln^{6-3\alpha}n}.
  \end{multline}
  Besides, by \cref{lem:plugin-var-0-1}, we have
  \begin{multline}
     \sum_{i=1}^k \ind{np_i > \Delta_{n,k}}\paren[\Bigg]{\Var\bracket*{\phi_{\mathrm{plugin}}(\tilde{N}_i)} \\ + \Bias\bracket*{\phi_{\mathrm{plugin}}(\tilde{N}_i) - \phi(p_i)}^2} \lesssim \frac{k}{n^{2\alpha}} + \frac{1}{n},
  \end{multline}
  Since $6 - 3\alpha \ge 2\alpha$ for $\alpha \in (1,3/2)$, we have
  \begin{align}
      \frac{k^2}{n^{2\alpha}\ln^{6-3\alpha}n} \lesssim \frac{k^2}{(n\ln n)^{2\alpha}}.
  \end{align}
  Under the condition $k \lesssim (n\ln n)^\alpha$, we have
  \begin{align}
      \frac{k}{n^{2\alpha}} \lesssim \frac{\ln^\alpha n}{n^\alpha} \lesssim \frac{1}{n}.
  \end{align}

  In summary, we proved the desired bounds under the conditions $C_2 > 8\alpha$, $C_2^3C_1 \le 1/2$, and $2 - 3C_1 \ln 2 - 2\sqrt{C_1C_2}\ln(2e) > \alpha$. Note that these conditions hold if we take sufficiently large $C_2$ and sufficiently small $C_1$.
\end{proof}

\subsection{Analyses for $\alpha \in [3/2,2]$}\label{sec:upper2}

To prove \cref{thm:upper-bound2}, we use the bias-variance decomposition as well as the case $\alpha \in (0,3/2)$, which gives
\begin{align}
  &\Mean\bracket*{\paren*{\hat\theta_{\mathrm{plugin}}(N) - \theta(P)}^2} \\ =& \Bias\bracket*{\hat\theta_{\mathrm{plugin}}(N) - \theta(P)}^2 + \Var\bracket*{\hat\theta_{\mathrm{plugin}}(N)}, \label{eq:plugin-bias-var-decomp}
\end{align}
where $N \sim \Mul(n,P)$. Hence, we obtain the estimation error of the plugin estimator by deriving the bias and variance.

For any $\alpha \in [3/2,2]$, the variance term is easily proved by using the fact that $\phi$ is Lipschitz continuous:
\begin{theorem}\label{thm:variance-lipschitz}
  If $\phi$ is Lipschitz continuous, then
  \begin{align}
    \Var\bracket*{\sum_{i=1}^n\phi\paren*{\frac{N_i}{n}}} \lesssim \frac{1}{n}.
  \end{align}
\end{theorem}
The proof of \cref{thm:variance-lipschitz} is accomplished by applying the concentration result of the bounded difference:
\begin{theorem}[see e.g., \autocite{Boucheron2013ConcentrationIndependence}]
  Suppose that $X_1,...,X_n$ are independent random variables on $\dom{X}$. For a function $f:\dom{X}^n\to\RealSet$, suppose there exist universal constants $c_1,...,c_n$ such that for any $i \in [n]$,
  \begin{align}
    \sup_{x_1,...,x_n,x'_i}\abs*{f(x_1,...,x_i,...,x_n) - f(x_1,...,x'_i,...,x_n)} \le c_i.
  \end{align}
  Then,
  \begin{align}
    \Var\bracket*{f(X_1,...,X_n)} \le \frac{1}{4}\sum_{i=1}^nc_i^2.
  \end{align}
\end{theorem}
\begin{proof}[Proof of \cref{thm:variance-lipschitz}]
  Suppose a sample $X_j$ is changed from $X_j=i$ to $X_j=i'$. Then, the change of the histogram is $N_i \to N_i-1$ and $N_{i'} \to N_{i'}+1$. Hence, for $f(S_n) = \sum_{i=1}^k\phi(N_i/n)$, we have
  \begin{multline}
    \sup_{N}\abs*{\phi\paren*{\frac{N_i}{n}}+\phi\paren*{\frac{N_{i'}}{n}} - \phi\paren*{\frac{N_i-1}{n}} - \phi\paren*{\frac{N_{i'}+1}{n}}} \\ \le \norm*{\phi}_{C^{H,1}}\frac{2}{n},
  \end{multline}
  where we use the Lipschitz continuity of $\phi$. Hence,
  \begin{align}
    \Var\bracket*{\sum_{i=1}^k\paren*{\frac{N_i}{n}}} \le \frac{n}{4}\paren*{\norm*{\phi}_{C^{H,1}}\frac{2}{n}}^2 = \frac{\norm*{\phi}_{C^{H,1}}^2}{n}.
  \end{align}
\end{proof}

If $\alpha = 2$, the bias is proved immediately from the Lipschitz continuity of $\phi^{(1)}$;
\begin{theorem}\label{thm:bias-lipschitz}
  If $\phi^{(1)}$ is Lipschitz continuous, then
  \begin{align}
    \Bias\bracket*{\sum_{i=1}^n\phi\paren*{\frac{N_i}{n}} - \theta(P)} \lesssim \frac{1}{n}.
  \end{align}
\end{theorem}
The proof of this theorem is obtained by simply applying the Taylor theorem.
\begin{proof}[Proof of \cref{thm:bias-lipschitz}]
  Application of the Taylor theorem yields there exists $\xi_1,...,\xi_k$ such that
  \begin{align}
    & \Bias\bracket*{\sum_{i=1}^k\paren*{\frac{N_i}{n}} - \theta(P)} \\
    =& \abs*{\Mean\bracket*{\sum_{i=1}^k\paren*{\phi^{(1)}(p_i)\paren*{\frac{N_i}{n}-p_i} + \frac{\phi^{(2)}(\xi_i)}{2}\paren*{\frac{N_i}{n}-p_i}^2}}} \\
    \le& \Mean\bracket*{\sum_{i=1}^k\frac{\abs*{\phi^{(2)}(\xi_i)}}{2}\paren*{\frac{N_i}{n}-p_i}^2}.
  \end{align}
  From the Lipschitz continuity, we have $\sup_{p \in (0,1)}\abs*{\phi^{(2)}(p)} \le \norm{\phi^{(1)}}_{C^0,1}$. Hence,
  \begin{align}
    & \Bias\bracket*{\sum_{i=1}^k\paren*{\frac{N_i}{n}} - \theta(P)} \\
    \le& \frac{\norm{\phi^{(1)}}_{C^0,1}}{2}\Mean\bracket*{\sum_{i=1}^k\paren*{\frac{N_i}{n}-p_i}^2} \\
    =& \frac{\norm{\phi^{(1)}}_{C^0,1}}{2}\sum_{i=1}^k\frac{p_i(1-p_i)}{n} \\
    \le& \frac{\norm{\phi^{(1)}}_{C^0,1}}{2n}.
  \end{align}
\end{proof}

In contrast, derivation of a bias bound for $\alpha \in (3/2,2)$ is not trivial. \textcite{Jiao2017MaximumDistributions} analyzed the plugin estimator for $\phi(p)=p^\alpha$ and showed that the bias of the plugin estimator is the same as the Bernstein polynomial approximation error. That is, for $N \sim \Mul(n,P)$, we have
\begin{align}
 \Bias\bracket*{\hat\theta_{\mathrm{plugin}}(N) - \theta(P)} =& \abs*{\sum_{i=1}^kB_n[\phi](p_i) - \phi(p_i)}. \label{eq:plugin-bias-bernstein}
\end{align}
Consequently, by taking advantage of the results on the Bernstein polynomial approximation, such as \cref{thm:bernstein-error}, the bias is bounded above by the modulus of smoothness.

To obtain an upper bound on the bias of the plugin estimator for $\alpha \in (3/2,2)$, we derive the bound on the second order modulus of smoothness:
\begin{lemma}\label{lem:modulus-bound}
  Suppose $\phi:[0,1]\to\RealSet$ is a function whose second divergence speed is $p^{\alpha}$ for $\alpha \in [3/2,2)$, where $\phi^{(1)}(0) = 0$. Then, we have
  \begin{align}
    \omega^2(\phi,t) \lesssim t^{\alpha}.
  \end{align}
\end{lemma}
\begin{proof}
  From \cref{lem:mod-derivative}, $\omega^2(\phi,t) \lesssim t\omega^1(\phi^{(1)},t)$. By definition, $\omega^1(\phi^{(1)},t) \lesssim t^{\alpha-1}$ because of the H\"older continuity of $\phi^{(1)}$ from \cref{lem:div-speed-holder3}.
\end{proof}
By utilizing \cref{lem:modulus-bound}, we prove the bias.
\begin{theorem}\label{thm:bias-plugin}
  Suppose $\phi:[0,1]\to\RealSet$ is a function whose second divergence speed is $p^{\alpha}$ for $\alpha \in [3/2,2)$ such that $\phi(0)=0$. Then, we have
  \begin{align}
    \Bias\bracket*{\sum_i\phi\paren*{\frac{N_i}{n}} - \theta(P)} \lesssim \frac{1}{n^{\alpha-1}}.
  \end{align}
\end{theorem}
\begin{proof}[Proof of \cref{thm:bias-plugin}]
  We divide the alphabets into two cases; $p_i \le 1/n$ and $p_i > 1/n$.

  {\bfseries Case $p_i \le 1/n$.} Since $\phi(0) = 0$, we have from the Taylor theorem that there exists $\xi_i$ between $\frac{N_i}{n}$ and $p_i$ such that
  \begin{align}
    & \abs*{\Mean\bracket*{\sum_{i : p_i \le 1/n}\paren*{\phi\paren*{\frac{N_i}{n}} - \phi(p_i)}}} \\
    \le& \begin{multlined}[t]
      \sum_{i : p_i \le 1/n}\p\cbrace*{N_i = 0}\abs*{\phi(p_i)} \\ + \abs*{\Mean\bracket*{\sum_{i : p_i \le 1/n}\frac{\phi^{(2)}(\xi_i)}{2}\paren*{\frac{N_i}{n} - p_i}^2 \middle| N_i > 0}}
    \end{multlined}\\
    \le& \begin{multlined}[t]
      \sum_{i : p_i \le 1/n}\p\cbrace*{N_i = 0}\abs*{\phi(p_i)} \\ + \Mean\bracket*{\sum_{i : p_i \le 1/n}\abs*{\frac{\phi^{(2)}(\xi_i)}{2}}\paren*{\frac{N_i}{n} - p_i}^2 \middle| N_i > 0} 
    \end{multlined}\\
    \le& \begin{multlined}[t]
      \sum_{i : p_i \le 1/n}\p\cbrace*{N_i = 0}\abs*{\phi(p_i)} \\ + \Mean\bracket*{\sum_{i : p_i \le 1/n}\frac{W_2\xi_i^{\alpha-2}+c_2}{2}\paren*{\frac{N_i}{n} - p_i}^2 \middle| N_i > 0}   
    \end{multlined}\\
    \le& \begin{multlined}[t]
      \sum_{i : p_i \le 1/n}\p\cbrace*{N_i = 0}\abs*{\phi(p_i)} \\ + \Mean\bracket*{\sum_{i : p_i \le 1/n}\frac{W_2p_i^{\alpha-2}+c_2}{2}\paren*{\frac{N_i}{n} - p_i}^2 \middle| N_i > 0}  
    \end{multlined}\\
    \lesssim& \sum_{i : p_i \le 1/n}\paren*{p_i^{\alpha-2}+1}\Mean\bracket*{\paren*{\frac{N_i}{n} - p_i}^2} \\
    \lesssim& \sum_{i : p_i \le 1/n}\paren*{p_i^{\alpha-2}+1}\frac{p_i}{n} \lesssim \frac{1}{n^{\alpha-1}}. \label{eq:bias-plugin1}
  \end{align}
  where we use $\abs*{\phi(p_i)} \lesssim p_i^\alpha \le p_i^2$ and $\frac{N_i}{n} \ge \frac{1}{n}$ if $N_i > 0$.

  {\bfseries Case $p_i > 1/n$.}
  Combining \cref{thm:bernstein-error,eq:plugin-bias-bernstein,lem:modulus-bound}, we have
  \begin{align}
    & \abs*{\Mean\bracket*{\sum_{i : p_i > 1/n}\paren*{\phi\paren*{\frac{N_i}{n}} - \phi(p_i)}}} \\
    \le& \sum_{i : p_i > 1/n}\Bias\bracket*{\phi\paren*{\frac{N_i}{n}} - \phi(p_i)} \\
    \lesssim& \sum_{i : p_i > 1/n}\frac{p_i^{\alpha/2}}{n^{\alpha/2}}.
  \end{align}
  Since $\sup_{P \in \dom{M}_k}\sum_{i : p_i > 1/n}p_i^{\alpha/2} \lesssim n^{1-\alpha/2}$, we have
  \begin{align}
    & \abs*{\Mean\bracket*{\sum_{i : p_i > 1/n}\paren*{\phi\paren*{\frac{N_i}{n}} - \phi(p_i)}}} \\
    \lesssim& \frac{n^{1-\alpha/2}}{n^{\alpha/2}} = \frac{1}{n^{\alpha-1}}. \label{eq:bias-plugin2}
  \end{align}
  Combining \cref{eq:bias-plugin1,eq:bias-plugin2}, we have
  \begin{align}
    & \Bias\bracket*{\sum_i\paren*{\phi\paren*{\frac{N_i}{n}} - \phi(p_i)}} \\
    \le& \begin{multlined}[t]
      \abs*{\Mean\bracket*{\sum_{i : p_i \le 1/n}\paren*{\phi\paren*{\frac{N_i}{n}} - \phi(p_i)}}} \\ + \abs*{\Mean\bracket*{\sum_{i : p_i > 1/n}\paren*{\phi\paren*{\frac{N_i}{n}} - \phi(p_i)}}}
    \end{multlined}\\
    \lesssim& \frac{1}{n^{\alpha-1}}.
  \end{align}
\end{proof}

By combining \cref{eq:plugin-bias-var-decomp,thm:variance-lipschitz,thm:bias-lipschitz,thm:bias-plugin}, we obtain the desired claim shown in \cref{thm:upper-bound2}:
\begin{proof}[Proof of \cref{thm:upper-bound2}]
 For $\alpha \ge 3/2$, by \cref{thm:bias-lipschitz,thm:bias-plugin},
 \begin{align}
     \Bias\bracket*{\sum_i\phi\paren*{\frac{N_i}{n}} - \theta(P)} \lesssim n^{-(\alpha-1)}\lor n^{-1} \lesssim n^{-1/2}.
 \end{align} 
 Substituting the bias and variance shown in \cref{thm:variance-lipschitz,thm:bias-lipschitz,thm:bias-plugin} into \cref{eq:plugin-bias-var-decomp} yields
 \begin{align}
     \sup_{P \in \dom{M}_k}\Mean\bracket*{\paren*{\hat\theta_{\mathrm{plugin}}(N) - \theta(P)}^2} \lesssim \paren*{\frac{1}{\sqrt{n}}}^2 + \frac{1}{n} \lesssim \frac{1}{n}.
 \end{align}
\end{proof}

\section{Lower Bound Analysis}\label{sec:lower-analysis}

In this section, we provide analyses of the lower bound on the minimax risk. Formally, we prove the following three theorems:
\begin{theorem}\label{thm:lower1}
  Suppose $\phi:[0,1]\to\RealSet$ is a function whose second divergence speed is $p^{\alpha}$ for $\alpha \in (0,2]$. If $\alpha \in (0,1]$, for any $n \ge 1$ and $k \ge 3$,
  \begin{align}
    R^*(n,k;\phi) \gtrsim \begin{dcases}
      \frac{\ln^2k}{n} & \textif \alpha = 1, \\
      \frac{k^{2-2\alpha}}{n} & \otherwise.
     \end{dcases}
  \end{align}
  Moreover, if $\alpha \in (1,2]$, there exists $K > 0$ such that for any $n \ge 1$ and any $k \ge K$,
  \begin{align}
      R^*(n,k;\phi) \gtrsim \frac{1}{n}.
  \end{align}
\end{theorem}
\begin{theorem}\label{thm:lower2}
 Suppose $\phi:[0,1]\to\RealSet$ is a function whose second divergence speed is $p^{\alpha}$ for $\alpha \in (0,1/2]$. If $n \gtrsim k^{1/\alpha}/\ln k$ and $k \gtrsim \ln^4 n$,
  \begin{align}
    R^*(n,k;\phi) \gtrsim \frac{k^2}{(n\ln n)^{2\alpha}}.
  \end{align}
\end{theorem}
\begin{theorem}\label{thm:lower3}
 Suppose $\phi:[0,1]\to\RealSet$ is a function whose second divergence speed is $p^{\alpha}$ for $\alpha \in (1/2,2]$. If $n \gtrsim k^{1/\alpha}/\ln k$, and $k^2/(n\ln n)^{2\alpha}$ dominates the lower bound in \cref{thm:lower1},
 \begin{align}
    R^*(n,k;\phi) \gtrsim \frac{k^2}{(n\ln n)^{2\alpha}}.
  \end{align}
\end{theorem}
The lower bound shown in \cref{thm:optimal-rate-0-1/2,thm:optimal-rate-1/2-1,thm:optimal-rate-1,thm:optimal-rate-1-3/2,thm:optimal-rate-3/2} can be obtained by simply combining \cref{thm:lower1,thm:lower2,thm:lower3}.

\subsection{Lower Bound Analysis for {\cref{thm:lower1}}}
We use the Le Cam's two point method~(see, e.g., \autocite{Tsybakov2009IntroductionEstimation}) to prove \cref{thm:lower1}. Let $P$ and $Q$ be two probability vectors in $\dom{M}_k$. Then, the lower bound is given by
\begin{lemma}[\autocite{Tsybakov2009IntroductionEstimation}]\label{lem:le-cam-two-point}
  The minimax lower bound is given as
  \begin{align}
    R^*(n,k;\phi) \ge \frac{1}{4}\paren*{\theta(P) - \theta(Q)}^2e^{-n\KL(P,Q)},
  \end{align}
  where $\KL$ denotes the KL divergence.
\end{lemma}
From \cref{lem:le-cam-two-point}, we want to appropriately choose $P$ and $Q$ that maximizes difference between $\theta(P)$ and $\theta(Q)$ with small KL divergence between $P$ and $Q$. Given $p$ and $q$, we define $P$ and $Q$ as
\begin{align}
  P =& \paren*{1-p,\frac{p}{k-1},...,\frac{p}{k-1}}, \\
  Q =& \paren*{1-q,\frac{q}{k-1},...,\frac{q}{k-1}}.
\end{align}
With this definition, we derive the upper bound on $\KL(P,Q)$ and the lower bound on $\paren{\theta(P) - \theta(Q)}^2$. Then, we choose $p$ and $q$ appropriately to obtain the minimax lower bound in \cref{thm:lower1}.
\begin{proof}[Proof of \cref{thm:lower1}]
  We first upper bound the KL divergence between $P$ and $Q$. Letting $\ChiSquare$ be the $\chi^2$ divergence, there is a well-known bound as $\KL(P,Q) \le \ChiSquare(P,Q)/2$~(see, e.g., \autocite{Tsybakov2009IntroductionEstimation}). Hence,
  \begin{align}
    \KL(P,Q) \le& \frac{1}{2}\ChiSquare(P,Q) \\
    =& \frac{(p-q)^2}{2(1-p)} + (k-1)\frac{\paren*{\frac{p}{k-1} - \frac{q}{k-1}}^2}{2\frac{p}{k-1}} \\
    =& \frac{(p-q)^2}{2(1-p)} + \frac{(p-q)^2}{2p} = \frac{(p-q)^2}{2p(1-p)}.
  \end{align}
  From \cref{lem:le-cam-two-point}, we choose $p$ and $q$ that maximizes $\abs{\theta(P) - \theta(Q)}$ under constraint $(p-q)^2/2p(1-p) \lesssim 1/n$. Application of the Taylor theorem yields that there exist $\xi_1$ between $1-p$ and $1-q$ and $\xi_2$ between $p$ and $q$ such that
  \begin{align}
    & \abs*{\theta(P) - \theta(Q)} \\
    =& \begin{multlined}[t]
      \abs[\Bigg]{\phi(1-p) - \phi(1-q) \\ + (k-1)\paren*{\phi\paren*{\frac{p}{k-1}} - \phi\paren*{\frac{q}{k-1}}}}
    \end{multlined}\\
    =& \abs*{\phi^{(1)}(\xi_1)(q-p) + \phi^{(1)}\paren*{\frac{\xi_2}{k-1}}(p-q)} \\
    =& \abs*{\phi^{(1)}(\xi_1) - \phi^{(1)}\paren*{\frac{\xi_2}{k-1}}}\abs*{p - q}. \label{eq:lower1-diff-th}
  \end{align}

  We derive a lower bound on \cref{eq:lower1-diff-th} by dividing $\alpha$ into three cases; $\alpha \in (0,1)$, $\alpha = 1$, and $\alpha \in (1,2]$. For the following analysis, it is worthy to note that we have $\abs*{\phi^{(2)}(p)} > 0$ for $p \in (0,p_0)$ where $p_0 = 1\land(c'_2/W_2)^{1/(\alpha-2)}$ because of the divergence speed assumption. Thus, $\phi^{(2)}$ has the same sign in $p \in (0,p_0)$.

  {\bfseries Case $\alpha \in (0,1)$.}
  Suppose $p,q \le p_0$. From the absolutely continuity of $\phi^{(1)}$ and \cref{lem:lowest-int-div}, we have
  \begin{align}
      & \abs*{\phi^{(1)}\paren*{\frac{\xi_2}{k-1}} - \phi^{(1)}(0)} \\
      =& \abs*{\int_0^{\xi_2/(k-1)}\phi^{(2)}(x) dx} \\
      \ge& \abs*{\int_0^{\xi_2/(k-1)}W_2x^{\alpha-2} - c'_2 dx} \\
      =& \frac{W_2}{1-\alpha}\paren*{\frac{k-1}{\xi_2}}^{1-\alpha} - c'_2\frac{\xi_2}{k-1} \gtrsim k^{1-\alpha}.
  \end{align}
  Also, we have
  \begin{align}
      \abs*{\phi^{(1)}(\xi_1) - \phi^{(1)}(0)} =& \abs*{\int_0^{\xi_1}\phi^{(2)}(x)dx} \\
      \le& \abs*{\int_0^{\xi_1}W_2x^{\alpha-2} + c_2dx} \\
      \le& \frac{W_2}{1-\alpha}\xi_1^{-(1-\alpha)} + c_2\xi_1.
  \end{align}
  Since  $\xi_2 \le p \lor q$ and $\xi_1 \ge (1-p)\land(1-q)$, we have $\abs{\phi^{(1)}\paren{\frac{\xi_2}{k-1}} - \phi^{(1)}(0)} \ge \abs*{\phi^{(1)}(\xi_1) - \phi^{(1)}(0)}$ for sufficiently small $p$ and $q$. Suppose $p$ is a sufficiently small universal constant such that $p \le p_0$, and $q = p - \frac{c}{\sqrt{n}}$ for an universal constant $c > 0$ such that $c < p_0$. Then, we have $\KL(P,Q) \lesssim 1/n$ and
  \begin{align}
    \abs*{\theta(P) - \theta(Q)} \gtrsim \frac{k^{1-\alpha}}{\sqrt{n}}.
  \end{align}

  {\bfseries Case $\alpha = 1$.}
  Suppose $p,q \le p_0$. From the absolutely continuity of $\phi^{(1)}$ and \cref{lem:lowest-int-div}, we have
  \begin{align}
      & \abs*{\phi^{(1)}\paren*{\frac{\xi_2}{k-1}} - \phi^{(1)}(p_0)} \\
      =& \abs*{\int_{\xi_2/(k-1)}^{p_0}\phi^{(2)}(x) dx} \\
      \ge& \abs*{\int_{\xi_2/(k-1)}^{p_0}W_2x^{-1} - c'_2 dx} \\
      =& \abs*{W_2\ln(p_0(k-1)/\xi_2) - c'_2\paren*{p_0-\frac{\xi_2}{k-1}}} \gtrsim \ln k.
  \end{align}
  Also, we have
  \begin{align}
      \abs*{\phi^{(1)}(\xi_1) - \phi^{(1)}(p_0)} =& \abs*{\int_{p_0}^{\xi_1}\phi^{(2)}(x)dx} \\
      \le& \abs*{\int_{p_0}^{\xi_1}W_2x^{-1} + c_2dx} \\
      \le& \abs*{W_2\ln(p_0/\xi_1) - c_2(\xi_1-p_0)}.
  \end{align}
  As well as the case $\alpha \in (0,1)$, we have $\abs{\phi^{(1)}\paren{\frac{\xi_2}{k-1}} - \phi^{(1)}(p_0)} \ge \abs*{\phi^{(1)}(\xi_1) - \phi^{(1)}(p_0)}$ for sufficiently small $p$ and $q$. Again, suppose $p$ is a sufficiently small universal constant such that $p \le p_0$, and $q = p - \frac{c}{\sqrt{n}}$ for an universal constant $c > 0$ such that $c < p_0$. Then, we have $\KL(P,Q) \lesssim 1/n$ and
  \begin{align}
    \abs*{\theta(P) - \theta(Q)} \gtrsim \frac{\ln k}{\sqrt{n}}.
  \end{align}

  {\bfseries Case $\alpha \in (1,2]$.}
  We can assume $\phi^{(1)}(0) = 0$ without loss of generality because for any $c \in \RealSet$, $\theta(P;\phi) = \theta(P;\phi_c)$ for $\phi_c(p) = \phi(p) + c(p-1/k)$. For some universal constant $c > 0$, we set $p$ and $q$ such that $c \le 1-p,1-q < p_0$. Then, we have
  \begin{align}
    \abs*{\phi^{(1)}(\xi_1)} =& \abs*{\int_0^{\xi_1}\phi^{(2)}(x)dx} \\
    \ge& \int_0^{\xi_1}\paren*{W_2x^{\alpha-2}-c'_2}dx \\
    \ge& \frac{W_2}{\alpha-1}c^{\alpha-1}-c'_2c > 0.
  \end{align}
  Also, we have
  \begin{align}
    & \abs*{\phi^{(1)}\paren*{\frac{\xi_2}{k-1}}} \\
    \le& \int_0^{\frac{\xi_2}{k-1}}\abs*{\phi^{(2)}(x)}dx \\
    \le& \int_0^{\frac{\xi_2}{k-1}}\paren*{W_2x^{\alpha-2}+c_2}dx \\
    \le& \frac{W_2}{(\alpha-1)(k-1)^{\alpha-1}}\xi_2^{\alpha-1}+\frac{c_2}{k-1}\xi_2 \\
    \le& \frac{W_2}{(\alpha-1)(k-1)^{\alpha-1}}(p \lor q)^{\alpha-1}+\frac{c_2}{k-1}(p \lor q).
  \end{align}
  Thus, for sufficiently large $k$ such that $\abs*{\phi^{(1)}(\xi_1)} \ge \abs*{\phi^{(1)}\paren*{\frac{\xi_2}{k-1}}}$, we have
  \begin{align}
    \abs*{\theta(P) - \theta(Q)} \gtrsim \abs*{p - q}.
  \end{align}
  Let $p$ be an sufficiently large universal constant such that $1-p < p_0$, and let $q=p+\frac{c}{\sqrt{n}}$ for sufficiently small $c > 0$. Then, $1-q < p_0$ and $(p-q)^2/2p(1-p) \lesssim 1/n$. Hence, we have $\KL(P,Q) \lesssim 1/n$ and
  \begin{align}
    \abs*{\theta(P) - \theta(Q)} \gtrsim \frac{1}{n}.
  \end{align}

  Combining all the cases and applying \cref{lem:le-cam-two-point} yields the claim.
\end{proof}

\subsection{Lower Bound Analysis for {\cref{thm:lower2,thm:lower3}}}
The proofs of \cref{thm:lower2,thm:lower3} basically follow the same manner of \textcite{Wu2016MinimaxApproximation} in which they characterized the lower bound on the minimax risk for Shannon entropy by the best polynomial approximation error. We generalize their result to be applicable for the general additive functional.

\subsubsection{Lower Bound using Best Polynomial Approximation}
The first step to prove \cref{thm:lower2,thm:lower3} is to connect the minimax risk to the best polynomial approximation error. More precisely, we prove the following claim:
\begin{theorem}\label{thm:risk-best-poly}
  Let $\phi:[0,1]\to\RealSet$ be a function whose second divergence speed is $p^\alpha$ for $\alpha > 0$. Given positive integers $L$, let $d > 0$ and $\lambda > 0$ be positive reals such that either of the following conditions holds:
  \begin{enumerate}
    \item $\lambda \le 1/12$ and $2kE_L\paren{\phi,[0,\lambda/k]} \ge d$,
    \item $\lambda \le \sqrt{k}/12$ and $\exists \gamma \in (0,1)$, $2\gamma E_L\paren{\phi^\star,[\gamma,\gamma^2\lambda/k]} \ge d$ and $\gamma^2\lambda \le k$,
  \end{enumerate}
  where $\phi^\star(x)=\phi(x)/x$. Then, if $\alpha \in (0,1)$, there exists a finite constant $W,W' > 0$ depending only on $\phi$ such that
  \begin{multline}
    \tilde{R}^*(n/2,k,;\phi) \ge \frac{d^2}{32}\paren*{\frac{7}{8} - k\paren*{\frac{2en\lambda}{Lk}}^L} \\  - Wk^{1-2\alpha}\lambda^{2\alpha} - W'k^{2-2\alpha}e^{-n/32} - 4^{2\alpha}W'k^{2-2\alpha}k^{-\alpha}\lambda^{2\alpha}.
  \end{multline}
  If $\alpha = 1$, there exists a finite constant $W,W' > 0$ depending only on $\phi$ such that
  \begin{multline}
    \tilde{R}^*(n/2,k;\phi) \ge \frac{d^2}{32}\paren*{\frac{7}{8} - k\paren*{\frac{2en\lambda}{Lk}}^L} - \frac{W\lambda^2\ln^2(\lambda/ek)}{k} \\ - W'\ln^2\paren*{ek}e^{-n/32} - 16W'\frac{\lambda^2}{k}\ln^2\paren*{ek} \\ - W' \paren*{1+\paren*{\frac{4\lambda}{\sqrt{k}}}}^2\ln^2\paren*{1+\paren*{\frac{4\lambda}{\sqrt{k}}}}.
  \end{multline}
  If $\alpha \in (1,2)$, there exists a finite constant $W,W' > 0$ depending only on $\phi$ such that
  \begin{multline}
    \tilde{R}^*(n/2,k,;\phi) \ge \frac{d^2}{32}\paren*{\frac{7}{8} - k\paren*{\frac{2en\lambda}{Lk}}^L} \\ - Wk^{1-2\alpha}\lambda^{2\alpha} - W'e^{-n/32} - 16W'k^{-2}\lambda^2.
  \end{multline}
\end{theorem}

{\bfseries Lower bound by the approximated minimax risk.}
To prove \cref{thm:risk-best-poly}, we first derive the association between the minimax risk and the approximated minimax risk defined below. For $\epsilon \in (0,1)$, define the approximated probabilities as
\begin{align}
  \dom{M}_k(\epsilon) = \cbrace*{(p_1,...,p_k) \in \RealSet^k_+ : \abs*{\sum_{i=1}^kp_i - 1} \le \epsilon}.
\end{align}
With this definition, we define the approximated minimax risk as
\begin{align}
 \tilde{R}^*(n,k,\epsilon;\phi) = \inf_{\hat\theta}\sup_{P \in \dom{M}_k(\epsilon)}\Mean\bracket*{\paren*{\theta(P) - \hat\theta(\tilde{N})}^2}. \label{eq:approximated-minimax-risk}
\end{align}
Then, we obtain the following connection from the approximated minimax risk to the non-approximated minimax risk:
\begin{theorem}\label{thm:lower-approximated}
 Suppose $\phi:[0,1]\to\RealSet$ is a function whose second divergence speed is $p^\alpha$ for $\alpha \in (0,2)$. If $\alpha \in (0,1)$, for any $k,n \in \NaturalSet$ and any $\epsilon < 1/3$,
  \begin{multline}
    \tilde{R}^*(n/2,k;\phi) \ge \frac{1}{2}\tilde{R}^*(n,k,\epsilon;\phi) \\ - Wk^{2-2\alpha}e^{-n/32} - Wk^{2-2\alpha}\epsilon^{2\alpha},
  \end{multline}
  where $W = 2\norm*{\phi}_{C^{H,\alpha}}^2$. If $\alpha = 1$, for any $k,n \in \NaturalSet$ and any $\epsilon < 1/3$,
  \begin{multline}
    \tilde{R}^*(n/2,k;\phi) \ge \frac{1}{2}\tilde{R}^*(n,k,\epsilon;\phi) - W\ln^2\paren*{ek}e^{-n/32} \\ - W\epsilon^{2}\ln^2\paren*{ek} - W (1+\epsilon)^2\ln^2(1+\epsilon),
  \end{multline}
  where $W = 2(W_1+c_1)^2$ in which $W_1$ and $c_1$ are constants from \cref{lem:lowest-int-div}. If $\alpha \in (1,2)$, for any $k,n \in \NaturalSet$ and any $\epsilon < 1/3$,
  \begin{align}
    \tilde{R}^*(n/2,k;\phi) \ge \frac{1}{2}\tilde{R}^*(n,k,\epsilon;\phi) - We^{-n/32} - W\epsilon^{2},
  \end{align}
  where $W = 2\norm*{\phi}_{C^{H,1}}^2$.
\end{theorem}

{\bfseries Lower bound by the fuzzy hypotheses method.}
Next, we derive a lower bound on the approximated minimax risk by taking advantage of the {\em two fuzzy hypotheses method}~\autocite{Tsybakov2009IntroductionEstimation}. This method constructs two prior distributions $\pi$ and $\pi'$ on the parameter space instead of choosing two parameters as in \cref{lem:le-cam-two-point}. In particular, we construct two stochastic probability vectors from random variables $U$ and $U'$ defined on $[0,\lambda]$ for some $\lambda > 0$ such that $\Mean[U] = \Mean[U'] = \beta \le 1$. Let
\begin{align}
    P = \paren*{\frac{U_1}{k},...,\frac{U_{k}}{k},1-\beta}, \textand P' = \paren*{\frac{U'_1}{k},...,\frac{U'_{k}}{k},1-\beta},
\end{align}
where $P$ and $P'$ are $k+1$-dimensional vectors, and $U_1,...,U_{k}$ and $U'_1,...,U'_{k}$ are i.i.d. copies of $U$ and $U'$, respectively. Define two events on $P$ and $P'$, respectively, as
\begin{align}
 \event =& \cbrace*{ P \in \dom{M}_{k+1}(\epsilon), \abs*{\Mean[\theta(P)] - \theta(P)} \le d/4}, \\
 \event' =& \cbrace*{P' \in \dom{M}_{k+1}(\epsilon), \abs*{\Mean[\theta(P')] - \theta(P')} \le d/4}.
\end{align}
Then, the prior distributions $\pi$ and $\pi'$ are defined as the distributions of $P$ and $P'$ conditioned on these events, i.e., for any $A \subseteq \dom{M}_k(\epsilon)$,
\begin{align}
  \pi A = \p\cbrace*{P \in A | \event} \textand \pi' A = \p\cbrace*{P' \in A | \event'}.
\end{align}
With this setup of the prior distributions, we obtain the following lower bound:
\begin{theorem}\label{thm:approx-tv-lower}
  Let $\phi:[0,1]\to\RealSet$ be a function whose second divergence speed is $p^\alpha$ for $\alpha > 0$. Let $U$ and $U'$ be random variables such that $U,U' \in [0,\lambda]$, $\Mean[U]=\Mean[U']\le 1$, and $k\abs*{\Mean\bracket{\phi(U/k) - \phi(U'/k)}} \ge d$. If $\alpha \in (0,2)$ and $\lambda \le k$ expect $\alpha \ne 1$, there exists a finite constant $W > 0$ depending only on $\phi$ such that
  \begin{multline}
    \tilde{R}^*(n,k+1,4\lambda/\sqrt{k};\phi) \ge \frac{d^2}{16}\paren[\Bigg]{\frac{7}{8} - \\ k\TV\paren*{\Mean[\Poi(nU/k)], \Mean[\Poi(nU'/k)]} - \frac{W\lambda^{2\alpha}}{k^{2\alpha-1}d^2}}.
  \end{multline}
  Moreover, if $\alpha = 1$ and $\lambda < ek$, there exists a finite constant $W > 0$ depending only on $\phi$ such that
  \begin{multline}
    \tilde{R}^*(n,k+1,4\lambda/\sqrt{k};\phi) \ge \frac{d^2}{16}\paren[\Bigg]{\frac{7}{8} -  \\ k\TV\paren*{\Mean[\Poi(nU/k)], \Mean[\Poi(nU'/k)]} - \frac{W\lambda^2\ln^2(\lambda/ek)}{kd^2}}.
  \end{multline}
\end{theorem}

{\bfseries Lower bound by the best polynomial approximation.}
We can prove \cref{thm:risk-best-poly} by using \cref{thm:approx-tv-lower} with appropriate choices of $U$ and $U'$. We choose $U$ and $U'$ so that their moments are agreed with. Formally, we choose $U$ and $U'$ such that for some positive integer $L$,
\begin{align}
    \Mean\bracket{U^m} = \Mean\bracket{U^{'m}} \for m = 1,...,L. \label{eq:moment-matching}
\end{align}
Under \cref{eq:moment-matching}, we can use the following lemma given by \textcite{Wu2016MinimaxApproximation} to upper bound the total variation term in \cref{thm:approx-tv-lower}:
\begin{lemma}[{\textcite[Lemma 3]{Wu2016MinimaxApproximation}}]\label{lem:tv-poi-bound}
  Let $V$ and $V'$ be random variables on $[0,M]$. If $\Mean[V^j] = \Mean[V'^j]$, $j = 1,...,L$ and $L > 2eM$, then
  \begin{align}
    \TV(\Mean[\Poi(V)], \Mean[\Poi(V')]) \le \paren*{\frac{2eM}{L}}^L.
  \end{align}
\end{lemma}
Besides, we can connect $d$ in \cref{thm:approx-tv-lower} to the best polynomial approximation error by using either of the following lemmas:
\begin{lemma}\label{lem:prior-construction}
  For any given integer $L > 0$ and an interval $I \subseteq [0,1]$, there exists two probability measures $\nu_0$ and $\nu_1$ on $I$ such that
  \begin{gather}
    \Mean_{X \sim \nu_0}[X^m] = \Mean_{X \sim \nu_1}[X^m], \for m=0,...,L, \\
    \Mean_{X \sim \nu_0}[\phi(X)] - \Mean_{X \sim \nu_1}[\phi(X)] = 2E_L(\phi, I).
  \end{gather}
\end{lemma}
\begin{lemma}\label{lem:best-approx-solution}
  Suppose $\phi:[0,1]\to\RealSet$ be a function such that $\phi(0) = 0$. Define $\phi^\star(p) = \phi(x)/x$. For any given integer $L > 0$, $\eta \in (0,1)$, and $\gamma \in (0,1)$ such that $\gamma \le \eta$, there exists two probability measures $\nu_0$ and $\nu_1$ on $[0,1/\eta\gamma]$ such that
  \begin{gather}
    \Mean_{X \sim \nu_0}[X] = \Mean_{X \sim \nu_1}[X] = \gamma, \\
    \Mean_{X \sim \nu_0}[X^m] = \Mean_{X \sim \nu_1}[X^m], \for m=2,...,L+1, \\
    \Mean_{X \sim \nu_0}[\phi(X)] - \Mean_{X \sim \nu_1}[\phi(X)] = 2\gamma E_L(\phi^\star, [\gamma\eta,\gamma]).
  \end{gather}
\end{lemma}
\begin{proof}[Proof of \cref{lem:prior-construction}]
  The proof is almost the same as the proof of \textcite[Lemma 10]{Jiao2015MinimaxDistributions}. It follows directly from a standard functional analysis argument proposed by \textcite{Lepski1999OnFunction}. It suffices to replace $x^\alpha$ with $\phi(x)$ and $[0,1]$ with $I$ in the proof of \autocite[Lemma 1]{Cai2011TestingFunctional}.
\end{proof}
\begin{proof}[Proof of \cref{lem:best-approx-solution}]
  From \cref{lem:prior-construction}, there exists a pair of probability measures $\rho_0$ and $\rho_1$ on $[\gamma,\gamma/\eta]$ such that
  \begin{gather}
    \Mean_{X \sim \rho_0}[X^m] = \Mean_{X \sim \rho_1}[X^m], \for m=0,...,L, \\
    \Mean_{X \sim \rho_0}[\phi(X)] - \Mean_{X \sim \rho_1}[\phi(X)] = 2E_L(\phi, [\gamma,\gamma\eta]).
  \end{gather}
  Define $\nu_0$ and $\nu_1$ such that
  \begin{align}
    \frac{d\nu_i}{d\rho_i}(u) = \frac{\gamma}{u} \textand \nu_i(\cbrace*{0}) = 1-\nu_i([\eta,1]).
  \end{align}
  By construction, $\nu_i$ are defined on $[0,1/\eta\gamma]$. Besides, we get
  \begin{align}
    \Mean_{X \sim \nu_i}[X] =& \gamma\Mean_{X \sim \rho_i}\bracket*{1} = \gamma, \\
    \Mean_{X \sim \nu_0}[X^m] =& \gamma\Mean_{X \sim \rho_0}\bracket*{X^{m-1}} \\
     =& \gamma\Mean_{X \sim \rho_1}\bracket*{X^{m-1}} \\ =& \Mean_{X \sim \nu_1}[X^m] \for m=2,...,L+1.
  \end{align}
  From the assumption $\phi(0)=0$, we have
  \begin{align}
    & \Mean_{X \sim \nu_0}[\phi(X)] - \Mean_{X \sim \nu_1}[\phi(X)] \\
    =& \gamma\paren*{\Mean_{X \sim \rho_0}\bracket*{\frac{\phi(X)}{X}} - \Mean_{X \sim \rho_1}\bracket*{\frac{\phi(X)}{X}}} \\
    =& 2\gamma E_L\paren*{\phi^\star,[\gamma,\gamma/\eta]}.
  \end{align}
\end{proof}

Combining \cref{thm:lower-approximated,thm:approx-tv-lower,lem:tv-poi-bound,lem:prior-construction,lem:best-approx-solution}, we obtain \cref{thm:risk-best-poly}.
\begin{proof}[Proof of \cref{thm:risk-best-poly}]
 We use \cref{thm:lower-approximated,thm:approx-tv-lower} to prove the claim. Given an positive integer $L$, assume \cref{eq:moment-matching}. We need to check all the conditions in \cref{thm:lower-approximated,thm:approx-tv-lower} are satisfied.

 {\bfseries First condition in the claim.} We firstly confirm that all the conditions in \cref{thm:lower-approximated,thm:approx-tv-lower} are satisfied if the first condition in the claim, i.e., $\lambda \le 1/12$ and $2kE_L(\phi,[0,\lambda/k]) \ge d$, holds. If $\lambda \le 1/12$, the conditions $\Mean[U]=\Mean[U'] \le 1$, $\lambda \le k$ in \cref{thm:approx-tv-lower}, and $\epsilon \le 1/3$ in \cref{thm:lower-approximated} with $\epsilon = 4\lambda/\sqrt{k}$ are satisfied obviously. By \cref{lem:prior-construction}, $\Mean[\phi(U/k) - \phi(U'/k)] = 2E_L(\phi,[0,\lambda/k])$.

 {\bfseries Second condition in the claim.} Next, we confirm that all the conditions in \cref{thm:lower-approximated,thm:approx-tv-lower} are satisfied if the second condition in the claim, i.e., $\lambda \le \sqrt{k}/12$ and $\exists \gamma \in (0,1)$, $2k\gamma E_L(\phi^\star,[\gamma,\gamma^2\lambda/k]) \ge d$ and $\gamma^2\lambda \le k$, holds. By \cref{lem:best-approx-solution} with $1/\eta\gamma = \lambda/k$, $\Mean[\phi(U/k) - \phi(U'/k)] = 2\gamma E_L(\phi^\star,[\gamma,\gamma^2\lambda/k])$, where we need $\gamma^2\lambda \le k$ to satisfy $\gamma \le \eta$. From the conditions $\epsilon \le 1/3$ in \cref{thm:lower-approximated} and $\epsilon = 4\lambda/\sqrt{k}$ in \cref{thm:approx-tv-lower}, we need $\lambda \le \sqrt{k}/12$. 
 
 Noting that $U,U' \in [0,\lambda]$ almost surely, we have from \cref{lem:tv-poi-bound} that
 \begin{align}
    \TV\paren*{\Mean[\Poi(nU/k)], \Mean[\Poi(nU'/k)]} \le \paren*{\frac{2en\lambda}{Lk}}^L.
 \end{align}
 Combining this, \cref{thm:lower-approximated}, and \cref{thm:approx-tv-lower} yields the desired results.
\end{proof}

\subsubsection{Proof of \cref{thm:lower2,thm:lower3}}

Here, we prove \cref{thm:lower2,thm:lower3} by combining \cref{thm:risk-best-poly} and the lower bound results on the best polynomial approximation errors, which can be found in \cref{sec:best-poly}.
\begin{proof}[Proof of \cref{thm:lower2}]
  We apply \cref{thm:risk-best-poly} with the first condition. Set $\lambda = C_1k\ln n/n$ and $L = \ceil{C_2\ln n}$ where $C_1$ and $C_2$ are universal constants. Let $c = C_1/C_2$. Under the condition $k \lesssim (n\ln n)^\alpha$, $\lambda < 1/12$ with sufficiently small $c$ and $C_2$.

  From \cref{thm:lower-poly-approx1}, $2kE_L(\phi,[0,\lambda/k]) \ge d$ for $d = \frac{2Cc^\alpha k}{C_2^{\alpha}(n\ln n)^\alpha}$ with some universal constant $C > 0$.

  If $c < 1/2e$, we have
  \begin{align}
      k\paren*{\frac{2en\lambda}{Lk}}^L \le kn^{-C_2\ln\frac{1}{2ec}}.
  \end{align}
  Under the condition $k \lesssim (n\ln n)^\alpha$, $kn^{-C_2\ln\frac{1}{2ec}} = o(1)$ if $C_2\ln\frac{1}{2ec} > \alpha$. The condition $C_2\ln\frac{1}{2ec} > \alpha$ holds if $c$ is sufficiently small with fixed $C_2$.

  Under the condition $k \gtrsim \ln^4 n$, there exists an universal constant $C > 0$ such that
  \begin{align}
      Wk^{1-2\alpha}\lambda^{2\alpha} \le \frac{W(cC_2)^{2\alpha}k\ln^{2\alpha}n}{n^{2\alpha}} \le \frac{W(cC_2)^{2\alpha}Ck^2}{(n\ln n)^{2\alpha}}.
  \end{align}
  Also, under the condition $k \gtrsim \ln^4 n$, there exists an universal constant $C > 0$ such that
  \begin{align}
      4^{2\alpha}W'k^{2-3\alpha}\lambda^{2\alpha} \le& \frac{4^{2\alpha}W'(cC_2)^{2\alpha}k^{2-\alpha}\ln^{2\alpha}n}{n^{2\alpha}} \\
      \le& \frac{4^{2\alpha}W'C(cC_2)^{2\alpha}k^2}{(n\ln n)^{2\alpha}}.
  \end{align}
  These terms are smaller than $d^2/32$ for sufficiently small $C_2$. Consequentlly, the minimax risk is larger than $d^2$ up to constant.
\end{proof}
\begin{proof}[Proof of \cref{thm:lower3}]
  If $\alpha \in (1/2,1)$, from the assumption of domination, we have $n^{1-1/2\alpha}\ln n \lesssim k$. Hence, with the same proof of the \cref{thm:lower2}, we obtain the claim.

  For $\alpha \in [1,3/2)$, we apply \cref{thm:risk-best-poly} with the second condition. Set $\lambda = C_1k\ln n/n$ and $L = \ceil{C_2\ln n}$ where $C_1$ and $C_2$ are universal constants. Let $c = C_1/C_2$. Under the condition $k \lesssim (n\ln n)^\alpha$, $\lambda < \sqrt{k}/12$ with sufficiently small $c$ and $C_2$.

  Since the setting of $\lambda$ and $L$ is equivalent to that in the proof of \cref{thm:lower2}, we obtain the same bounds on $k\paren{\frac{2en\lambda}{Lk}}^L$. From the assumption of domination, we have $n^{1/2}\ln n \lesssim k/\ln k \lesssim k$ for $\alpha = 1$ and $n^{\alpha-1/2}\ln^{\alpha}n \lesssim k$ for $\alpha \in (1,3/2)$. Thus, $k\paren{\frac{2en\lambda}{Lk}}^L = o(1)$ if $C_2\ln\frac{1}{2ec} > \alpha$.
  
  If $\alpha = 1$, we need to check $\lambda^{2}\ln^2(\lambda/k)/k = o(k^2/(n\ln n)^{2})$, $\epsilon^2\ln^2(ek) = o(k^2/(n\ln n)^{2})$, and $(1+\epsilon)^2\ln^2(1+\epsilon) = o(k^2/(n\ln n)^{2})$. For the first condition, we have $\lambda^{2}\ln^2(\lambda/k)/k \lesssim k\ln^4n/n^2$. To apply \cref{thm:approx-tv-lower}, we need to set $\epsilon = 4\lambda/\sqrt{k}$ and obtain $\epsilon^2\ln^2(ek) \lesssim k\ln^2n\ln^2k/n^2$ and $(1+\epsilon)^2\ln^2(1+\epsilon) \lesssim k\ln^2n/n^2$ because $(1+x)\ln(1+x) \le 2\ln2 x$ for $x \in (0,1)$. Under the domination assumption, we have $k/\ln k \gtrsim n^{1/2}\ln n$. Combining this fact and the assumption $n \gtrsim k^{1/\alpha}/\ln k$, we have $k\ln^4n/n^2 \lesssim k^2\ln^{4}n/n^{1+3/2} = o(k^2/(n\ln n)^2)$ and $k\ln^2n/n^2 \lesssim k\ln^2n\ln^2k/n^2 \lesssim k^2\ln^{4}n/n^{\alpha+3/2} = o(k^2/(n\ln n)^2)$.
  
  If $\alpha \in (1,3/2)$, we need to check $\lambda^{2\alpha}/k^{2\alpha-1} = o(k^2/(n\ln n)^{2\alpha})$ and $\epsilon^2 = o(k^2/(n\ln n)^{2\alpha})$. As well as the case $\alpha = 1$, we have $\lambda^{2\alpha}/k^{2\alpha-1} \lesssim k\ln^{2\alpha}n/n^{2\alpha}$ and $\epsilon^2 \lesssim k\ln^2 n/n^2$. Under the domination assumption, we have $k \gtrsim n^{\alpha-1/2}\ln^\alpha n$. Thus, $k\ln^{2\alpha}n/n^{2\alpha} \lesssim k^2\ln^{\alpha}n/n^{3\alpha-1/2} = o(k^2/(n\ln n)^{2\alpha})$ and $k\ln^2 n/n^2 \lesssim k^2\ln^{2-\alpha}n/n^{\alpha+3/2} = o(k^2/(n\ln n)^{2\alpha})$ if $\alpha < 3/2$.

  In the rest of the proof, it suffices to show that there is $\gamma \in (0,1)$ such that $2k\gamma E_L(\phi^\star,[\gamma,\lambda/k]) \gtrsim k/(n\ln n)^\alpha$ and $\gamma^2\lambda/k \le 1$. Let $\gamma = \lambda/2L^2k$. Then, since $\lambda \le k$, we have $\gamma^2\lambda \le k$. From \cref{thm:lower-poly-approx2}, we have $2k\gamma E_L(\phi^\star,[\gamma,\lambda/k]) \gtrsim k\gamma^{\alpha} \gtrsim k/(n\ln n)^\alpha$.
\end{proof}

\section{Analyses on Best Polynomial Approximation Error}\label{sec:best-poly}

In the upper and lower bounds analyses in \cref{sec:lower-analysis,sec:upper-analysis}, we connect the risks to the best polynomial approximation error of $\phi(p)$ or $\phi^\star(p)=\phi(p)/p$ on the specific intervals. In this section, we derive upper and lower bounds on these errors. More precisely, we analyze $E_L(\phi,[0,\lambda])$ for $\lambda \in (0,1)$ and $E_L(\phi^\star,[\gamma,\gamma/\eta])$ for $\eta > 0$ and $\gamma \in (0,1)$ such that $\gamma \le \eta$. We can use \cref{thm:modulus-best-approx} to obtain the best polynomial approximation error regarding the degree $L$, however, we cannot obtain the dependency on $\lambda$, $\gamma$, and $\eta$ from this theorem. We therefore carry out more precise analyses by using the direct and converse results in \cref{lem:best-modulus-direct,lem:best-modulus-converse}.

We firstly derive an upper bound on $E_L(\phi,[0,\lambda])$ regarding $L$ and $\lambda$ under the divergence speed assumption with $\alpha \in (0,3/2)$.
\begin{theorem}\label{thm:upper-poly-approx1}
  Suppose $\phi:[0,1]\to\RealSet$ be a function such that one of the following condition is satisfied:
  \begin{enumerate}
   \item the first divergence speed of $\phi$ is $p^\alpha$ for $\alpha \in (0,1/2]$,
   \item the second divergence speed of $\phi$ is $p^\alpha$ for $\alpha \in (1/2,1]$,
   \item the third divergence speed of $\phi$ is $p^\alpha$ for $\alpha \in (1,3/2)$.
  \end{enumerate}
  With a increasing sequence of $L$ and a decreasing sequence of $\lambda \in (0,1)$, we have
  \begin{align}
      E_L(\phi,[0,\lambda]) \lesssim \paren*{\frac{\lambda}{L^2}}^\alpha.
  \end{align}
\end{theorem}

Next, we derive a matching lower bound on $E_L(\phi,[0,\lambda])$ under the divergence speed assumption with $\alpha \in (0,1)$.
\begin{theorem}\label{thm:lower-poly-approx1}
  Suppose $\phi:[0,1]\to\RealSet$ be a function whose second divergence speed is $p^\alpha$ for $\alpha \in (0,1)$. Then, we have
  \begin{align}
      \liminf_{L \to \infty, \lambda \to 0}\paren*{\frac{L^2}{\lambda}}^\alpha E_L\paren*{\phi,[0,\lambda]} > 0.
  \end{align}
\end{theorem}
Combining \cref{thm:upper-poly-approx1,thm:lower-poly-approx1}, we see that the best polynomial approximation error $E_L\paren*{\phi,[0,\lambda]}$ for $\alpha \in (0,1)$ is $\paren*{\frac{\lambda}{L^2}}^\alpha$ up to constant. 

Next, we derive a lower bound on $E_L(\phi^\star,[\gamma,\gamma/\eta])$ under the divergence speed assumption with $\alpha \in [1,3/2)$.
\begin{theorem}\label{thm:lower-poly-approx2}
  Let $\cbrace*{\phi_\gamma}$ be a family of functions over $\gamma \in (0,1)$ such that $\phi_\gamma(\gamma) = 0$ and the second order divergence speed of all elements is $p^\alpha$ for $\alpha \in [1,3/2)$. Denote $\phi_\gamma^\star(x)=\phi_\gamma(x)/x$. Then, if $\alpha = 1$ and $\phi_\gamma^{(1)}(\gamma)=0$ for any $\gamma \in (0,1)$,
  \begin{align}
    \liminf_{L \to \infty, \gamma \to 0 : \gamma \le 1/2L^2}E_L\paren*{\phi_\gamma^\star, [\gamma,2L^2\gamma]} > 0.
  \end{align}
  If $\alpha \in (1,3/2)$ and $\phi_\gamma^{(1)}(0)=0$ for any $\gamma \in (0,1)$,
  \begin{align}
    \liminf_{L \to \infty, \gamma \to 0 : \gamma \le 1/2L^2}\gamma^{1-\alpha}E_L\paren*{\phi_\gamma^\star, [\gamma,2L^2\gamma]} > 0.
  \end{align}
\end{theorem}
\begin{remark}
  We can choose such family $\cbrace*{\phi_\gamma}$ because the minimax risk is invariant among $\phi_{c,c'}(x) = \phi(x)+c+c'(x-1/k)$ for any constants $c,c' \in \RealSet$. Note that $\phi^{(\ell)}_{c,c'}(x) = \phi^{(\ell)}(x)$ for any $\ell \ge 2$.
\end{remark}


\section{Conclusion}
In this paper, we investigate the minimax optimal risk of the additive functional estimation problem in large-$k$ regime, and we reveal that the divergence speed characterizes the minimax optimal risk. Our result gives comprehensive understanding of the additive functional estimation problem through the divergence speed. However, our analysis does not cover all the function $\phi$; for example, non-differentiable function. The ambitious goal of this study is to find out a characterization of the minimax optimal risk which is valid for arbitrary function $\phi$. 

\printbibliography

\appendix

\section{Proof for $\alpha \le 0$}\label{sec:no-estimator}
We use the Hellinger distance version of the Le Cam's two point method to prove \cref{prop:no-estimator}:
\begin{lemma}[\autocite{Tsybakov2009IntroductionEstimation}]\label{lem:le-cam-two-point-he}
  The minimax lower bound is given as
  \begin{multline}
    R^*(n,k;\phi) \ge \frac{1}{2}\paren*{\theta(P) - \theta(Q)}^2\\\times\paren*{1 - \sqrt{1 - \paren*{1 - \frac{\Hellinger(P,Q)^2}{4}}^{2n}}},
  \end{multline}
  where $\Hellinger$ denotes the Hellinger distance.
\end{lemma}
\begin{proof}[Proof of \cref{prop:no-estimator}]
  For $\beta \in (0,1]$ and $\delta > 0$, define
  \begin{align}
      P =& \paren*{\frac{\beta}{k-1},...,\frac{\beta}{k-1},1-\beta}, \\
      Q =& \paren*{\frac{\beta+\delta}{k-1},...,\frac{\beta+\delta}{k-1},1-\beta-\delta}.
  \end{align}
  Then, the total variation distance between $P$ and $Q$ is obtained as $\TV(P,Q) = \delta$. From the Le Cam's inequality, we have
  \begin{align}
      \Hellinger(P,Q)^2 \le 2\TV(P,Q).
  \end{align}
  Thus, we have
  \begin{align}
      & \sqrt{1 - \paren*{1 - \frac{\Hellinger(P,Q)^2}{4}}^{2n}} \\
      \le& \sqrt{1 - \paren*{1 - \frac{\delta}{2}}^{2n}}.
  \end{align}
  If $\delta \le 2 - 2(3/4)^{1/2n}$, we have
  \begin{align}
     \sqrt{1 - \paren*{1 - \frac{\Hellinger(P,Q)^2}{4}}^{2n}} \le \frac{1}{2}.
  \end{align}
  Note that this argument is true for any $\beta \in (0,1]$.
  
  By the inverse triangle inequality, we have
  \begin{align}
       & \abs*{\theta(P) - \theta(Q)} \\
      =& \abs*{(k-1)\int_{\frac{\beta+\delta}{k-1}}^{\frac{\beta}{k-1}}\phi^{(1)}(s)ds + \int_{1-\beta-\delta}^{1-\beta}\phi^{(1)}(s)ds} \\
      \ge& (k-1)\abs*{\int_{\frac{\beta+\delta}{k-1}}^{\frac{\beta}{k-1}}\phi^{(1)}(s)ds} - \abs*{\int_{1-\beta-\delta}^{1-\beta}\phi^{(1)}(s)ds}.
  \end{align}
  Let $p_1 = \paren*{\frac{W_1}{W_1 \lor c'_1}}^{1/(1-\alpha)}$, we have $\abs*{\phi^{(1)}(p)} > 0$ for $p \in (0,p_1)$.  From continuity of $\phi^{(1)}$, $\phi^{(1)}(p)$ has the same sign in $p \in (0,p_1]$. Thus, if $\beta+\delta$ is small enough so that $\beta+\delta < p_1\land\frac{1}{2}$, we have
  \begin{align}
      & \abs*{\theta(P) - \theta(Q)} \\
      \ge& \begin{multlined}[t]
        (k-1)\abs*{\int_{\frac{\beta+\delta}{k-1}}^{\frac{\beta}{k-1}}\paren*{W_1s^{\alpha-1} - c'_1}ds} \\ - \abs*{\int_{1-\beta-\delta}^{1-\beta}\paren*{W_1s^{1-\alpha} + c_1} ds}
      \end{multlined}\\
      \ge& (k-1)\abs*{\int_{\frac{\beta+\delta}{k-1}}^{\frac{\beta}{k-1}}W_1s^{\alpha-1}ds} - c'_1\delta - \paren*{W_12^{1-\alpha} + c_1}\delta \\
      \ge& (k-1)\abs*{\int_{\frac{\beta+\delta}{k-1}}^{\frac{\beta}{k-1}}W_1s^{-1}ds} - c'_1\delta - \paren*{W_12^{1-\alpha} + c_1}\delta \\
      \ge& W_1(k-1)\ln\paren*{1 + \frac{\delta}{\beta}} - c'_1\delta - \paren*{W_12^{1-\alpha} + c_1}\delta
  \end{align}
  By \cref{lem:le-cam-two-point-he} and letting $\beta = \delta$, for sufficiently small $\delta$, we have
  \begin{align}
    R^*(n,k;\phi) \ge \frac{1}{4}\paren*{W_1(k-1)\ln2 - (c'_1+W_12^{1-\alpha} + c_1)\delta}^2.
  \end{align}
  From arbitrariness of $\delta > 0$, we have
  \begin{align}
    R^*(n,k;\phi) \ge \frac{1}{4}\paren*{W_1(k-1)\ln2}^2.
  \end{align}
\end{proof}

\section{Proofs for Lower Order Divergence Speed}\label{sec:lower-div}
\begin{proof}[Proof of {\cref{lem:lower-div}}]
 Letting $p_\ell = \paren*{\frac{W_\ell}{W_\ell \lor c'_\ell}}^{1/(\ell-\alpha)}$, we have $\abs*{\phi^{(\ell)}(p)} > 0$ for $p \in (0,p_\ell)$. From continuity of $\phi^{(\ell)}$, $\phi^{(\ell)}(p)$ has the same sign in $p \in (0,p_\ell]$, and thus we have either $\phi^{(\ell)}(p) \ge W_\ell p^{\alpha-\ell} - c'_\ell$ or $\phi^{(\ell)}(p) \le -(Wp^{\alpha-\ell} - c'_m)$ in $p \in (0,p_\ell]$. From absolutely continuity of $\phi^{(m-1)}$, we have for any $p \in (0,1)$,
 \begin{align}
   \phi^{(\ell-1)}(p) = \phi^{(\ell-1)}(p_\ell) + \int_{p_\ell}^p \phi^{(\ell)}(x) dx. \label{eq:lower-div-decomp}
 \end{align}
 The absolute value of the second term in \cref{eq:lower-div-decomp} has an upper bound as
 \begin{align}
  & \abs*{\int_{p_\ell}^p \phi^{(\ell)}(x) dx} \\
  \le& \abs*{\int_{p_\ell}^p W_\ell x^{\alpha-\ell} + c_\ell dx} \\
   \le& \abs*{ (\alpha-\ell+1)W_\ell \paren*{p_\ell^{\alpha-\ell+1}-p^{\alpha-\ell+1}} + c_\ell\paren*{p - p_\ell}} \\
   \le& \begin{multlined}[t]
     (\alpha-\ell+1)W_\ell p^{\alpha-(\ell-1)} \\ + \abs*{(\alpha-\ell+1) W_\ell p_\ell^{\alpha-\ell+1} + c_\ell(p_\ell-p)}
   \end{multlined}\\
   \le& (\alpha-\ell+1)W_\ell p^{\alpha-m+1} + (\alpha-\ell+1)W_\ell p_\ell^{\alpha-\ell+1} + c_\ell.
 \end{align}
 Also, we have a lower bound of the second term in \cref{eq:lower-div-decomp} as
 \begin{align}
  & \abs*{\int_{p_\ell}^p \phi^{(\ell)}(x) dx} \\
   =& \abs*{\int_{p_\ell}^{p \land p_\ell} \phi^{(\ell)}(x) dx + \int_{p \land p_\ell}^p \phi^{(\ell)}(x) dx}\\
   \ge& \abs*{\int_{p_\ell}^{p \land p_\ell} W_\ell x^{\alpha-\ell} - c'_\ell dx} - \abs*{\int_{p \land p_\ell}^{p} W_\ell x^{\alpha-\ell} + c_\ell dx} \\
   \ge& \begin{multlined}[t][.9\columnwidth]
    \abs[\Bigg]{(\alpha-\ell+1)W_\ell\paren*{p_\ell^{\alpha-\ell+1} - \paren*{p \land p_\ell}^{\alpha-\ell+1}} \\ - c'_\ell((p \land p_\ell) - p_\ell)} - \abs*{\paren*{W_\ell p_\ell^{\alpha-\ell} + c_\ell}(p - (p \land p_\ell))} \\ \because x^{\alpha-\ell} \le p_\ell^{\alpha-\ell} \for x \in [p_\ell,1) \textif \alpha \le \ell
   \end{multlined} \\
   \ge& \begin{multlined}[t][.9\columnwidth]
    (\alpha-\ell+1)W_\ell\paren*{p \land p_\ell}^{\alpha-\ell+1} - (\alpha-\ell+1)W_\ell p_\ell^{\alpha-\ell+1} \\ - \abs*{c'_\ell(p_\ell - (p \land p_\ell))} - \paren*{W_\ell p_\ell^{\alpha-\ell} + c_\ell}(p - (p \land p_\ell))
   \end{multlined} \\
   \ge& \begin{multlined}[t]
    (\alpha-\ell+1)W_\ell p^{\alpha-\ell+1} -  (\alpha-\ell+1)W_\ell p_\ell^{\alpha-\ell+1} - c'_\ell p_\ell \\ - \paren*{W_\ell p_\ell^{\alpha-\ell} + c_\ell}(1 - p_\ell),
   \end{multlined}
 \end{align}
 where we use the reverse triangle inequality to obtain the third and fifth lines. Applying the triangle inequality and the reverse triangle inequality gives
 \begin{multline}
    \abs*{\int_{p_\ell}^p \phi^{(\ell)}(x) dx} - \abs*{\phi^{(\ell-1)}(p_\ell)}\le \\ \abs*{\phi^{(\ell-1)}(p)} \\ \le \abs*{\int_{p_\ell}^p \phi^{(\ell)}(x) dx} + \abs*{\phi^{(\ell-1)}(p_\ell)}.
 \end{multline}
 Thus, setting $W_{\ell-1} = (\alpha-\ell+1)W_\ell$, $c_\ell = W_{\ell-1}p^{\alpha-\ell+1}+c_\ell+\abs{\phi^{(\ell-1)}(p_\ell)}$, and $c'_{\ell-1} = W_{\ell-1}p_\ell^{\alpha-\ell+1} + c'_\ell p_\ell - \paren{W_\ell p_\ell^{\alpha-\ell} + c_\ell}(1 - p_\ell) - \abs*{\phi^{(\ell-1)}(p_\ell)}$, we have for all $p \in (0,1)$,
 \begin{multline}
     W_{\ell-1}p^{\alpha-(\ell-1)} - c'_{\ell-1} \le \\ \abs*{\phi^{(\ell-1)}(p)} \\ \le W_{\ell-1}p^{\alpha-(\ell-1)} + c_{\ell-1}.
 \end{multline}
 Thus, $(\ell-1)$th divergence speed of $\phi$ is $p^\alpha$.
\end{proof}

\begin{proof}[Proof of \cref{lem:lowest-int-div}]
  By \cref{lem:lowest-div}, $(\alpha+1)$th divergence speed is $p^\alpha$. Then, we prove the claim in the same manner of the proof of \cref{lem:lower-div}. Letting $p_{\alpha+1}=\paren*{\frac{W_{\alpha+1}}{W_{\alpha+1}\lor c'_{\alpha+1}}}$, $\phi^{(\alpha+1)}(p)$ has the same sign in $p \in (0,p_{\alpha+1}]$. For any $p \in (0,1)$, we have
  \begin{align}
      & \abs*{\int_{p_{\alpha+1}}^p\phi^{(\alpha+1)}(x)dx} \\
      \le& \abs*{\int_{p_{\alpha+1}}^p W_{\alpha+1}x^{-1} + c_{\alpha+1} dx} \\
       =& \abs*{W_{\alpha+1}\paren*{\ln(p) - \ln(p_{\alpha+1})} + c_{\alpha+1}\paren*{p-p_{\alpha+1}}} \\
       \le& W_{\alpha+1}\ln(1/p) + W_{\alpha+1}\ln(1/p_{\alpha+1}) + c_{\alpha+1}.
  \end{align}
  Also, we have
  \begin{align}
      & \abs*{\int_{p_{\alpha+1}}^p\phi^{(\alpha+1)}(x)dx} \\
      =& \abs*{\int_{p_{\alpha+1}}^{p\land p_{\alpha+1}}\phi^{(\alpha+1)}(x)dx + \int_{p\land p_{\alpha+1}}^p\phi^{(\alpha+1)}(x)dx} \\
      \ge& \begin{multlined}[t]
        \abs*{\int_{p_{\alpha+1}}^{p\land p_{\alpha+1}}W_{\alpha+1}x^{-1} - c'_{\alpha+1}dx} \\ - \abs*{\int_{p\land p_{\alpha+1}}^pW_{\alpha+1}x^{-1} + c_{\alpha+1}dx}
      \end{multlined}\\
      \ge& \begin{multlined}[t]
        \abs*{W_{\alpha+1}\ln\paren*{\frac{p \land p_{\alpha+1}}{p_{\alpha+1}}} - c'_{\alpha+1}\paren*{p \land p_{\alpha+1} - p_{\alpha+1}}} \\ - \abs*{W_{\alpha+1}p_{\alpha+1}^{-1} + c_{\alpha+1}}\abs*{p - p\land p_{\alpha+1}}
      \end{multlined}\\
      \ge& \begin{multlined}[t]
        W_{\alpha+1}\ln(1/p) - W_{\alpha+1}\ln(1/p_{\alpha+1}) \\ - c'_{\alpha+1}p_{\alpha+1} - \paren*{W_{\alpha+1}p_{\alpha+1}^{-1} + c_{\alpha+1}}(1-p_{\alpha+1}).
      \end{multlined}
  \end{align}
  From the absolutely continuity and the triangle and inverse triangle inequalities, we have
  \begin{multline}
      \abs*{\int_{p_{\alpha+1}}^p\phi^{(\alpha+1)}(x)dx} - \abs*{\phi^{(\alpha)}(p_{\alpha+1})} \le \\ \abs*{\phi^{(\alpha)}(p)} \\ \le \abs*{\int_{p_{\alpha+1}}^p\phi^{(\alpha+1)}(x)dx} + \abs*{\phi^{(\alpha)}(p_{\alpha+1})}.
  \end{multline}
  Hence, we get the claim by setting $W_\alpha=W_{\alpha+1}$, $c_\alpha = W_{\alpha+1}\ln(1/p_{\alpha+1}) + c_{\alpha+1} + \abs{\phi^{(\alpha)}(p_{\alpha+1})}$, and $c'_\alpha = W_{\alpha+1}\ln(1/p_{\alpha+1}) + c'_{\alpha+1}p_{\alpha+1} + \paren{W_{\alpha+1}p_{\alpha+1}^{-1} + c_{\alpha+1}}(1-p_{\alpha+1}) + \abs*{\phi^{(\alpha)}(p_{\alpha+1})}$.
\end{proof}

\section{Proofs for H\"older continuity}

\begin{proof}[Proof of \cref{lem:div-speed-holder1}]
  The H\"older continuity is proved by showing there exists an universal constant $C > 0$ such that for any $x,y \in (0,1)$,
  \begin{align}
    \abs*{\phi(x) - \phi(y)} \le C\abs*{x - y}^{\alpha}.
  \end{align}
  The absolutely continuity of $\phi$ yields that for any $x,y \in (0,1)$
  \begin{align}
      \abs*{\phi(x) - \phi(y)} =& \abs*{\int_x^y \phi^{(1)}(s) ds} \\
      \le& \abs*{\int_x^y\abs*{\phi^{(1)}(s)}ds} \\
      \le& \abs*{\int_x^y\paren*{W_1s^{\alpha-1}+c_1}ds} \\
      \le& \frac{W_1}{\alpha}\abs*{x^\alpha - y^\alpha} + c_1\abs*{x - y}.
  \end{align}
  Since a function $x \to x^\beta$ for $\beta \in (0,1)$ is $\beta$-H\"older continuous and $\abs{x - y} \le 1$ for all $x,y \in (0,1)$, we have for any $x,y \in (0,1)$
  \begin{align}
      \abs*{\phi(x) - \phi(y)} \le& \paren*{\frac{W_1}{\alpha}+c_1}\abs*{x - y}^\alpha.
  \end{align}
\end{proof}
\begin{proof}[Proof of \cref{lem:div-speed-holder2}]
  The combination of \cref{lem:lowest-int-div} and absolutely continuity of $\phi$ yields that for any $x,y \in (0,1)$
  \begin{align}
      \abs*{\phi(x) - \phi(y)} =& \abs*{\int_x^y \phi^{(1)}(s) ds} \\
      \le& \abs*{\int_x^y\abs*{\phi^{(1)}(s)}ds} \\
      \le& \abs*{\int_x^y\paren*{W_1\ln(1/s)+c_1}ds} \\
      \le& W_1\abs*{\int_x^y\ln(1/s)ds} + c_1\abs*{x - y}.
  \end{align}
  Application of the H\"older inequality yields
  \begin{align}
    & \abs*{\int_x^y\ln(1/s)ds} \\
    \le& \abs*{\abs*{\int_x^y \ln^{\frac{1}{1-\alpha}}(1/s)ds}^{1-\alpha}\abs*{\int_x^y 1 ds}^\alpha} \\
     =& \abs*{\abs*{\int_x^y \ln^{\frac{1}{1-\alpha}}(1/s)ds}^{1-\alpha}}\abs*{x - y}^\alpha \\
     =& \begin{multlined}[t]
       \abs[\Bigg]{\gamma\paren*{1+\frac{1}{1-\alpha},\ln(1/x)} \\ - \gamma\paren*{1+\frac{1}{1-\alpha},\ln(1/y)}}^{1-\alpha}\abs*{x - y}^\alpha,
     \end{multlined}
  \end{align}
  where $\gamma(s,x)=\int_0^x t^{s-1}e^{-t}dt$ which is known as the lower incomplete gamma function. Since the lower incomplete gamma function is non-decreasing for $x > 0$, and $\lim_{x \to \infty}\gamma(s,x) = \Gamma(s)$, where $\Gamma(s)$ is the gamma function, we have
  \begin{multline}
    \abs*{\phi(x) - \phi(y)} \le \\ W_1\paren*{\Gamma\paren*{1+\frac{1}{1-\alpha}}}^{1-\alpha}\abs*{x - y}^\alpha + c_1\abs*{x - y}.
  \end{multline}
  The gamma function $\Gamma(s)$ is finite if $s < \infty$. Thus,
  \begin{align}
      & \sup_{x, y \in (0,1)}\frac{\abs*{\phi(x) - \phi(y)}}{\abs*{x - y}^\alpha} \\
      \le& W_1\paren*{\Gamma\paren*{1+\frac{1}{1-\alpha}}}^{1-\alpha} + c_1\sup_{x,y \in (0,1)}\abs*{x - y}^{1-\alpha} \\ =& W_1\paren*{\Gamma\paren*{1+\frac{1}{1-\alpha}}}^{1-\alpha} + c_1 < \infty,
  \end{align}
  for $\alpha \in (0,1)$.
\end{proof}
\begin{proof}[Proof of \cref{lem:div-speed-holder3}]
  The Lipschitz continuity is proved by showing there exists an universal constant $C > 0$ such that
  \begin{align}
    \sup_{p \in (0,1)}\abs*{\phi^{(1)}(p)} \le C.
  \end{align}
  For any $p \in (0,1)$, the absolutely continuity of $\phi^{(1)}$ gives
  \begin{align}
    \abs*{\phi^{(1)}(p)} =& \abs*{\int_0^p\phi^{(2)}(s)ds} \\
    \le& \int_0^p\abs*{\phi^{(2)}(s)}ds \\
    \le& \int_0^p\paren*{W_2s^{\alpha-2} + c_2}ds \\
    =& \frac{W_2}{\alpha-1}p^{\alpha-1} + c_2p \le \frac{W_2}{\alpha-1} + c_2.
  \end{align}
  Next, we prove the H\"older continuity of $\phi^{(1)}$. The absolutely continuity of $\phi^{(1)}$ yields for any $x,y \in (0,1)$,
  \begin{align}
    \abs*{\phi^{(1)}(x) - \phi^{(1)}(y)} \le& \abs*{\int_x^y\phi^{(2)}(s)ds} \\
    \le& \abs*{\int_x^y\abs*{\phi^{(2)}(s)}ds} \\
    \le& \abs*{\int_x^y\paren*{W_2s^{\alpha-2} + c_2}ds} \\
    =& \frac{W_2}{\alpha-1}\abs*{x^{\alpha-1} - y^{\alpha-1}} + c_2\abs*{x-y}.
  \end{align}
  Since a function $x \to x^\beta$ for $\beta \in (0,1)$ is $\beta$-H\"older continuous and $\abs{x-y} \le 1$ for any $x,y \in (0,1)$, we have
  \begin{align}
    \abs*{\phi^{(1)}(x) - \phi^{(1)}(y)} \le& \frac{W_2}{\alpha-1}\abs*{x-y}^{\alpha-1} + c_1\abs*{x-y}^{\alpha-1}.
  \end{align}
\end{proof}

\section{Proofs for Upper Bound Analysis}

We use the following helper lemma for proving \cref{lem:upper-ind-var}.
\begin{lemma}[\textcite{Cai2011TestingFunctional}, Lemma 4]\label{lem:cai-lem4}
  Suppose $\ind{\event}$ is an indicator random variable independent of $X$ and $Y$, then
  \begin{multline}
    \Var\bracket*{X\ind{\event} + Y\ind{\event^c}} = \\ \Var\bracket*{X}\p\event + \Var\bracket*{Y}\p\event^c + \paren*{\Mean\bracket{X} - \Mean\bracket{Y}}^2\p\event\p\event^c.
  \end{multline}
\end{lemma}
\begin{proof}[Proof of \cref{lem:upper-ind-bias}]
  From the property of the absolute value, the bias is bounded above as
  \begin{align}
    & \Bias\bracket*{ \hat\theta\paren{\tilde{N}} - \theta(P) } \\
     \le& \begin{multlined}[t]
       \sum_{i=1}^k \paren[\Bigg]{\Bias\bracket*{\ind{\tilde{N}'_i \ge 2\Delta_{n,k}}\paren*{\phi_{\mathrm{plugin}}(\tilde{N}_i) - \phi(p_i)} } \\ + \Bias\bracket*{\ind{\tilde{N}'_i < 2\Delta_{n,k}}\paren*{\phi_{\mathrm{poly}}(\tilde{N}_i) - \phi(p_i)}} }.
     \end{multlined}
  \end{align}
  Because of the independence between $\tilde{N}$ and $\tilde{N}'$, we have
  \begin{align}
    & \Bias\bracket*{\ind{\tilde{N}'_i \ge 2\Delta_{n,k}}\paren*{\phi_{\mathrm{plugin}}(\tilde{N}_i) - \phi(p_i)} } \\
     =& \Bias\bracket*{\phi_{\mathrm{plugin}}(\tilde{N}_i) - \phi(p_i) }\p\cbrace*{\tilde{N}'_i \ge 2\Delta_{n,k}}, \\
    & \Bias\bracket*{\ind{\tilde{N}'_i < 2\Delta_{n,k}}\paren*{\phi_{\mathrm{poly}}(\tilde{N}_i) - \phi(p_i)} } \\
     =& \Bias\bracket*{ \phi_{\mathrm{poly}}(\tilde{N}_i) - \phi(p_i) }\p\cbrace*{\tilde{N}'_i < 2\Delta_{n,k}}. \label{eq:part-bias-decomp}
  \end{align}

  For the bias of the bias-corrected plugin estimator, we have
  \begin{align}
    & \Bias\bracket*{\phi_{\mathrm{plugin}}(\tilde{N}_i) - \phi(p_i) }\p\cbrace*{\tilde{N}'_i \ge 2\Delta_{n,k}} \\
    =& \begin{multlined}[t][.9\columnwidth]
      \Bias\bracket*{\phi_{\mathrm{plugin}}(\tilde{N}_i) - \phi(p_i) }\p\cbrace*{\tilde{N}'_i \ge 2\Delta_{n,k}}\\\times\paren*{\ind{np_i \le \Delta_{n,k}} + \ind{np_i > \Delta_{n,k}}}.
    \end{multlined}
  \end{align}
  \sloppy The Chernoff bound for the Poisson distribution gives $\p\cbrace*{\tilde{N}'_i \ge 2\Delta_{n,k}}\ind{np_i \le \Delta_{n,k}} \le (e/4)^{\Delta_{n,k}}\ind{np_i \le \Delta_{n,k}}$. Thus, we have
  \begin{align}
    & \Bias\bracket*{\phi_{\mathrm{plugin}}(\tilde{N}_i) - \phi(p_i) }\p\cbrace*{\tilde{N}'_i \ge 2\Delta_{n,k}} \\
    \le& \begin{multlined}[t][.8\columnwidth]
     \paren*{(e/4)^{\Delta_{n,k}}\ind{np_i \le \Delta_{n,k}} + \ind{np_i > \Delta_{n,k}}}\\\times\Bias\bracket*{\phi_{\mathrm{plugin}}(\tilde{N}_i) - \phi(p_i) }.
    \end{multlined} \label{eq:part-plugin-bias}
  \end{align}

  Since the Chernoff bound yields $\p\cbrace*{\tilde{N}'_i < 2\Delta_{n,k}} \le e^{-\Delta_{n,k}/8}$ for $p_i > 4\Delta_{n,k}$, as well as the \cref{eq:part-plugin-bias}, we have
  \begin{align}
    & \Bias\bracket*{ \phi_{\mathrm{poly}}(\tilde{N}_i) - \phi(p_i) }\p\cbrace*{\tilde{N}'_i < 2\Delta_{n,k}} \\
    \le& \begin{multlined}[t][.9\columnwidth]
     \Bias\bracket*{\phi_{\mathrm{poly}}(\tilde{N}_i) - \phi(p_i) }\p\cbrace*{\tilde{N}'_i < 2\Delta_{n,k}}\ind{np_i \le 4\Delta_{n,k}} + \\ \Bias\bracket*{\phi_{\mathrm{poly}}(\tilde{N}_i) - \phi(p_i) }\p\cbrace*{\tilde{N}'_i < 2\Delta_{n,k}}\ind{np_i > 4\Delta_{n,k}}
    \end{multlined} \\
    \le& \begin{multlined}[t][.8\columnwidth]
      \paren*{\ind{np_i \le 4\Delta_{n,k}} + e^{-\Delta_{n,k}/8}\ind{np_i < 4\Delta_{n,k}}}\\\times\Bias\bracket*{\phi_{\mathrm{poly}}(\tilde{N}_i) - \phi(p_i) }.
    \end{multlined}\label{eq:part-poly-bias}
  \end{align}
  Combining \cref{eq:part-bias-decomp,eq:part-plugin-bias,eq:part-poly-bias} gives the desired result.
\end{proof}
\begin{proof}[Proof of \cref{lem:upper-ind-var}]
 \sloppy Because of the independence of $\tilde{N}_1,..,\tilde{N}_k,\tilde{N}'_1,...,\tilde{N}'_k$, applying \cref{lem:cai-lem4} gives
  \begin{align}
    & \Var\bracket*{\hat\theta\paren{\tilde{N}}} \\
    =& \Var\bracket*{ \sum_{i=1}^k \ind{\tilde{N}'_i \ge 2\Delta_{n,k}}\phi_{\mathrm{plugin}}(\tilde{N}_i) + \ind{\tilde{N}'_i < 2\Delta_{n,k}}\phi_{\mathrm{poly}}(\tilde{N}_i) } \\
    =& \sum_{i=1}^k \Var\bracket*{ \ind{\tilde{N}'_i \ge 2\Delta_{n,k}}\phi_{\mathrm{plugin}}(\tilde{N}_i) + \ind{\tilde{N}'_i < 2\Delta_{n,k}}\phi_{\mathrm{poly}}(\tilde{N}_i) } \\
   \le& \begin{multlined}[t][.8\columnwidth]
    \sum_{i=1}^k \paren[\Bigg]{ \Var\bracket*{\phi_{\mathrm{plugin}}(\tilde{N}_i)}\p\cbrace*{\tilde{N}'_i \ge 2\Delta_{n,k}} \\
    + \Var\bracket*{\phi_{\mathrm{poly}}(\tilde{N}_i)}\p\cbrace*{\tilde{N}'_i < 2\Delta_{n,k}} \\
    + \paren*{\Mean\bracket*{\phi_{\mathrm{plugin}}(\tilde{N}_i)} - \Mean\bracket*{\phi_{\mathrm{poly}}(\tilde{N}_i)}}^2\\\times\p\cbrace*{\tilde{N}'_i \ge 2\Delta_{n,k}}\p\cbrace*{\tilde{N}'_i < 2\Delta_{n,k}} }.
   \end{multlined} \label{eq:divided-var}
  \end{align}

  We can derive upper bounds on the first two terms of \cref{eq:divided-var} in the same manner of \cref{eq:part-plugin-bias,eq:part-poly-bias} as
  \begin{multline}
    \Var\bracket*{\phi_{\mathrm{plugin}}(\tilde{N}_i)}\p\cbrace*{\tilde{N}'_i \ge 2\Delta_{n,k}}
     \le \\ \paren*{(e/4)^{\Delta_{n,k}}\ind{np_i \le \Delta_{n,k}} + \ind{np_i > \Delta_{n,k}}}\Var\bracket*{\phi_{\mathrm{plugin}}(\tilde{N}_i) },
  \end{multline}
  and
  \begin{multline}
    \Var\bracket*{\phi_{\mathrm{poly}}(\tilde{N}_i)}\p\cbrace*{\tilde{N}'_i < 2\Delta_{n,k}}
     \le \\ \paren*{\ind{np_i \le 4\Delta_{n,k}} + e^{-\Delta_{n,k}/8}\ind{np_i > 4\Delta_{n,k}}}\Var\bracket*{\phi_{\mathrm{poly}}(\tilde{N}_i)}.
  \end{multline}

  For the third term of \cref{eq:divided-var}, application of the triangle inequality yields
  \begin{align}
    & \begin{multlined}[t][.9\columnwidth]
     \paren*{\Mean\bracket*{\phi_{\mathrm{plugin}}(\tilde{N}_i)} - \Mean\bracket*{\phi_{\mathrm{poly}}(\tilde{N}_i)}}^2\\\times\p\cbrace*{\tilde{N}'_i \ge 2\Delta_{n,k}}\p\cbrace*{\tilde{N}'_i < 2\Delta_{n,k}}
    \end{multlined}\\
    =& \begin{multlined}[t][.9\columnwidth]
      \paren*{\Mean\bracket*{\phi_{\mathrm{plugin}}(\tilde{N}_i) - \phi(p_i)} - \Mean\bracket*{\phi_{\mathrm{poly}}(\tilde{N}_i) - \phi(p_i)}}^2\\\times\p\cbrace*{\tilde{N}'_i \ge 2\Delta_{n,k}}\p\cbrace*{\tilde{N}'_i < 2\Delta_{n,k}}
    \end{multlined}\\
    \le& \begin{multlined}[t][.9\columnwidth]
      2\paren[\Bigg]{\Bias\bracket*{\phi_{\mathrm{plugin}}(\tilde{N}_i) - \phi(p_i)}^2 \\ + \Bias\bracket*{\phi_{\mathrm{poly}}(\tilde{N}_i) - \phi(p_i)}^2}\\\times\p\cbrace*{\tilde{N}'_i \ge 2\Delta_{n,k}}\p\cbrace*{\tilde{N}'_i < 2\Delta_{n,k}}.
    \end{multlined}
  \end{align}
  As well as \cref{eq:part-plugin-bias,eq:part-poly-bias}, we have
  \begin{align}
    & \p\cbrace*{\tilde{N}'_i \ge 2\Delta_{n,k}}\p\cbrace*{\tilde{N}'_i < 2\Delta_{n,k}} \\
    \le& (e/4)^{\Delta_{n,k}}\ind{np_i \le \Delta_{n,k}} + \ind{np_i > \Delta_{n,k}},
  \end{align}
  or
  \begin{align}
    & \p\cbrace*{\tilde{N}'_i \ge 2\Delta_{n,k}}\p\cbrace*{\tilde{N}'_i < 2\Delta_{n,k}} \\
    \le& \ind{np_i \le 4\Delta_{n,k}} + e^{-\Delta_{n,k}/8}\ind{np_i > 4\Delta_{n,k}}.
  \end{align}
  Assigning these bounds into \cref{eq:divided-var}, we get the desired result.
\end{proof}

\begin{proof}[Proof of \cref{lem:poly-o1}]
  From the triangle inequality, we have
 \begin{align}
  & \Bias\bracket*{(g_L(\tilde{N}) \land \phi_{{\mathrm{sup}},\Delta})\lor \phi_{{\mathrm{inf}},\Delta} - \phi(p) } \\
  \le& \abs*{\phi_{{\mathrm{sup}},\Delta}}\lor\abs*{\phi_{{\mathrm{inf}},\Delta}} + \sup_{p \in [0,1]}\abs*{\phi(p)}.
 \end{align}
 By the assumption, there exists a finite universal constant $C > 0$ such that $\abs*{\phi_{{\mathrm{sup}},\Delta}}\lor\phi_{{\mathrm{inf}},\Delta} \le C$ and $\sup_{p \in [0,1]}\abs*{\phi(p)} \le C$.

 By the last truncation, we have
 \begin{align}
    \Var\bracket*{(g_L(\tilde{N}) \land \phi_{{\mathrm{sup}},\Delta})\lor \phi_{{\mathrm{inf}},\Delta} } \le \abs*{\phi_{{\mathrm{sup}},\Delta}}\lor\abs*{\phi_{{\mathrm{inf}},\Delta}}.
 \end{align}
 With the same reason of the bias, we obtain the desired claim.
\end{proof}

To prove the variance upper bound in \cref{lem:poly-var}, we use the following lemma:
\begin{lemma}[\textcite{Wu2016MinimaxApproximation}]\label{lem:poi-var}
  Let $X \sim \Poi(\lambda)$. For any positive integer $m$, $\Var[(X)_m]$ is increasing in $\lambda$ and
  \begin{align}
    \Var\bracket*{\paren*{X}_m} \le (\lambda m)^m\paren*{\frac{(2e)^{2\sqrt{\lambda m}}}{\pi\sqrt{\lambda m}} \lor 1}.
  \end{align}
\end{lemma}
\begin{proof}[Proof of \cref{lem:poly-bias}]
  Let $\phi'_{{\mathrm{sup}},\Delta} = \phi_{{\mathrm{sup}},\Delta} \lor \sup_{p \in [0,\Delta]}\phi_L(p)$ and $\phi'_{{\mathrm{inf}},\Delta} = \phi_{{\mathrm{inf}},\Delta} \land \inf_{p \in [0,\Delta]}\phi_L(p)$. By the triangle inequality and the fact that $g_L$ is an unbiased estimator of $\phi_L$, we have
 \begin{multline}
 \Bias\bracket*{(g_L(\tilde{N}) \land \phi_{{\mathrm{sup}},\Delta})\lor \phi_{{\mathrm{inf}},\Delta} - \phi(p) } \\
  \le \Bias\bracket[\Bigg]{(g_L(\tilde{N}) \land \phi_{{\mathrm{sup}},\Delta})\lor \phi_{{\mathrm{inf}},\Delta} \\ - (g_L(\tilde{N}) \land \phi'_{{\mathrm{sup}},\Delta})\lor \phi'_{{\mathrm{inf}},\Delta} } \\ + \Bias\bracket*{(g_L(\tilde{N}) \land \phi'_{{\mathrm{sup}},\Delta})\lor \phi'_{{\mathrm{inf}},\Delta} - \phi_L(p) }  \\ + \Bias\bracket*{g_L(\tilde{N}) - \phi(p) }.
 \end{multline}
 The first term is bounded above as
 \begin{align}
   & \begin{multlined}[t][.9\columnwidth]
    \Bias\bracket[\Bigg]{(g_L(\tilde{N}) \land \phi_{{\mathrm{sup}},\Delta})\lor \phi_{{\mathrm{inf}},\Delta} \\ - (g_L(\tilde{N}) \land \phi'_{{\mathrm{sup}},\Delta})\lor \phi'_{{\mathrm{inf}},\Delta} }
   \end{multlined}\\
   \le& (\phi'_{{\mathrm{sup}},\Delta} - \phi_{{\mathrm{sup}},\Delta})\lor(\phi_{{\mathrm{inf}},\Delta} - \phi'_{{\mathrm{inf}},\Delta}) \\
   \le& \sup_{p \in [0,1]}\abs*{\phi_L(p) - \phi(p)} = E_L(\phi, [0,\Delta]).
 \end{align}
 Also, the third term is bounded above as
 \begin{align}
  \Bias\bracket*{g_L(\tilde{N}) - \phi(p) } =& \abs*{\phi_L(p) - \phi(p)} \le E_L\paren*{\phi, [0,\Delta]}.
 \end{align}
 The second term has upper bound as
 \begin{align}
 & \Bias\bracket*{(g_L(\tilde{N}) \land \phi'_{{\mathrm{sup}},\Delta})\lor \phi'_{{\mathrm{inf}},\Delta} - \phi_L(p) } \\
  =& \sqrt{\paren*{\Mean\bracket*{(g_L(\tilde{N}) \land \phi'_{{\mathrm{sup}},\Delta})\lor \phi'_{{\mathrm{inf}},\Delta} - \phi_L(p)}}^2} \\
  \le& \sqrt{\Mean\bracket*{\paren*{(g_L(\tilde{N}) \land \phi'_{{\mathrm{sup}},\Delta})\lor \phi'_{{\mathrm{inf}},\Delta} - \phi_L(p) }^2}}.
 \end{align}
 Since $\phi_L(p) \in [\phi'_{{\mathrm{inf}},\Delta},\phi'_{{\mathrm{sup}},\Delta}]$ for $p \in [0,\Delta]$, we have $\paren*{(g_L(\tilde{N}) \land \phi'_{{\mathrm{sup}},\Delta})\lor \phi'_{{\mathrm{inf}},\Delta} - \phi_L(p) }^2 \le \paren*{g_L(\tilde{N}) - \phi_L(p) }^2$. Thus, we have
 \begin{multline}
   \Bias\bracket*{(g_L(\tilde{N}) \land \phi'_{{\mathrm{sup}},\Delta})\lor \phi'_{{\mathrm{inf}},\Delta} - \phi_L(p) } \\ \le \sqrt{\Var\bracket*{g_L(\tilde{N})}}.
 \end{multline}
\end{proof}
\begin{proof}[Proof of \cref{lem:poly-var}]
 Letting $\phi_\Delta(p) = \phi(\Delta x)$ and $a_0,...,a_L$ be coefficients of the optimal uniform approximation of $\phi_\Delta$ by degree-$L$ polynomials on $[0,1]$, we have $\sum_{m=0}^L \frac{\Delta^m a_m}{n^m}(\tilde{N})_m = g_L(\tilde{N})$. Then, since the standard deviation of sum of random variables is at most the sum of individual standard deviation, we have
 \begin{align}
   \Var\bracket*{g_L(\tilde{N}) - \phi(p) }
    \le \paren*{\sum_{m=1}^L \frac{\Delta^m \abs*{a_m}}{n^m}\sqrt{\Var \paren{\tilde{N}}_m } }^2.
 \end{align}
 From \autocite{Petrushev1988RationalFunctions} and the assumption that $\phi$ is bounded, there is a positive constant $C$ such that $\abs*{a_m} \le C2^{3L}$. From \cref{lem:poi-var} and the assumption $p \le \Delta$, $\Var\paren{\tilde{N}}_m \le \Var\paren*{X}_m$ where $X \sim \Poi(\Delta n)$. Thus, we have
 \begin{align}
  &  \Var\bracket*{g_L(\tilde{N})} \\
    \lesssim& \paren*{\sum_{m=1}^L \frac{\Delta^m 2^{3L}}{n^{m}}\sqrt{ (\Delta n L )^m(2e)^{2\sqrt{\Delta n L}} } }^2 \\
    \le& \paren*{\sum_{m=1}^L \sqrt{\frac{\Delta^{3m} L^{m}}{n^{m}}} 2^{3L}(2e)^{\sqrt{\Delta n L}} }^2.
 \end{align}
 From the assumption $\frac{\Delta^3 L}{n} \le \frac{1}{2}$, we have
 \begin{align}
   & \paren*{\sum_{m=1}^L c^m\sqrt{\frac{\Delta^{3m} L^{m}}{n^{m}}} 2^{3L}(2e)^{\sqrt{\Delta n L}} }^2 \\
   \le& \paren*{2^{3L}(2e)^{\sqrt{\Delta n L}}\sum_{m=1}^L \paren*{\sqrt{\frac{\Delta^{3} L}{n}}}^m }^2 \\
   \le& \paren*{2^{3L}(2e)^{\sqrt{\Delta n L}}\paren*{\sqrt{\frac{\Delta^{3} L}{n}} + \int_1^L \paren*{\sqrt{\frac{\Delta^{3} L}{n}}}^x dx} }^2 \\
   \le& \begin{multlined}[t][.9\columnwidth]
     \paren[\Bigg]{2^{3L}(2e)^{\sqrt{\Delta n L}}\paren[\Bigg]{\sqrt{\frac{\Delta^{3} L}{n}} + \\ \frac{2}{\ln\paren*{\frac{\Delta^{3} L}{n}}}\paren*{\paren*{\sqrt{\frac{\Delta^{3} L}{n}}}^L - \sqrt{\frac{\Delta^{3} L}{n}}}} }^2
   \end{multlined}\\
   \le& \begin{multlined}[t][.9\columnwidth]
    \paren[\Bigg]{\sqrt{\frac{\Delta^{3} L}{n}}2^{3L}(2e)^{\sqrt{\Delta n L}}\paren[\Bigg]{1 + \\ \frac{2}{\ln 2}\paren*{1 - \paren*{\sqrt{\frac{\Delta^{3} L}{n}}}^{L-1}}} }^2 
    \end{multlined}\\
   \le& \frac{16\Delta^3L64^L(2e)^{2\sqrt{\Delta n L}}}{n}.
 \end{align}
\end{proof}

\begin{proof}[Proof of \cref{lem:plugin-o1-2}]
  We have
  \begin{multline}
    \Bias\bracket*{\bar\phi_{2,\Delta}\paren*{\frac{\tilde{N}}{n}} - \phi(p) } \le \\ \sup_{p > 0}\abs*{T_\Delta[\phi](p)} + \sup_{p > 0}\frac{p}{2n}\abs*{T^{(2)}_\Delta[\phi](p)} + \sup_{p \in [0,1]}\abs*{\phi(p)},
  \end{multline}
  and
  \begin{multline}
    \Var\bracket*{ \bar\phi_{2,\Delta}\paren*{\frac{\tilde{N}}{n}} } \le  \\ \sup_{p > 0}\paren*{T_\Delta[\phi](p)}^2 + \sup_{p > 0}\paren*{\frac{p}{2n}T^{(2)}_\Delta[\phi](p)}^2.
  \end{multline}
  From \cref{lem:div-speed-holder1,lem:div-speed-holder2,lem:div-speed-holder3}, since $\phi$ is $\beta$-H\"older continuous for some $\beta \in (0,1]$, $\phi$ is a continuous and bounded function. Hence, for any $p \in (0,1]$,
  \begin{align}
     \abs*{\phi(p)} \lesssim 1.
  \end{align}
  Moreover, for any $p > 0$,
  \begin{align}
      \abs*{T_\Delta[\phi](p)} \lesssim 1.
  \end{align}
  From the divergence speed assumption, we have
  \begin{align}
      \sup_{p > 0}\abs*{pT^{(2)}_\Delta[\phi](p)} \lesssim \Delta^{\alpha-1}\lor 1.
  \end{align}
  Under the assumptions $\Delta \in (0,1)$ and $\Delta \gtrsim n^{-1}$, we have $\Delta^{\alpha-1} \lesssim n^{1-\alpha}\lor 1$ for $\alpha \in (0,2)$. Hence,
  \begin{align}
      \sup_{p > 0}\abs*{\frac{p}{2n}T^{(2)}_\Delta[\phi](p)} \lesssim n^{-\alpha} \lor 1 \lesssim 1.
  \end{align}
\end{proof}
\begin{proof}[Proof of \cref{lem:plugin-o1-4}]
  In the same manner of the proof of \cref{lem:plugin-o1-2}, it suffices to prove the claim by showing $\sup_{p > 0}\abs{\frac{2p}{3n^2}T^{(3)}_\Delta[\phi](p)} \lesssim 1$, $\sup_{p > 0}\abs{\frac{7p}{24n^3}T^{(4)}_\Delta[\phi](p)} \lesssim 1$, and $\sup_{p > 0}\abs{\frac{3p^2}{8n^2}T^{(4)}_\Delta[\phi](p)} \lesssim 1$. By the divergence speed assumption and \cref{lem:lower-div}, we have
  \begin{align}
      \sup_{p > 0}\abs*{pT^{(3)}_\Delta[\phi](p)} \lesssim& \Delta^{\alpha-2} \lor 1, \\
      \sup_{p > 0}\abs*{pT^{(4)}_\Delta[\phi](p)} \lesssim& \Delta^{\alpha-3} \lor 1, \\
      \sup_{p > 0}\abs*{p^2T^{(4)}_\Delta[\phi](p)} \lesssim& \Delta^{\alpha-2} \lor 1.
  \end{align}
  Under the assumptions $\Delta \in (0,1)$ and $\Delta \gtrsim n^{-1}$, we have
  \begin{align}
      \sup_{p > 0}\abs*{pT^{(3)}_\Delta[\phi](p)} \lesssim& n^{2-\alpha} \lor 1, \\
      \sup_{p > 0}\abs*{pT^{(4)}_\Delta[\phi](p)} \lesssim& n^{3-\alpha} \lor 1, \\
      \sup_{p > 0}\abs*{p^2T^{(4)}_\Delta[\phi](p)} \lesssim& n^{2-\alpha} \lor 1.
  \end{align}
  Hence,
  \begin{align}
      \sup_{p > 0}\abs*{\frac{2p}{3n^2}T^{(3)}_\Delta[\phi](p)} \lesssim& n^{-\alpha} \lor 1 \lesssim 1, \\
      \sup_{p > 0}\abs*{\frac{7p}{24n^3}T^{(4)}_\Delta[\phi](p)} \lesssim& n^{-\alpha} \lor 1 \lesssim 1, \\
      \sup_{p > 0}\abs*{\frac{3p^2}{8n^2}T^{(4)}_\Delta[\phi](p)} \lesssim& n^{-\alpha} \lor 1 \lesssim 1.
  \end{align}
\end{proof}

\begin{proof}[Proof of \cref{lem:converge-hermite}]
  Assume $\Delta < p < 1$. It is obvious that for any $\Delta < p < 1$, $H^{(\ell)}_{L,\Delta,\delta}[\phi](p) = \phi^{(\ell)}(p) = T_\Delta[\phi](p)$ for any $\ell = 0,...,L$ and any $\delta > 0$.

  By a property of the generalized Hermite interpolation, $H^{(\ell)}_L(a;\phi,a,b) = \phi^{(\ell)}(a)$ for any $\delta > 0$ and any $\ell = 0,...,L$. Hence, if $p = \Delta$ or $p = 1$, $H^{(\ell)}_{L,\Delta,\delta}[\phi](p) = \phi^{(\ell)}(p) = T_\Delta[\phi](p)$ for any $\delta > 0$ and any $\ell = 0,...,L$.

  Assume $p < \Delta$. For any $\delta < 1 - \Delta^{-1}p$,
  \begin{align}
      H_{L,\Delta,\delta}[\phi](p) = \phi(\Delta) = T_\Delta[\phi](p).
  \end{align}
  Besides, for any $\delta < 1 - \Delta^{-1}p$ and any $\ell \le L$,
  \begin{align}
      H^{(\ell)}_{L,\Delta,\delta}[\phi](p) = 0 = T_\Delta[\phi](p).
  \end{align}
  In the same manner, we can prove the claim for $p > 1$.
\end{proof}

\begin{proof}[Proof of \cref{lem:hermite-bound}]
  By the divergence speed assumption and the fact that $p^{\alpha-\ell} \ge 1$ for any $p \in (0,1)$, we have $\abs{\phi^{(\ell)}(p)} \le (W_\ell+c_\ell)p^{\alpha-\ell}$. Thus, it is clear that $p^\beta\abs*{H^{(\ell)}_{L,\Delta,\delta}[\phi](p)} \le (W_\ell+c_\ell)p^{\alpha+\beta-\ell}$ for $p \in [\Delta,1]$.

  Fix $\delta > 0$. For $p \le \Delta(1-\delta)$ and $p \ge 1 + \delta$, $\abs*{H^{(\ell)}_{L,\Delta}[\phi](p)} = 0$ by definition. For $p \in (1,1+\delta)$, we have
  \begin{align}
    & p^\beta\abs*{H^{(\ell)}_L\paren*{p;\phi,1,1+\delta}} \\
    =& \begin{multlined}[t][.9\columnwidth]
      p^\beta\abs[\Bigg]{\sum_{u=0}^\ell\binom{\ell}{u}\sum_{m=1\lor u}^{L}\frac{\phi^{(m)}(1)}{m!}(p-1)^{m-u}\sum_{s=0}^{L-m}\frac{L+1}{L+s+1}\\\times\paren*{\prod_{w=0}^{\ell-u-1}(L+s+1-w)}\sum_{w=\ell-u-s}^{\ell-u}(-1)^w\\\times\binom{\ell-u}{w}\Beta_{s-\ell+u+w,L+s+1-\ell+u}\paren*{\frac{p-1}{\delta}}}
    \end{multlined} \\
    \le& \begin{multlined}[t][.9\columnwidth]
      (L+1)p^\beta\sum_{u=0}^\ell\binom{\ell}{u}\sum_{m=1\lor u}^{L}\frac{\abs*{\phi^{(m)}(1)}}{m!}\delta^{m-u}\\\times\sum_{s=0}^{L-m}\paren*{\prod_{w=1}^{\ell-u-1}(L+s+1-w)}\\\times\sum_{w=\ell-u-s}^{\ell-u}\binom{\ell-u}{w}\abs*{\Beta_{s-\ell+u+w,L+s+1-\ell+u}\paren*{\frac{p-1}{\delta}}}
    \end{multlined}\\
    \le& \begin{multlined}[t][.9\columnwidth]
      (L+1)p^\beta\sum_{u=0}^\ell\binom{\ell}{u}\sum_{m=1\lor u}^{L}\frac{W_m+c_m}{m!}\delta^{m-u}\\\times\sum_{s=0}^{L-m}\paren*{\prod_{w=1}^{\ell-u-1}(L+s+1-w)}\\\times\sum_{w=\ell-u-s}^{\ell-u}\binom{\ell-u}{w}\binom{L+s+1-\ell+u}{s-\ell+u+w}\\\times\paren*{\frac{s-\ell+u+w}{L+s+1-\ell+u}}^{s-\ell+u+w}\\\times\paren*{\frac{L+1-w}{L+s+1-\ell+u}}^{L+1-w},
    \end{multlined}
  \end{align}
  where we use the fact that $\sup_x\abs*{\Beta_{\nu,n}(x)} = \nu^\nu n^{-n}(n-\nu)^{n-\nu}{\binom n \nu}$ to obtain the last line. Letting
  \begin{multline}
    C_{L,m,u} = \sum_{s=0}^{L-m}\paren*{\prod_{w=1}^{\ell-u-1}(L+s+1-w)}\\\times\sum_{w=\ell-u-s}^{\ell-u}\binom{\ell-u}{w}\binom{L+s+1-\ell+u}{s-\ell+u+w}\\\times\paren*{\frac{s-\ell+u+w}{L+s+1-\ell+u}}^{s-\ell+u+w}\\\times\paren*{\frac{L+1-w}{L+s+1-\ell+u}}^{L+1-w},
  \end{multline}
  we have
  \begin{align}
    & p^\beta\abs*{H^{(\ell)}_L\paren*{p;\phi,1,1+\delta}} \\
    \le& (L+1)(1+\delta)^\beta\sum_{u=0}^\ell\binom{\ell}{u}\sum_{m=1\lor u}^{L}C_{L,m,u}\frac{W_m+c_m}{m!}\delta^{m-u} \\
    \le& (L+1)2^\beta\sum_{u=0}^\ell\binom{\ell}{u}\sum_{m=1\lor u}^{L}C_{L,m,u}\frac{W_m+c_m}{m!} \\
    \le& (L+1)2^{\ell+\beta}e\max_{u=0,...,\ell,m=1\lor m,...,L}C_{L,m,u}(W_m+c_m).
  \end{align}
  Noting that $\max_{u=0,...,\ell,m=1\lor m,...,L}C_{L,m,u}(W_m+c_m)$ is only depending on the universal constant $L$ and the function $\phi$, we have $p^\beta\abs*{H^{(\ell)}_L\paren*{p;\phi,1,1+\delta}} \lesssim 1$.

  Similarly, for $p \in (\Delta(1-\delta),\Delta)$,
  \begin{align}
    & p^\beta\abs*{H^{(\ell)}_L\paren*{p;\phi,\Delta,\Delta(1-\delta)}} \\
    \le& (L+1)p^\beta\sum_{u=0}^\ell\binom{\ell}{u}\sum_{m=1\lor u}^{L}C_{L,m,u}\frac{\abs*{\phi^{(m)}(\Delta)}}{m!}\paren*{\Delta-p}^{m-u},
  \end{align}
  where $C_{L,m,u}$ are the same constants introduced in the case $p \in (1,1+\delta)$. From the bound coming from the divergence speed assumption, we have
  \begin{align}
    & p^\beta\abs*{H^{(\ell)}_L\paren*{p;\phi,\Delta,\Delta(1-\delta)}} \\
    \le& \begin{multlined}[t][.9\columnwidth]
      (L+1)\Delta^\beta\sum_{u=0}^\ell\binom{\ell}{u}\\\times\sum_{m=1\lor u}^{L}C_{L,m,u}\frac{(W_m+c_m)\Delta^{\alpha-m}}{m!}\Delta^{m-u}\delta^{m-u} 
    \end{multlined}\\
    =& \begin{multlined}[t][.9\columnwidth]
     (L+1)\Delta^{\alpha+\beta-\ell}\sum_{u=0}^\ell\binom{\ell}{u}\\\times\sum_{m=1\lor u}^{L}C_{L,m,u}\frac{W_m+c_m}{m!}\Delta^{\ell-u}\delta^{m-u} 
     \end{multlined}\\
    \le& (L+1)2^\ell e \Delta^{\alpha+\beta-\ell}\max_{u=0,...,\ell,m=1\lor m,...,L}C_{L,m,u}(W_m+c_m) \\
    \lesssim& \Delta^{\alpha+\beta-\ell}.
  \end{align}
  Hence, we have $p^\beta\abs*{H^{(\ell)}_{L,\Delta}[\phi](p)} \lesssim 1\lor\Delta^{\alpha+\beta-\ell}$.
\end{proof}

\begin{proof}[Proof of \cref{lem:reminder-bound1}]
 By the mean value theorem, for any continuously differentiable function $G(x)$ such that $\abs*{G(x)} > 0$ on the open interval between $z$ and $a$, there exists $\xi$ between $z$ and $a$ such that
 \begin{align}
     R_{2\ell-1}(z;g,a) =& \frac{g^{(2\ell)}(\xi)}{(2\ell-1)!}(z - \xi)^{2\ell-1}\frac{G(z) - G(a)}{G^{(1)}(\xi)}
 \end{align}
 Let $\hat{p} = \tilde{N}/n$ and $G(x) = x^{-r}(\hat{p} - x)^{2\ell}$. Then, we have
 \begin{align}
     G^{(1)}(x) =  -\paren*{(2\ell-r)x + r\hat{p}} \frac{(\hat{p} - x)^{2\ell-1}}{x^{r+1}}.
 \end{align}
 Hence, there exists $\xi$ between $p$ and $\hat{p}$ such that
 \begin{align}
     R_{2\ell-1}(\hat{p};g,p) =& \frac{\xi^{r+1}g^{(2\ell)}(\xi)}{(2\ell-1)!}\frac{(\hat{p} - p)^{2\ell}}{p^{r}((2\ell-r)\xi + r\hat{p})}.
 \end{align}
 Noting that $p, \hat{p}, \xi \ge 0$, we have
 \begin{align}
     & \abs*{\Mean\bracket*{R_{2\ell-1}(\hat{p};g,p)}} \\
     \le& \frac{p^{-r}}{(2\ell-r)(2\ell-1)!}\abs*{\Mean\bracket*{\xi^{r}g^{(2\ell)}(\xi)(\hat{p} - p)^{2\ell}}} \\
     \le& \frac{p^{-r}}{(2\ell-r)(2\ell-1)!}\Mean\bracket*{(\hat{p} - p)^{2\ell}}\sup_{\xi > 0}\abs*{\xi^{r}g^{(2\ell)}(\xi)}. \label{eq:reminder-moment}
 \end{align}

 The central moments of the Poisson distribution has the recursive formula. Let $X \sim \Poi(\lambda)$ and $\mu_m = \Mean[(X-\lambda)^m]$. Then, we have
 \begin{align}
     \mu_{m+1} = \lambda\paren*{\frac{d\mu_m}{d\lambda} + m\mu_{m-1}},
 \end{align}
 where $\mu_1 = 0$ and $\mu_2 = \lambda$. From this formula, the central moments of the Poisson distribution forms as a polynomial in $\lambda$. Thus, there are an integer $\ell_m$ and a sequence$a_{m,1},...,a_{m,\ell_m}$ such that
 \begin{align}
     \mu_m = \sum_{i=1}^{\ell_m}a_{m,i}\lambda^i.
 \end{align}
 From the recursive formula, we have for any $\lambda$,
 \begin{align}
     \sum_{i=1}^{\ell_{m+1}}a_{m+1,i}\lambda^i = \sum_{i=1}^{\ell_m}ma_{m,i}\lambda^i + \sum_{i=2}^{\lambda_{k-1}+1}ma_{m-1,i-1}\lambda^i.
 \end{align}
 The degree of polynomials in the left-hand and right-hand sides should be matched. Thus, we have $\ell_{m+1} = \ell_m \lor (\ell_{m-1} + 1)$. Noting that $\ell_1 = 0$ and $\ell_2 = 1$, we have $\ell_m = \floor{m/2}$. Consequently, we have $p^{-r}\Mean\bracket*{(\hat{p} - p)^{2\ell}} \lesssim \sum_{m=0}^{\ell-1}\frac{p^{\ell-m-r}}{n^{\ell+m}} \lesssim \frac{p^{\ell-r}}{n^\ell}$ because of the assumption $p \gtrsim n^{-1}$. We obtain the claim by substituting this bound into \cref{eq:reminder-moment}.
\end{proof}

\begin{proof}[Proof of \cref{lem:reminder-bound2}]
 Let $\hat{p} = \tilde{N}/n$ and $G(x) = x^{-\frac{r}{2}}(\hat{p} - x)^{\ell}$. Then, we have
 \begin{align}
     G^{(1)}(x) = -\paren*{\paren*{\ell-\frac{r}{2}}x + \frac{r}{2}\hat{p}} x^{-(1+\frac{r}{2})}(\hat{p} - x)^{\ell-1}
 \end{align}
 By the mean value theorem, there exists $\xi$ between $p$ and $\hat{p}$ such that
 \begin{align}
     \paren*{R_{\ell-1}(\hat{p};g,p)}^2 =& \frac{\xi^{2+r}\paren*{g^{(\ell)}(\xi)}^2}{(\ell-1)!}\frac{4(\hat{p} - p)^{2\ell}}{p^{r}((2\ell-r)\xi + r\hat{p})^2}.
 \end{align}
 Noting that $p, \hat{p}, \xi \ge 0$, we have
 \begin{align}
     & \Mean\bracket*{\paren*{R_{\ell-1}(\hat{p};g,p)}^2} \\
     \le& \frac{4p^{-r}}{(2\ell-r)(\ell-1)!}\Mean\bracket*{\xi^{r}\paren*{g^{(\ell)}(\xi)}^2(\hat{p} - p)^{2\ell}} \\
     \le& \frac{4p^{-r}}{(2\ell-r)(\ell-1)!}\Mean\bracket*{(\hat{p} - p)^{2\ell}}\sup_{\xi > 0}\abs*{\xi^{r}\paren*{g^{(\ell)}(\xi)}^2}.
 \end{align}
 We obtain the claim by substituting the bound on the central moment of the Poisson distribution $p^{-r}\Mean\bracket*{(\hat{p} - p)^{2\ell}} \lesssim \frac{p^{\ell-r}}{n^{\ell}}$.
\end{proof}

\section{Proofs for Lower Bound Analysis}

\begin{proof}[Proof of \cref{thm:lower-approximated}]
  This proof is following the same manner of the proof of \autocite[Lemma 1]{Wu2016MinimaxApproximation}. Fix $\delta > 0$. Let $\hat\theta(\cdot,n)$ be a near-minimax optimal estimator for fixed sample size $n$, i.e.,
  \begin{align}
    \sup_{P \in \dom{M}_k}\Mean\bracket*{\paren*{\hat\theta(N,n) - \theta(P)}} \le \delta + R^*(k,n;\phi).
  \end{align}
  For an arbitrary approximate distribution $P \in \dom{M}_k$, we construct an estimator
  \begin{align}
    \tilde\theta(\tilde{N}) = \hat\theta(\tilde{N}, n'),
  \end{align}
  where $\tilde{N}_i \sim \Poi(np_i)$ and $n' = \sum_i\tilde{N}_i$.

  If $\alpha \in (0,2)$ except $\alpha = 1$, from \cref{lem:div-speed-holder1,lem:div-speed-holder3}, $\phi$ is $\beta$-H\"older continuous for some $\beta \in (0,1]$. From the triangle inequality, $\beta$-H\"older continuity of $\phi$, and \cref{lem:bound-sum-alpha}, we have
  \begin{align}
    & \frac{1}{2}\paren*{\tilde\theta(\tilde{N}) - \theta(P)}^2 \\
    \le& \frac{1}{2}\paren*{\abs*{\tilde\theta(\tilde{N}) - \theta\paren*{\frac{P}{\sum_ip_i}}} + \abs*{\theta\paren*{\frac{P}{\sum_ip_i}} - \theta(P)}}^2 \\
    \le& \frac{1}{2}\paren*{\abs*{\tilde\theta(\tilde{N}) - \theta\paren*{\frac{P}{\sum_ip_i}}} + \abs*{\norm*{\phi}_{C^{H,\beta}}\sum_i\abs*{\frac{p_i}{\sum_ip_i} - p_i}^{\beta}}}^2 \\
    \le& \frac{1}{2}\paren*{\abs*{\tilde\theta(\tilde{N}) - \theta\paren*{\frac{P}{\sum_ip_i}}} + \norm*{\phi}_{C^{H,\beta}}\epsilon^{\beta}\sum_i\abs*{\frac{p_i}{\sum_ip_i}}^{\beta}}^2 \\
    \le& \frac{1}{2}\paren*{\abs*{\tilde\theta(\tilde{N}) - \theta\paren*{\frac{P}{\sum_ip_i}}} + \norm*{\phi}_{C^{H,\beta}}k^{1-\beta}\epsilon^{\beta}}^2 \\
    \le& \paren*{\tilde\theta(\tilde{N}) - \theta\paren*{\frac{P}{\sum_ip_i}}}^2 + \norm*{\phi}^2_{C^{H,\beta}}k^{2-2\beta}\epsilon^{2\beta}. \label{eq:approx-to-exact1}
  \end{align}

  If $\alpha = 1$, we have from \cref{lem:lowest-int-div} that for any $x,y \in (0,1)$,
  \begin{align}
      & \abs*{\phi(x) - \phi(y)} \\
      \le& \abs*{\int_x^y\abs*{\phi^{(1)}(s)}ds} \\
      \le& \abs*{\int_x^y \paren*{W_1\ln(1/s) + c_1} ds} \\
       =& \abs*{W_1\paren*{x\ln x - y\ln y} + \paren*{W_1+c_1}(y-x)} \\
       \le&  W_1\abs*{x\ln x - y\ln y} + (W_1+c_1)\abs*{x - y}.
  \end{align}
  As well as the case $\alpha \ne 1$, we have
  \begin{align}
    & \frac{1}{2}\paren*{\tilde\theta(\tilde{N}) - \theta(P)}^2 \\
    \le& \frac{1}{2}\paren*{\abs*{\tilde\theta(\tilde{N}) - \theta\paren*{\frac{P}{\sum_ip_i}}} + \abs*{\theta\paren*{\frac{P}{\sum_ip_i}} - \theta(P)}}^2 \\
    \le& \paren*{\abs*{\tilde\theta(\tilde{N}) - \theta\paren*{\frac{P}{\sum_ip_i}}}}^2 + \paren*{\sum_{i=1}^k\abs*{\phi\paren*{\frac{p_i}{\sum_ip_i}} - \phi(p_i)}}^2.
  \end{align}
  For the second term, we have
  \begin{align}
    & \paren*{\sum_{i=1}^k\abs*{\phi\paren*{\frac{p_i}{\sum_ip_i}} - \phi(p_i)}}^2 \\
    \le& \begin{multlined}[t][.9\columnwidth]
     \paren[\Bigg]{\sum_{i=1}^k\paren[\Bigg]{W_1\abs*{\frac{p_i}{\sum_ip_i}\ln\paren*{\frac{p_i}{\sum_ip_i}} - p_i\ln p_i} \\ + (W_1+c_1)\abs*{\frac{p_i}{\sum_ip_i} - p_i}}}^2
    \end{multlined}\\
    \le& \begin{multlined}[t][.9\columnwidth]
     \paren[\Bigg]{\sum_{i=1}^k\paren[\Bigg]{W_1\abs[\Bigg]{(1-\sum_ip_i)\frac{p_i}{\sum_ip_i}\ln\paren*{\frac{p_i}{\sum_ip_i}} \\ + \frac{p_i}{\sum_ip_i}\paren*{\sum_ip_i}\ln\paren*{\frac{1}{\sum_ip_i}}} \\ + (W_1+c_1)\frac{p_i}{\sum_ip_i}\abs*{1 - \sum_ip_i}}}^2 
    \end{multlined}\\
    \le& \begin{multlined}[t][.9\columnwidth]
     \paren[\Bigg]{W_1\abs*{1-\sum_ip_i}\abs[\Bigg]{\sum_{i=1}^k\frac{p_i}{\sum_ip_i}\ln\paren*{\frac{p_i}{\sum_ip_i}}} \\ + W_1\paren*{\sum_ip_i}\ln\paren*{\frac{1}{\sum_ip_i}} \\ + (W_1+c_1)\abs*{1 - \sum_ip_i}}^2 
    \end{multlined}\\
    \le& \begin{multlined}[t][.9\columnwidth]
     \paren[\Bigg]{W_1\abs*{1-\sum_ip_i}\abs[\Bigg]{\sum_{i=1}^k\frac{p_i}{\sum_ip_i}\ln\paren*{\frac{p_i}{\sum_ip_i}}} \\ + W_1\paren*{\sum_ip_i}\ln\paren*{\frac{1}{\sum_ip_i}} \\ + (W_1+c_1)\abs*{1 - \sum_ip_i}}^2 
    \end{multlined}\\
    \le& \paren*{W_1\epsilon\ln k + W_1(1+\epsilon)\ln(1+\epsilon) + (W_1+c_1)\epsilon}^2 \\
    \le& 2(W_1+c_1)^2\epsilon^2\ln^2\paren*{ek} + 2W_1^2(1+\epsilon)^2\ln^2(1+\epsilon).
  \end{align}
  Hence,
  \begin{align}
    & \frac{1}{2}\paren*{\tilde\theta(\tilde{N}) - \theta(P)}^2 \\
    \le& \begin{multlined}[t][.8\columnwidth]
     \paren*{\abs*{\tilde\theta(\tilde{N}) - \theta\paren*{\frac{P}{\sum_ip_i}}}}^2 \\ + 2(W_1+c_1)^2\epsilon^2\ln^2\paren*{ek} + 2W_1^2(1+\epsilon)^2\ln^2(1+\epsilon). \label{eq:approx-to-exact2}
    \end{multlined}
  \end{align}

  For the first term in \cref{eq:approx-to-exact1,eq:approx-to-exact2}, we observe that $\tilde{N} \sim \Mul(m, \tfrac{P}{\sum p_i})$ conditioned on $n' = m$. Therefore, we have
  \begin{align}
    & \Mean\paren*{\tilde\theta(\tilde{N}) - \theta\paren*{\frac{P}{\sum_{i=1}^k p_i}}}^2 \\
     =& \sum_{m = 0}^\infty \Mean\bracket*{\paren*{\tilde\theta(\tilde{N},m)\!-\!\theta\paren*{\frac{P}{\sum_{i=1}^k p_i}}}^2 \middle| n'\!=\!m}\p\cbrace{n'\!=\!m} \\
     \le& \sum_{m = 0}^\infty \tilde{R}^*(m,k;\phi)\p\cbrace{n' = m} + \delta.
  \end{align}
  Note that $\tilde{R}^*(m,k;\phi)$ is a non-increasing function with respect to $m$. Since $n' \sim \Poi(n\sum_i p_i)$ and $\abs*{\sum_i p_i - 1} \le \epsilon \le 1/3$, applying Chernoff bound yields $\p\cbrace*{n' \le n/2} \le e^{-n/32}$. Thus, we have
  \begin{align}
    & \Mean\paren*{\tilde\theta(\tilde{N}) - \theta\paren*{\frac{P}{\sum_{i=1}^k p_i}}}^2 \\
     \le& \begin{multlined}[t]
       \sum_{m \ge n/2} \tilde{R}^*(m,k;\phi)\p\cbrace{n' = m} \\ + \p\cbrace*{n' \le n/2}\max_{m = 0,...,\floor{n/2}}\tilde{R}^*(m,k;\phi) + \delta
     \end{multlined}\\
     \le& \tilde{R}^*(n/2,k;\phi)  + e^{-n/32}\tilde{R}^*(0,k;\phi) + \delta.
  \end{align}

  If $\alpha \in (0,2)$ except $\alpha = 1$, from $\beta$-H\"older continuity of $\phi$ and \cref{lem:bound-sum-alpha}, we have
  \begin{align}
    & \tilde{R}^*(0,k;\phi) \\
     \le& \sup_{P,P' \in \dom{M}_k}\paren*{\theta(P)-\theta(P')}^2 \\
     \le& \norm*{\phi}^2_{C^{H,\beta}}\sup_{P,P' \in \dom{M}_k}\paren*{\sum_i\abs*{p_i - p'_i}^{\beta}}^2 \\
     \le& 2\norm*{\phi}^2_{C^{H,\beta}}\sup_{P \in \dom{M}_k}\paren*{\sum_ip_i^{\beta}}^2 \\
     \le& 2\norm*{\phi}^2_{C^{H,\beta}}k^{2-2\beta}.
  \end{align}

  If $\alpha = 1$, we have
  \begin{align}
    & \tilde{R}^*(0,k;\phi) \\
    \le& \sup_{P,P' \in \dom{M}_k}\paren*{\theta(P)-\theta(P')}^2 \\
     \le& \begin{multlined}[t][.9\columnwidth]
      \sup_{P,P' \in \dom{M}_k}\paren[\Bigg]{W_1\sum_i\abs*{p_i\ln p_i - p'_i\ln p'_i} \\ + (W_1+c_1)\sum_i\abs*{p_i-p'_i}}^2
     \end{multlined}\\
     \le& 2\sup_{P \in \dom{M}_k}\abs*{W_1\sum_i\abs*{p_i\ln p_i}+(W_1+c_1)}^2 \\
     \le& 2(W_1+c_1)^2\ln^2(ek).
  \end{align}

  The arbitrariness of $\delta$ gives the desired result.
\end{proof}

\begin{proof}[Proof of \cref{thm:approx-tv-lower}]
  Put $\epsilon = 4\lambda/\sqrt{k}$. By the assumption, we have
  \begin{align}
      \abs*{\Mean\bracket*{\theta(P)} - \Mean\bracket*{\theta(P')}} =& \abs*{\sum_{i=1}^{k}\Mean\bracket*{\phi(U_i/k)} - \sum_{i=1}^{k}\Mean\bracket*{\phi(U'_i/k)}} \\
       =& k\abs*{\Mean\bracket*{\phi(U/k)} - \Mean\bracket*{\phi(U'/k)}} \ge d.
  \end{align}
  Hence, under $\event$ and $\event'$, the triangle inequality gives
  \begin{align}
      \abs*{\theta(P) - \theta(P')} \ge \frac{d}{2}.
  \end{align}

  Let $P_{\tilde{N}}$ and $P'_{\tilde{N}}$ be the marginal distributions of the histogram $\tilde{N}$ with priors $P$ and $P'$, respectively, e.g., $\tilde{N}_i \sim \Poi(np_i)$, $kp_i \sim U_i$ for $i \in [k]$, and $p_{k+1} = 1 - \beta$. Let $P_{\tilde{N}|\pi}$ and $P'_{\tilde{N}|\pi'}$ be the marginal distributions of the histogram with priors $\pi$ and $\pi'$, respectively. Then, we have from the triangle inequality that
  \begin{align}
      & \TV\paren*{P_{\tilde{N}|\pi},P'_{\tilde{N}|\pi'}} \\
      \le& \TV\paren*{P_{\tilde{N}|\pi},P'_{\tilde{N}}} + \TV\paren*{P_{\tilde{N}},P'_{\tilde{N}}} + \TV\paren*{P'_{\tilde{N}},P'_{\tilde{N}|\pi'}} \\
      =& \TV\paren*{P_{\tilde{N}},P'_{\tilde{N}}} + \p\event + \p\event'.
  \end{align}
  Since the total variation of product distribution can be upper bounded by the summation of individual ones, the total variation between $P_{\tilde{N}}$ and $P'_{\tilde{N}}$ is bounded above as
  \begin{align}
    & \TV(P_{\tilde{N}}, P'_{\tilde{N}}) \\
     \le& \sum_{i=1}^{k}\TV(P_{\tilde{N}_i}, P'_{\tilde{N}_i}) + \TV(n(1-\beta), n(1-\beta)) \\
     =& k\TV(\Mean[\Poi(nU/k)],\Mean[\Poi(nU'/k)]).
  \end{align}
  From the Le Cam’s lemma~\autocite{LeCam1986AsymptoticTheory}, we have
  \begin{multline}
      \tilde{R}^*(n,k+1,\epsilon) \ge  \frac{d^2}{16}\paren[\bigg]{1 - \\ k\TV(\Mean[\Poi(nU/k)],\Mean[\Poi(nU'/k)]) - \p\event - \p\event'}. \label{eq:approx-lower-tv}
  \end{multline}

  Next, we derive upper bounds on $\p\event$ and $\p\event'$. Applying Chebyshev's inequality and the union bound gives
  \begin{align}
    \p\event^c \le& \p\cbrace*{\abs*{\sum_i\frac{U_i}{k} - \beta} > \epsilon} + \p\cbrace*{\abs*{\theta(P) - \Mean[\theta(P)]} > d/4} \\
    \le& \frac{\Var[U]}{k\epsilon^2} + \frac{16\sum_{i}\Var[\phi(U_i/k)]}{d^2} \\
    \le& \frac{1}{16} + \frac{16\Mean\bracket*{\paren*{\phi(U_i/k) - \phi(0)}^2}}{d^2}. \label{eq:approx-error-prop}
  \end{align}
  We derive upper bounds on \cref{eq:approx-error-prop} by dividing into three cases; $\alpha \in (0,1)$, $\alpha = 1$, and $\alpha \in (1,2)$.

  {\bfseries Case $\alpha \in (0,1)$.}
  From \cref{lem:div-speed-holder1}, we have
  \begin{align}
    \frac{16\sum_i\Mean\bracket*{\paren*{\phi(U_i/k) - \phi(0)}^2}}{d^2} \le& \frac{16\norm*{\phi}_{C^{H,\alpha}}^2\sum_i\Mean\bracket*{U_i^{2\alpha}}}{k^{2\alpha}d^2} \\\le& \frac{16\norm*{\phi}_{C^{H,\alpha}}^2\lambda^{2\alpha}}{k^{2\alpha-1}d^2}. \label{eq:approx-tv-expect1}
  \end{align}

  {\bfseries Case $\alpha = 1$.}
  From the absolutely continuity, we have
  \begin{align}
    &\frac{16\sum_i\Mean\bracket*{\paren*{\phi(U_i/k) - \phi(0)}^2}}{d^2} \\ =& \frac{16\sum_i\Mean\bracket*{\paren*{\int_0^1U_i\phi^{(1)}(xU_i/k)dx}^2}}{k^2d^2} \\
    \le& \frac{16\sum_i\Mean\bracket*{\paren*{\int_0^1U_i\abs*{\phi^{(1)}(xU_i/k)}dx}^2}}{k^2d^2}.
  \end{align}
  From \cref{lem:lowest-int-div} and the fact that $\ln(1/p) > 0$ for $p \in (0,1)$,  there exists $W > 0$ such that $\abs*{\phi^{(1)}(p)} \le W\ln(1/p)$ for $p \in (0,1)$. Hence,
  \begin{align}
    &\frac{16\sum_i\Mean\bracket*{\paren*{\phi(U_i/k) - \phi(0)}^2}}{d^2} \\ \le& \frac{16W^2\sum_i\Mean\bracket*{\paren*{\int_0^1U_i\ln(xU_i/k)dx}^2}}{k^2d^2} \\
    =& \frac{16W^2\sum_i\Mean\bracket*{\paren*{U_i\ln(U_i/ek)}^2}}{k^2d^2} \\
    \le& \frac{16W^2\lambda^2\ln^2(\lambda/ek)}{kd^2}, \label{eq:approx-tv-expect2}
  \end{align}
  where we use the fact that $\max_{x \in [0,\lambda]}x^2\ln^2(x/ek) = \lambda^2\ln^2(\lambda/ek)$ if $\lambda < ek$ to obtain the last line.

  {\bfseries Case $\alpha \in (1,2)$.}
  Without loss of generality, we can assume $\phi^{(1)}(0) = 0$ because $\theta(P;\phi) = \theta(P;\phi_c)$ for any $c \in \RealSet$ where $\phi_c(p) = \phi(p) + c(p-1/k)$ and $\phi$ is Lipschitz continuous as shown in \cref{lem:div-speed-holder1}. Hence, the Taylor theorem indicates that there exist $\xi_i$ between $0$ and $U_i/k$ such that
  \begin{align}
    \p\event^c \le& \frac{1}{16} + \frac{16\sum_i\Mean\bracket*{\paren*{U_i\paren*{\phi^{(1)}(\xi_i) - \phi^{(1)}(0)}/k}^2}}{d^2}.
  \end{align}
  From \cref{lem:div-speed-holder3}, we obtain
  \begin{align}
    \p\event^c \le& \frac{1}{16} + \frac{16\sum_i\Mean\bracket*{\norm{\phi^{(1)}}_{C^{H,\alpha-1}}U_i\xi^{\alpha-1}/k}^2}{d^2} \\
    \le& \frac{1}{16} + \frac{16\norm{\phi^{(1)}}^2_{C^{H,\alpha-1}}\lambda^{2\alpha}}{k^{2\alpha-1}d^2}\label{eq:approx-tv-expect3}.
  \end{align}

  Note that in all the cases, $\p\event'^c$ has the same upper bound as in \cref{eq:approx-tv-expect1,eq:approx-tv-expect2,eq:approx-tv-expect3} by the same manner. Substituting \cref{eq:approx-tv-expect1,eq:approx-tv-expect2,eq:approx-tv-expect3} into \cref{eq:approx-lower-tv} yields the desired result.
\end{proof}

\section{Proofs for Best Polynomial Approximation Error Analysis}

\begin{proof}[Proof of \cref{thm:upper-poly-approx1}]
  First, we prove the claim for $\alpha \ne 1$. Let $\dom{P}_{\mathrm{even},L}$ be the set of all polynomials up to degree $L$ that only consist of even degree monomials. Then, we have
  \begin{align}
    E_L(\phi,[0,\lambda]) \le& \inf_{g \in \dom{P}_{\mathrm{even},L}}\sup_{x \in [0,\lambda]}\abs{\phi(x) - g(x)} \\
    =& E_{\floor{L/2}}(\phi_\lambda,[0,1]) \le E_{\floor{L/2}}(\phi_\lambda,[-1,1]),
  \end{align}
  where $\phi_\lambda(x) = \phi(\lambda x^2)$. Since $\varphi(x) \le 1$, we have $\omega^r_\varphi(f,t;[-1,1]) \le \omega^r(f,t;[-1,1])$. Thus, from the direct result in \cref{lem:best-modulus-direct}, we can obtain the upper bound on $E_L(\phi,[0,\lambda])$ by analyzing the modulus of smoothness $\omega^r(\phi_\lambda,t;[-1,1])$ for some positive integer $r$. We divide into three cases; $\alpha \in (0,1/2]$, $\alpha \in (1/2,1)$, and $\alpha \in (1,3/2)$.

  {\bfseries Case $\alpha \in (0,1/2]$.}
  In this case, we analyze the first order modulus of smoothness, which is rewritten as
  \begin{align}
      \omega^1(f,t;[-1,1]) = \sup_{x,y \in [-1,1]}\cbrace*{\abs*{f(x)-f(y)} : \abs*{x-y} \le t}.
  \end{align}
  Note that $\phi_\lambda$ is continuously differentiable, and its derivative is obtained as $\phi^{(1)}_\lambda(x) = 2\lambda x\phi^{(1)}(\lambda x^2)$. From the divergence speed assumption and the fact that $p^{\alpha-1} \ge 1$ for any $p \in (0,1)$, we have $\abs*{\phi^{(1)}(p)} \le (W_1+c_1)p^{\alpha-1}$. Thus, we have for any $x,y \in [-1,1]$,
  \begin{align}
     \abs*{\phi_\lambda(x) - \phi_\lambda(y)} =& \abs*{\int_x^y \phi^{(1)}_\lambda(s)ds} \\
     \le& 2(W_1+c_1)\lambda^\alpha\abs*{\int_x^y \abs*{s}^{2\alpha-1}ds}.
  \end{align}

  Assume $xy \ge 0$. Then, we have
  \begin{align}
    \abs*{\phi_\lambda(x) - \phi_\lambda(y)} \le& \frac{2(W_1+c_1)}{2\alpha}\lambda^\alpha\abs*{x^{2\alpha} - y^{2\alpha}} \\
    \le& \frac{2(W_1+c_1)}{2\alpha}\lambda^\alpha\abs*{x - y}^{2\alpha},
  \end{align}
  where we use $x \to x^\beta$ is $\beta$-H\"older continuous for $\beta \in (0,1]$ to obtain the last inequality. Under the condition $\abs*{x-y} \le t$, we have
  \begin{align}
    \abs*{\phi_\lambda(x) - \phi_\lambda(y)} \le& \frac{2(W_1+c_1)}{2\alpha}\lambda^\alpha t^{2\alpha}.
  \end{align}

  Next, assume $xy < 0$. Then, we have
  \begin{align}
    & \abs*{\phi_\lambda(x) - \phi_\lambda(y)} \\
    \le& 2(W_1+c_1)\lambda^\alpha\abs*{\int_0^{\abs*{x}} s^{2\alpha-1}ds + \int_0^{\abs*{y}} s^{2\alpha-1}ds} \\
    =& \frac{2(W_1+c_1)}{2\alpha}\lambda^\alpha\paren*{\abs*{x}^{2\alpha} + \abs*{y}^{2\alpha}}.
  \end{align}
  Since $\abs*{x-y} = \abs*{x} + \abs*{y}$, under the condition $\abs{x-y} \le t$, we have
  \begin{align}
     \abs*{\phi_\lambda(x) - \phi_\lambda(y)} \le \frac{4(W_1+c_1)}{2\alpha}\lambda^\alpha t^{2\alpha}.
  \end{align}

  From the direct result in \cref{lem:best-modulus-direct}, we have
  \begin{align}
      E_L\paren*{\phi,[0,\lambda]} \lesssim& \omega^1_\varphi(\phi_\lambda,L^{-1};[-1,1]) \\
      \le& \omega^1(\phi_\lambda,L^{-1};[-1,1]) \lesssim \paren*{\frac{\lambda}{L^2}}^\alpha.
  \end{align}

  {\bfseries Case $\alpha \in (1/2,1]$.}
  From \cref{lem:mod-derivative}, it is sufficient to derive an upper bound on the first order modulus of smoothness of the first derivative $\omega^1(\phi^{(1)}_\lambda,t;(-1,1))$. Noting that $\phi_\lambda^{(2)}(x) = 2\lambda\phi^{(1)}(\lambda x^2) + 4\lambda^2x^2\phi^{(2)}(\lambda x^2)$, we have for any $x,y \in (-1,1)$,
  \begin{multline}
    \abs*{\phi^{(1)}_\lambda(x) - \phi^{(1)}_\lambda(y)} \le \\ 2\lambda\abs*{\int_x^y\phi^{(1)}(\lambda s^2)ds}  + 4\lambda^2\abs*{\int_x^ys^2\phi^{(2)}(\lambda s^2) ds}.
  \end{multline}
  From the divergence speed assumption and \cref{lem:lower-div}, we have $\abs{\phi^{(1)}(p)} \le (W_1+c_1)p^{\alpha-1}$ and $\abs{\phi^{(2)}(p)} \le (W_2+c_2)p^{\alpha-2}$. Hence,
  \begin{multline}
    \abs*{\phi^{(1)}_\lambda(x) - \phi^{(1)}_\lambda(y)} \le 2(W_1+c_1)\lambda^\alpha\abs*{\int_x^y \abs*{s}^{2\alpha-2}ds} \\ + 4(W_2+c_2)\lambda^\alpha\abs*{\int_x^y\abs*{s}^{2\alpha-2} ds}.
  \end{multline}
  By following the same manner in the case $\alpha \in (0,1/2]$, under the condition $\abs*{x - y} \le t$, we have
  \begin{align}
    \abs*{\phi^{(1)}_\lambda(x) - \phi^{(1)}_\lambda(y)} \le& \frac{4(W_1+c_1)+8(W_2+c_2)}{2\alpha-1}\lambda^\alpha t^{2\alpha-1}.
  \end{align}
  From the direct result in \cref{lem:best-modulus-direct} and \cref{lem:mod-derivative}, we have
  \begin{align}
    E_L\paren*{\phi,[0,\lambda]} \lesssim& \omega^2_\varphi\paren*{\phi_\lambda,L^{-1};(-1,1)} \\
    \lesssim& L^{-1}\omega^1\paren*{\phi^{(1)}_\lambda,L^{-1};[-1,1]} \\
    \lesssim& \paren*{\frac{\lambda}{L^2}}^\alpha.
  \end{align}

  {\bfseries Case $\alpha \in (1,3/2)$.}
  From \cref{lem:mod-derivative}, it is sufficient to derive an upper bound on the first order modulus of smoothness of the second derivative $\omega^1(\phi^{(2)}_\lambda,t;(-1,1))$. Noting that $\phi_\lambda^{(3)}(x) = 12\lambda^2x\phi^{(2)}(\lambda x^2) + 8\lambda^3x^3\phi^{(3)}(\lambda x^2)$, we have for any $x,y \in (-1,1)$,
  \begin{multline}
    \abs*{\phi^{(2)}_\lambda(x) - \phi^{(2)}_\lambda(y)} \le \\ 12\lambda^2\abs*{\int_x^ys\phi^{(2)}(\lambda s^2)ds}  + 8\lambda^3\abs*{\int_x^ys^3\phi^{(3)}(\lambda s^2) ds}.
  \end{multline}
  From the divergence speed assumption and \cref{lem:lower-div}, we have $\abs{\phi^{(2)}(p)} \le (W_2+c_2)p^{\alpha-2}$ and $\abs{\phi^{(3)}(p)} \le (W_3+c_3)p^{\alpha-3}$. Hence,
  \begin{multline}
    \abs*{\phi^{(2)}_\lambda(x) - \phi^{(2)}_\lambda(y)} \le 12(W_2+c_2)\lambda^\alpha\abs*{\int_x^y \abs*{s}^{2\alpha-3}ds} \\ + 8(W_3+c_3)\lambda^\alpha\abs*{\int_x^y\abs*{s}^{2\alpha-3} ds}.
  \end{multline}
  By following the same manner in the case $\alpha \in (0,1/2]$, under the condition $\abs*{x - y} \le t$, we have
  \begin{align}
    \abs*{\phi^{(2)}_\lambda(x) - \phi^{(2)}_\lambda(y)} \le& \frac{12(W_2+c_2)+8(W_3+c_3)}{2\alpha-2}\lambda^\alpha t^{2\alpha-2}.
  \end{align}
  From the direct result in \cref{lem:best-modulus-direct} and \cref{lem:mod-derivative}, we have
  \begin{align}
    E_L\paren*{\phi,[0,\lambda]} \lesssim& \omega^3_\varphi\paren*{\phi_\lambda,L^{-1};(-1,1)} \\
    \lesssim& L^{-2}\omega^1\paren*{\phi^{(2)}_\lambda,L^{-1};[-1,1]} \\
    \lesssim& \paren*{\frac{\lambda}{L^2}}^\alpha.
  \end{align}

  Next, we prove the claim for $\alpha = 1$. Letting $\phi'_\lambda(x)=\phi(\Delta(1+x)/2)$, we have $E_L(\phi,[0,\Delta]) = E_L(\phi'_\lambda,[-1,1])$. From the direct result in \cref{lem:best-modulus-direct}, we have $E_L(\phi,[0,\Delta]) \lesssim \omega^2_\varphi(\phi'_\lambda,L^{-1};[-1,1])$ where the second-order weighted modulus of smoothness can be rewritten as
  \begin{multline}
      \omega^2_\varphi(f,t;[-1,1]) = \\ \sup_{x,y \in [-1,1]}\cbrace[\Bigg]{\abs*{f(x)+f(y)-2f\paren*{\frac{x+y}{2}}} : \\ \abs*{x-y} \le t\varphi\paren*{\frac{x+y}{2}}}.
  \end{multline}
  Letting $z = (x+y)/2$ and $h = (x-y)/2$, we have $x = z + h$ and $y = z - h$. Hence,
  \begin{multline}
      \omega^2_\varphi(f,t;[-1,1]) = \\ \sup_{z+h,z-h \in (-1,1)}\cbrace[\Bigg]{\abs*{f(z+h)+f(z-h)-2f\paren*{z}} : \\ \frac{h^2}{1-z^2} \le 4t^2}.
  \end{multline}
  By the Taylor theorem, we have
  \begin{align}
      & \abs*{\phi'_\lambda(z+h)+\phi'_\lambda(z-h) - 2\phi'_\lambda(z)} \\
      =& \begin{multlined}[t][.9\columnwidth]
       \abs[\Bigg]{\int_0^h\phi^{'(2)}_\lambda(z+s)(h-s)ds \\ + \int_0^{-h}\phi^{'(2)}_\lambda(z+s)(-h-s)ds}
      \end{multlined}\\
      =& \abs*{\int_{-h}^h\phi^{'(2)}_\lambda(z+s)\abs*{h-s}\land\abs*{h+s}ds} \\
      \le& \abs*{\int_{-h}^h\abs*{\phi^{'(2)}_\lambda(z+s)}\abs*{h-s}\land\abs*{h+s}ds}.
  \end{align}
  From the divergence speed assumption and the fact $p^{-1} \ge 1$ for $p \in (0,1)$, we have $\abs*{\phi^{(2)}(p)} \le (W_2+c_2)p^{-1}$ for any $p \in (0,1)$. Noting that $\phi^{'(2)}_\lambda(x) = \frac{\lambda^2}{4}\phi^{(2)}(\lambda(1+x)/2)$, we have
  \begin{align}
      & \abs*{\phi'_\lambda(z+h)+\phi'_\lambda(z-h) - 2\phi'_\lambda(z)} \\
      \le& \frac{(W_2+c_2)\lambda}{2}\abs*{\int_{-h}^h\frac{1}{1+z+s}\abs*{h-s}\land\abs*{h+s}ds} \\
      =& \frac{(W_2+c_2)\lambda}{2}\abs*{\int_0^h\frac{h-s}{1+z+s}ds + \int_0^{-h}\frac{-h-s}{1+x+s}ds}.
  \end{align}
  For the integral term, we have
  \begin{align}
    & \int_0^h\frac{h-s}{1+z+s}ds \\
    =& \begin{multlined}[t][.9\columnwidth]
      -h\ln(1+z) + (1+z+h)\ln(1+z+h) \\ - (1+z+h) - (1+z)\ln(1+z) + (1+z)
    \end{multlined}\\
    =& (1+z)\paren*{\paren*{1+\frac{h}{1+z}}\ln\paren*{1+\frac{h}{1+z}} - \frac{h}{1+z}}.
  \end{align}
  Similarly, we have
  \begin{multline}
    \int_0^{-h}\frac{-h-s}{1+z+s}ds = \\ (1+z)\paren*{\paren*{1+\frac{-h}{1+z}}\ln\paren*{1+\frac{-h}{1+z}} - \frac{-h}{1+z}}.
  \end{multline}
  Hence,
  \begin{align}
      & \abs*{\phi'_\lambda(z+h)+\phi'_\lambda(z-h) - 2\phi'_\lambda(z)} \\
      =& \begin{multlined}[t][.9\columnwidth]
        \frac{(W_2+c_2)\lambda}{2}(1+z)\abs[\Bigg]{\paren*{1+\frac{h}{1+z}}\ln\paren*{1+\frac{h}{1+z}}\\+\paren*{1-\frac{h}{1+z}}\ln\paren*{1-\frac{h}{1+z}}}.
      \end{multlined}
  \end{align}
  For $x \in (-1,1)$, since the Taylor expansion of $f(z)=\ln(1-z)$ at $z=0$ yields $\ln(1-x)=-\sum_{m=1}^\infty x^m/m$, we have
  \begin{align}
    (1-x)\ln(1-x) =& \sum_{m=2}^\infty\frac{x^{m}}{m-1} - \sum_{m=1}^\infty \frac{x^m}{m}\\
    =& -x + \sum_{m=2}^\infty\paren*{\frac{1}{m-1}-\frac{1}{m}}x^m.
  \end{align}
  Hence,
  \begin{align}
      & (1+x)\ln(1+x)+(1-x)\ln(1-x) \\
      =& 2\sum_{m=1}^\infty\paren*{\frac{1}{2m-1}-\frac{1}{2m}}x^{2m} \\
      =& \sum_{m=1}^\infty\frac{2}{2m(2m-1)}x^{2m}.
  \end{align}
  For $x \in (-1,1)$, we have
  \begin{multline}
      (1+x)\ln(1+x)+(1-x)\ln(1-x) \le \\ x^2\sum_{m=1}^\infty\frac{2}{2m(2m-1)} = (2\ln 2) x^2.
  \end{multline}
  Since $z + h, z - h \in (-1,1)$, we have $z \in (-1,1)$ and
  \begin{align}
      1 + z + h > 0 \iff& \frac{h}{1+z} > -1, \\
      1 + z + h < 2 \iff& \frac{h}{1+z} < 1.
  \end{align}
  Hence, under $h^2/(1-z^2)\le 4t^2$,
  \begin{align}
      & \abs*{\phi'_\lambda(z+h)+\phi'_\lambda(z-h) - 2\phi'_\lambda(z)} \\
      =& \ln2(W_2+c_2)\lambda\frac{h^2}{1+z} \\
      =& \ln2(W_2+c_2)\lambda(1-z)\frac{h^2}{1-z^2} \\
      \le& 8\ln2(W_2+c_2)\lambda t^2.
  \end{align}
  Application of the direct result in \cref{lem:best-modulus-direct} yields the desired claim.
\end{proof}

\begin{proof}[Proof of \cref{thm:lower-poly-approx1}]
  Letting $\phi_\lambda(x)=\phi(\lambda(x+1)/2)$, we have $E_L(\phi,[0,\lambda])=E_L(\phi_\lambda,[-1,1])$. We use the second order weighted modulus of smoothness, which can be rewritten as
  \begin{multline}
      \omega^2_\varphi(f,t;[-1,1]) = \sup_{x,y \in [-1,1]}\cbrace[\Bigg]{\\\abs*{f(x) + f(y) - 2f\paren*{\frac{x+y}{2}}} : \abs*{x-y} \le 2t\varphi\paren*{\frac{x+y}{2}}}.
  \end{multline}
  Fix $y = -1$. For $t > 0$, we have
  \begin{align}
      \abs*{x - y} \le 2t\varphi\paren*{\frac{x+y}{2}} \iff -1 \le x \le -1 + \frac{4}{t^{-2}+1}.
  \end{align}
  Since $2/(t^{-2}+1) \ge t^2$ for $t \in (0,1)$, we have
  \begin{align}
      & \omega^12_\varphi(\phi_\lambda,t;[-1,1]) \\
      \ge& \begin{multlined}[t][.9\columnwidth]
       \sup_{x \in [-1,1]}\cbrace[\bigg]{\abs*{\phi_\lambda(x) + \phi_\lambda(-1) - 2\phi_\lambda\paren*{\frac{x-1}{2}}} \\ : 0 \le 1+x \le 4/(t^{-2}+1)}
      \end{multlined}\\
      \ge& \sup_{x \in [0,1]}\cbrace*{\abs*{\phi(\lambda x) + \phi(0) - 2\phi(\lambda x/2)} : 0 \le x \le t^2}.
  \end{align}
  Letting $p_0 = 1 \land (c_2'/W_2)^{1/(\alpha-2)}$, we have $\abs*{\phi^{(2)}(p)} \ge 0$ for $p \in (0,p_0]$ because of the divergence speed assumption. Since $\phi^{(2)}$ is continuous, $\phi^{(2)}(p)$ has the same sign in $p \in (0,p_0]$. By the Taylor theorem, for any $x \in (0,p_0]$,
  \begin{align}
      & \abs*{\phi(x) + \phi(0) - 2\phi(x/2)} \\
      =& \frac{1}{2}\abs*{\int_{x/2}^x(x-s)\phi^{(2)}(s)ds + \int_0^{x/2}s\phi^{(2)}(s)ds} \\
      \ge& \begin{multlined}[t][.9\columnwidth]
       \frac{1}{2}\abs[\Bigg]{\int_{x/2}^x(x-s)\paren*{W_2 s^{\alpha-2} - c'_2}ds \\ + \int_0^{x/2}s\paren*{W_2 s^{\alpha-2} - c'_2}ds}
      \end{multlined}\\
      =& \frac{1}{2}\paren*{\frac{W_2}{2^\alpha\alpha(1-\alpha)}\paren*{2-2^\alpha}x^\alpha - \frac{c'_2}{4}x^2}
  \end{align}
  Thus, for sufficiently small $\lambda$, we have
  \begin{multline}
      \omega^2_\varphi\paren*{\phi_\lambda,t;[-1,1]} \ge \\ \paren*{\lambda t^2}^\alpha\paren*{\frac{W_2}{\alpha(1-\alpha)}\paren*{2^{1-\alpha}-1} - \frac{c'_2}{4}\paren*{\lambda t^2}^{2-\alpha}} > 0. \label{eq:lower-poly-approx1-mod2}
  \end{multline}

  Let $L'$ be an integer such that $L' = c'L$ where $c' > 1$. Then, we have
  \begin{align}
    & E_L\paren*{\phi,[0,\lambda]} \\
    \ge& \frac{1}{L'-L}\sum_{m=L+1}^{L'}E_m\paren*{\phi_\lambda,[-1,1]} \\
    \ge& \frac{1}{2L^{'2}}\sum_{m=L+1}^{L'}(m+1)E_m\paren*{\phi_\lambda,[-1,1]} \\
    \ge& \begin{multlined}[t][.8\columnwidth]
      \frac{1}{2L^{'2}}\sum_{m=0}^{L'}(m+1)E_m\paren*{\phi_\lambda,[-1,1]} \\ - \frac{1}{2L^{'2}}\sum_{m=0}^2E_m\paren*{\phi_\lambda,[-1,1]} \\ - \frac{1}{2L^{'2}}\sum_{m=2}^{L}(m+1)E_m\paren*{\phi_\lambda,[-1,1]}.
    \end{multlined} \label{eq:lower-poly-approx1-eq1}
  \end{align}
  From H\"older continuity shown in \cref{lem:div-speed-holder1}, we have
  \begin{align}
      & E_1\paren*{\phi_\lambda,[-1,1]} \le E_0\paren*{\phi_\lambda,[-1,1]} \\
      =& \inf_{g \in \RealSet}\sup_{x \in [-1,1]}\abs*{\phi_\lambda(x) - g} \\
      \le& \sup_{x \in [-1,1]}\abs*{\phi\paren*{\frac{\lambda(x+1)}{2}} - \phi(0)} \\
      \le& \norm*{\phi}_{C^{H,\alpha}}\lambda^\alpha\sup_{x \in [-1,1]}\paren*{\frac{x+1}{2}}^\alpha = \norm*{\phi}_{C^{H,\alpha}}\lambda^\alpha. \label{eq:upper-poly-error01}
  \end{align}
  Combining \cref{eq:lower-poly-approx1-eq1,lem:best-modulus-converse,eq:lower-poly-approx1-mod2,eq:upper-poly-error01,thm:upper-poly-approx1} yields that there exist universal constants $C_1,C_2,C_3 > 0$ only depending on $\phi$ such that
  \begin{multline}
    E_L\paren*{\phi,[0,\lambda]} \ge \paren*{\frac{\lambda}{(c'L)^2}}^\alpha\paren[\Bigg]{C_1 - \\ \frac{C_2}{(c'L)^{2-2\alpha}} - \frac{C_3}{(c'L)^{2-2\alpha}}\sum_{m=2}^Lm^{1-2\alpha}}.
  \end{multline}
  For the third term, we have
  \begin{align}
    \sum_{m=2}^Lm^{1-2\alpha} \le& \int_0^L s^{1-2\alpha}ds \\
    =& \frac{1}{2-2\alpha}L^{2-2\alpha}.
  \end{align}
  Hence, there exists an universal constant $C'_3 > 0$ only depending on $\phi$ such that
  \begin{align}
    E_L\paren*{\phi,[0,\lambda]} \ge& \paren*{\frac{\lambda}{(c'L)^2}}^\alpha\paren*{C_1 - \frac{C_2}{(c'L)^{2-2\alpha}} - \frac{C'_3}{(c')^{2-2\alpha}}}.
  \end{align}
  For sufficiently large $c'$, the right-hand side is larger than zero. Thus, we get the claim.
\end{proof}

\begin{proof}[Proof of \cref{thm:lower-poly-approx2}]
  Let $L'$ be an integer such that $L' = c_\ell L$ where $c_\ell  > 1$. Then, we have
  \begin{align}
    & E_L(\phi^\star,[\gamma,2L^2\gamma])\\
     \ge& \frac{1}{L' - L}\sum_{m = L+1}^{L'} E_m(\phi^\star,[\gamma,2m^2\gamma]) \\
     \ge& \frac{1}{L'}\sum_{m = L+1}^{L'} E_m(\phi^\star,[\gamma,2m^2\gamma]) \\
     \ge& \frac{1}{L'}\sum_{m = 1}^{L'}E_m(\phi^\star,[\gamma,2m^2\gamma]) - \frac{1}{L'}\sum_{m=1}^LE_m(\phi^\star,[\gamma,2m^2\gamma]). \label{eq:lower-best-approx}
  \end{align}
  From \cref{eq:lower-best-approx}, we can obtain the desired lower bound by deriving a lower bound on the first term and an upper bound on the second term.

  Letting $\phi^\star_{\eta,\gamma}(x) = \phi^\star_\gamma(\gamma\frac{1+\eta+(1-\eta)x}{2\eta})$, we have $E_L(\phi^\star_\gamma,[\gamma,\gamma/\eta]) = E_L(\phi^\star_{\eta,\gamma},[-1,1])$. From the converse result in \cref{lem:best-modulus-converse} and the fact $E_m(\phi^\star,[\gamma,2m^2\gamma]) = 0$ for $m = 0$, the first term in \cref{eq:lower-best-approx} is bounded below by the first order weighted modulus of smoothness $\omega^1_\varphi(\phi^\star_{1/2L^2,\gamma},L^{-1})$. The first order weighted modulus of smoothness can be rewritten as
  \begin{multline}
    \omega^1_\varphi(f,t) = \\ \sup_{x,y \in [-1,1]}\cbrace*{\abs*{f(x)-f(y)} : \abs*{x-y} \le 2t\varphi\paren*{\frac{x+y}{2}}}.
  \end{multline}
  In the same manner of the proof of \cref{thm:lower-poly-approx1}, we have
  \begin{align}
    & \omega^1_\varphi(\phi^\star_{\eta,\gamma},t) \\
    \ge& \sup_{x \in [-1,1]}\cbrace*{\abs*{\phi^\star_{\eta,\gamma}(x)-\phi^\star_{\eta,\gamma}(-1)} : 0 \le 1 + x \le \frac{4}{t^{-2}+1}} \\
    \ge& \sup_{x \in [0,1]}\cbrace[\Bigg]{\abs*{\phi^\star_\gamma\paren*{\gamma\paren*{1+\frac{(1-\eta)x}{\eta}}}-\phi^\star_\gamma(\gamma)} : 0 \le x \le t^2} \\
    =& \begin{multlined}[t][.85\columnwidth]
      \frac{1}{\gamma}\sup_{x \in [0,1]}\cbrace[\Bigg]{\frac{1}{1+\frac{(1-\eta)x}{\eta}}\abs[\bigg]{\phi\paren*{\gamma\paren*{1+\frac{(1-\eta)x}{\eta}}} \\ -\phi(\gamma)\paren*{1+\frac{(1-\eta)x}{\eta}}} : 0 \le x \le t^2}
    \end{multlined}\\
    =& \begin{multlined}[t][.85\columnwidth]
    \frac{1}{\gamma}\sup_{x \in [0,(1-\eta)/\eta]}\cbrace[\Bigg]{\frac{1}{1+x}\abs*{\phi\paren*{\gamma\paren*{1+x}}-\phi(\gamma)\paren*{1+x}} \\ : 0 \le x \le \frac{1-\eta}{\eta}t^2}
    \end{multlined}\\
    =& \begin{multlined}[t][.85\columnwidth]
    \frac{1}{\gamma}\sup_{x \in [0,(1-\eta)/\eta]}\cbrace[\Bigg]{\frac{1}{1+x}\abs*{\int_\gamma^{\gamma\paren*{1+x}}\phi^{(1)}(s)ds} \\ : 0 \le x \le \frac{1-\eta}{\eta}t^2}. \label{eq:approx2-first-modulus}
    \end{multlined}
  \end{align}

  As well as the proof of \cref{thm:upper-poly-approx1}, we have $E_L(\phi^\star,[\gamma,2L^2\gamma]) \le E_L(\phi^{\star}_\gamma,[-1,1])$ for $\phi^{\star}_\gamma(x) = \phi^\star(\gamma(1+2L^2x^2))$. From the direct result in \cref{lem:best-modulus-direct} and the fact $\varphi \le 1$, we have
  \begin{align}
    E_L(\phi^{\star}_\gamma,[-1,1]) \lesssim \omega_1(\phi^{\star}_\gamma,L^{-1})
  \end{align}
  For any $x,y \in (-1,1)$, we have
  \begin{align}
    &\abs*{\phi^{\star}_\gamma(x)-\phi^{\star}_\gamma(y)} \\
    =& \abs*{\int_y^x \phi^{\star(1)}_\gamma(x)ds} \\
    \le& \begin{multlined}[t][.85\columnwidth]
      \abs*{\int_y^x\frac{4L^2s}{1+2L^2s^2}\abs*{\phi^{(1)}(\gamma(1+2L^2s^2))}ds} \\ + \abs*{\int_y^x\frac{4 L^2s}{\gamma(1+2L^2s^2)^2}\abs*{\phi(\gamma(1+2L^2s^2))}ds}. \label{eq:approx2-modulus-upper}
    \end{multlined}
  \end{align}

  To derive bounds on \cref{eq:approx2-first-modulus,eq:approx2-modulus-upper}, we divide into two cases; $\alpha = 1$ and $\alpha \in (1,3/2)$.

  {\bfseries Case $\alpha = 1$.}
  From absolutely continuity of $\phi^{(1)}$ and the assumption that $\phi^{(1)}(\gamma) = 0$, we have
  \begin{align}
    \abs*{\int_\gamma^{\gamma\paren*{1+x}}\phi^{(1)}(s)ds} =  \abs*{\int_\gamma^{\gamma\paren*{1+x}}\int_\gamma^s\phi^{(2)}(s')ds'ds}.
  \end{align}
  Letting $p_0 = 1\land(c'_2/W_2)^{1/(\alpha-2)}$, $\abs*{\phi^{(2)}(p)} \ge W_2p^{\alpha-2} - c'_2 \ge 0$ for $(0,p_0]$. From continuity of $\phi^{(2)}$, $\phi^{(2)}(x)$ has the same sign for $x \in (0,p_0]$. For sufficiently small $\gamma$ such that $\gamma(1+x) \le p_0$, we have
  \begin{align}
    & \abs*{\int_\gamma^{\gamma\paren*{1+x}}\int_\gamma^sW_2(s')^{-1} - c'_2 ds'ds} \\
    =& \abs*{\int_\gamma^{\gamma\paren*{1+x}}W_2\ln\paren*{\frac{s}{\gamma}} - c'_2s ds} \\
    =& \abs*{W_2\gamma\paren*{(1+x)\ln(1+x) - x} - \frac{c'_2\gamma^2}{2}\paren*{\paren*{1+x}^2-1}}.
  \end{align}
  Set $\eta = t^2/2$. Then, we have $x \le 1 - t^2/2 \le 1$ for $t \le 1$. By the inverse triangle inequality, we have
  \begin{align}
    & \abs*{W_2\gamma\paren*{(1+x)\ln(1+x) - x} - \frac{c'_2\gamma^2}{2}\paren*{\paren*{1+x}^2-1}} \\
    \ge& \gamma\paren*{W_2\paren*{(1+x)\ln(1+x) - x} - \frac{c'_2\gamma}{2}\paren*{\paren*{1+x}^2-1}} \\
    \ge& \gamma\paren*{W_2\paren*{(1+x)\ln(1+x) - x} - \frac{3c'_2\gamma}{2}}.
  \end{align}
  For $t \le 1$, we have $x \in [0,1/2]$. Thus, we obtain a lower bound on \cref{eq:approx2-first-modulus} as
  \begin{align}
    & \omega^1_\varphi(\phi^\star_{\eta,\gamma},t) \\
    \ge& \sup_{x \in [0,1/2]}\cbrace*{\frac{1}{1+x}\paren*{W_2\paren*{(1+x)\ln(1+x) - x} - \frac{3c'_2\gamma}{2}}} \\
    \ge& \sup_{x \in [0,1/2]}\cbrace*{\frac{2}{3}\paren*{W_2\paren*{(1+x)\ln(1+x) - x} - \frac{3c'_2\gamma}{2}}} \\
    =& \Omega(1) \as \gamma \to 0. \label{eq:approx2-lower-modulus1}
  \end{align}

  For any $x,y \in (-1,1)$, \cref{eq:approx2-modulus-upper} is bounded above by
  \begin{multline}
    \abs*{\int_y^x\frac{4L^2s}{1+2L^2s^2}\int_\gamma^{\gamma(1+2L^2s^2)}\abs*{\phi^{(2)}(s')}ds'ds} \\ + \abs[\Bigg]{\int_y^x\frac{4\gamma L^2s}{\gamma^2(1+2L^2s^2)^2}\\\times\int_\gamma^{\gamma(1+2L^2s^2)}\int_\gamma^{s'}\abs*{\phi^{(2)}(s'')}ds''ds'ds}.
  \end{multline}
  For any $p \in (0,1)$, $\abs*{\phi^{(2)}(p)} \le (W_2+c_2)p^{\alpha-2}$ because of the divergence speed assumption and the fact $p^{\alpha-2} \ge 1$. For the first term, we have
  \begin{align}
    &\abs*{\int_y^x\frac{4L^2s}{1+2L^2s^2}\int_\gamma^{\gamma(1+2L^2s^2)}\abs*{\phi^{(2)}(s')}ds'ds} \\
    \le& \abs*{\int_y^x\frac{4L^2s}{1+2L^2s^2}\int_\gamma^{\gamma(1+2L^2s^2)}(W_2+c_2)(s')^{-1}ds'ds}\\
    =& (W_2+c_2)\abs*{\int_y^x\frac{4L^2s}{1+2L^2s^2}\ln\paren*{1+2L^2s^2}ds} \\
    \le& (W_2+c_2)\abs*{\int_y^x2s^{-1}\ln\paren*{1+2L^2s^2}ds} \\
    \le& (W_2+c_2)\abs*{\int_y^x2^{\frac{3}{2}}Lds} \\
    =& 2^{\frac{3}{2}}(W_2+c_2)L\abs*{x - y},
  \end{align}
  where we use the fact $\ln(1+x) \le \sqrt{x}$ for any $x \ge 0$ to obtain the fifth line. For the second term, we have
  \begin{align}
    &\abs*{\int_y^x\frac{4 L^2s}{\gamma(1+2L^2s^2)^2}\int_\gamma^{\gamma(1+2L^2s^2)}\int_\gamma^{s'}\abs*{\phi^{(2)}(s'')}ds''ds'ds} \\
    \le& \begin{multlined}[t][.9\columnwidth]
     \abs[\Bigg]{\int_y^x\frac{4 L^2s}{\gamma(1+2L^2s^2)^2}\int_\gamma^{\gamma(1+2L^2s^2)}\\\times\int_\gamma^{s'}(W_2+c_2)(s'')^{-1}ds''ds'ds}
    \end{multlined}\\
    =& (W_2+c_2)\abs*{\int_y^x\frac{4L^2s}{(1+2L^2s^2)^2}\int_1^{1+2L^2s^2}\ln\paren*{s'}ds'ds} \\
    \le& (W_2+c_2)\abs*{\int_y^x\frac{4L^2s}{(1+2L^2s^2)^2}\int_1^{1+2L^2s^2}\paren*{1+\ln\paren*{s'}}ds'ds} \\
    =& (W_2+c_2)\abs*{\int_y^x\frac{4L^2s}{(1+2L^2s^2)}\ln\paren*{1+2L^2s^2} ds} \\
    \le& 2^{\frac{3}{2}}(W_2+c_2)L\abs*{x - y}.
  \end{align}
  Consequently, we have
  \begin{align}
    E_L(\phi^{\star}_\gamma,[-1,1]) \lesssim \omega_1(\phi^{\star}_\gamma,L^{-1}) \lesssim 1. \label{eq:approx2-upper-modulus1}
  \end{align}

  {\bfseries Case $\alpha \in (1,3/2)$.}
  The proof in this case follows the same manner in the case $\alpha = 1$. From absolutely continuity of $\phi^{(1)}$ and the assumption that $\phi^{(1)}(0) = 0$, we have
  \begin{align}
    \abs*{\int_\gamma^{\gamma\paren*{1+x}}\phi^{(1)}(s)ds} =  \abs*{\int_\gamma^{\gamma\paren*{1+x}}\int_0^s\phi^{(2)}(s')ds'ds}.
  \end{align}
  Letting $p_0 = 1\land(c'_2/W_2)^{1/(\alpha-2)}$, $\abs*{\phi^{(2)}(p)} \ge W_2p^{\alpha-2} - c'_2 \ge 0$ for $(0,p_0]$. From continuity of $\phi^{(2)}$, $\phi^{(2)}(x)$ has the same sign for $x \in (0,p_0]$. For sufficiently small $\gamma$ such that $\gamma(1+x) \le p_0$, we have
  \begin{align}
    &  \abs*{\int_\gamma^{\gamma\paren*{1+x}}\int_0^s\phi^{(2)}(s')ds'ds} \\
    \ge& \abs*{\int_\gamma^{\gamma\paren*{1+x}}\int_0^sW_2(s')^{\alpha-2} - c'_2 ds'ds} \\
    =& \abs*{\int_\gamma^{\gamma\paren*{1+x}}\frac{W_2}{\alpha-1}s^{\alpha-1} - c'_2s ds} \\
    =& \abs*{\frac{W_2}{\alpha(\alpha-1)}\gamma^{\alpha}\paren*{\paren*{1+x}^{\alpha}-1} - \frac{c'_2\gamma^2}{2}\paren*{\paren*{1+x}^2-1}}.
  \end{align}

  Set $\eta = t^2/2$. Then, we have $x \le 1 - t^2/2 \le 1$ for $t \le 1$. By the inverse triangle inequality, we have
  \begin{align}
    & \abs*{\frac{W_2}{\alpha(\alpha-1)}\gamma^{\alpha}\paren*{\paren*{1+x}^{\alpha}-1} - \frac{c'_2\gamma^2}{2}\paren*{\paren*{1+x}^2-1}} \\
    \ge& \gamma^\alpha\paren*{\frac{W_2}{\alpha(\alpha-1)}\paren*{\paren*{1+x}^{\alpha}-1} - \frac{c'_2\gamma^{2-\alpha}}{2}\paren*{\paren*{1+x}^2-1}} \\
    \ge& \gamma^\alpha\paren*{\frac{W_2}{\alpha(\alpha-1)}\paren*{\paren*{1+x}^{\alpha}-1} - \frac{3c'_2\gamma^{2-\alpha}}{2}}.
  \end{align}
  For $t \le 1$, we have $x \in [0,1/2]$. Thus, we obtain a lower bound on \cref{eq:approx2-first-modulus} as
  \begin{align}
    & \omega^1_\varphi(\phi^\star_{\eta,\gamma},t) \\
    \ge& \begin{multlined}[t][.9\columnwidth]
      \gamma^{\alpha-1}\smashoperator[l]{\sup_{x \in [0,1/2]}}\cbrace[\Bigg]{\frac{1}{1+x}\\\times\paren*{\frac{W_2}{\alpha(\alpha-1)}\paren*{\paren*{1+x}^{\alpha}-1} - \frac{3c'_2\gamma^{2-\alpha}}{2}}}
    \end{multlined}\\
    \ge& \gamma^{\alpha-1}\smashoperator[l]{\sup_{x \in [0,1/2]}}\cbrace*{\frac{2}{3}\paren*{\frac{W_2}{\alpha(\alpha-1)}\paren*{\paren*{1+x}^{\alpha}-1} - \frac{3c'_2\gamma^{2-\alpha}}{2}}}\\
    =& \Omega(\gamma^{\alpha-1}) \as \gamma \to 0. \label{eq:approx2-lower-modulus2}
  \end{align}

  For any $x,y \in (-1,1)$, \cref{eq:approx2-modulus-upper} is bounded above by
  \begin{multline}
    \abs*{\int_y^x\frac{4L^2s}{1+2L^2s^2}\int_0^{\gamma(1+2L^2s^2)}\abs*{\phi^{(2)}(s')}ds'ds} \\ + \abs[\Bigg]{\int_y^x\frac{4\gamma L^2s}{\gamma^2(1+2L^2s^2)^2}\\\times\int_\gamma^{\gamma(1+2L^2s^2)}\int_0^{s'}\abs*{\phi^{(2)}(s'')}ds''ds'ds}.
  \end{multline}
  For any $p \in (0,1)$, $\abs*{\phi^{(2)}(p)} \le (W_2+c_2)p^{\alpha-2}$ because of the divergence speed assumption and the fact $p^{\alpha-2} \ge 1$. For the first term, we have
  \begin{align}
    &\abs*{\int_y^x\frac{4L^2s}{1+2L^2s^2}\int_0^{\gamma(1+2L^2s^2)}\abs*{\phi^{(2)}(s')}ds'ds} \\
    \le& \abs*{\int_y^x\frac{4L^2s}{1+2L^2s^2}\int_0^{\gamma(1+2L^2s^2)}(W_2+c_2)(s')^{\alpha-2}ds'ds}\\
    =& 2(W_2+c_2)\gamma^{\alpha-1}\abs*{\int_y^x\frac{2L^2s(1+2L^2s^2)^{\alpha-2}}{\alpha-1}ds} \\
    \le& 2(W_2+c_2)\gamma^{\alpha-1}\abs*{\int_y^x\frac{2^{\alpha-1}L^{2\alpha-2}s^{2\alpha-3}}{\alpha-1}ds} \\
    \le& \frac{2^\alpha (W_2+c_2)}{(\alpha-1)(2\alpha-2)}L^{2\alpha-2}\abs*{x^{2\alpha-2}-y^{2\alpha-2}} \\
    \lesssim& \gamma^{\alpha-1}L^{2\alpha-2}\abs*{x-y}^{2\alpha-2}.
  \end{align}
  For the second term, we have
  \begin{align}
    &\abs*{\int_y^x\frac{4 L^2s}{\gamma(1+2L^2s^2)^2}\int_\gamma^{\gamma(1+2L^2s^2)}\int_0^{s'}\abs*{\phi^{(2)}(s'')}ds''ds'ds} \\
    \le& \begin{multlined}[t][.9\columnwidth]
      \abs[\Bigg]{\int_y^x\frac{4 L^2s}{\gamma(1+2L^2s^2)^2}\int_\gamma^{\gamma(1+2L^2s^2)}\\\times\int_0^{s'}(W_2+c_2)(s'')^{\alpha-2}ds''ds'ds}
    \end{multlined}\\
    \le& \begin{multlined}[t][.9\columnwidth]
     \frac{1}{\alpha(\alpha-1)}\abs[\Bigg]{\int_y^x\frac{4 L^2s}{\gamma(1+2L^2s^2)^2}\\\times(W_2+c_2)\gamma^\alpha\paren*{(1+2L^2s^2)^\alpha-1}ds} 
     \end{multlined}\\
    \le& \frac{2(W_2+c_2)\gamma^{\alpha-1}}{\alpha(\alpha-1)}\abs*{\int_y^x2L^2s(1+2L^2s^2)^{\alpha-2}ds} \\
    \lesssim& \gamma^{\alpha-1}L^{2\alpha-2}\abs*{x-y}^{2\alpha-2}.
  \end{align}
  Consequently, we have
  \begin{align}
    E_L(\phi^{\star}_\gamma,[-1,1]) \lesssim \omega_1(\phi^{\star}_\gamma,L^{-1}) \lesssim \gamma^{\alpha-1}. \label{eq:approx2-upper-modulus2}
  \end{align}

  From \cref{eq:lower-best-approx,eq:approx2-lower-modulus1,eq:approx2-upper-modulus1,eq:approx2-lower-modulus2,eq:approx2-upper-modulus2}, there are universal constants $C > 0$ and $C' > 0$ such that
  \begin{align}
    & E_L(\phi^\star,[\gamma,2L^2\gamma])\\
    \ge& C\gamma^{\alpha-1} - \frac{C'\gamma^{\alpha-1}}{c_\ell} \\
    \ge& \gamma^{\alpha-1}\paren*{C - \frac{C'}{c_\ell}}.
  \end{align}
  For sufficiently large $c_\ell$, there exists an universal constant $c > 0$ such that $C - \frac{C'}{c_\ell} > c$, which gives the desired result.
\end{proof}

\section{Helper Lemmas}
\begin{lemma}\label{lem:bound-sum-alpha}
  Given $\alpha \in [0,1]$, $\sup_{P \in \dom{M}_k} \sum_{i=1}^k p_i^\alpha = k^{1-\alpha}$.
\end{lemma}
\begin{proof}[Proof of \cref{lem:bound-sum-alpha}]
  If $\alpha = 1$, the claim is obviously true. Thus, we assume $\alpha < 1$. We introduce the Lagrange multiplier $\lambda$ for a constraint $\sum_{i=1}^n p_i = 1$, and let the partial derivative of $\sum_{i=1}^k p_i^\alpha + \lambda(1 - \sum_{i=1}^k p_i)$ with respect to $p_i$ be zero. Then, we have
  \begin{align}
   \alpha p_i^{\alpha-1} - \lambda = 0. \label{eq:sum-alpha-diff}
  \end{align}
  Since $p^{\alpha-1}$ is a monotone function, the solution of \cref{eq:sum-alpha-diff} is given as $p_i = (\lambda/\alpha)^{1/(\alpha-1)}$, i.e., the values of $p_1,...,p_k$ are equivalent. Thus, the function $\sum_{i=1}^k p_i^\alpha$ is maximized at $p_i = 1/k$ for $i=1,...,k$. Substituting $p_i=1/k$ into $\sum_{i=1}^k p_i^\alpha$ gives the claim.
\end{proof}

\begin{lemma}\label{lem:bound-sum-log}
  $\sup_{P \in \dom{M}_k} \sum_{i=1}^k p_i\ln^2p_i = \ln^2k$.
\end{lemma}
\begin{proof}[Proof of \cref{lem:bound-sum-log}]
We introduce the Lagrange multiplier $\lambda$ for a constraint $\sum_{i=1}^n p_i = 1$, and let the partial derivative of $\sum_{i=1}^k p_i\ln^2p_i + \lambda(1 - \sum_{i=1}^k p_i)$ with respect to $p_i$ be zero. Then, we have
  \begin{align}
   \ln^2p_i + 2\ln p_i - \lambda = 0. \label{eq:sum-log-diff}
  \end{align}
  From \cref{eq:sum-log-diff} and the fact that $\sum_{i=1}^kp_i = 1$, we have $p_i = \exp(\pm \sqrt{\lambda + 1})/e$ and $\lambda = \ln^2(e/k) - 1$. Hence, $p_i = 1/k$. Substituting this into $\sum_{i=1}^k p_i\ln^2p_i$ yields the claim.
\end{proof}
\end{document}